\documentclass{article}
\usepackage[T1]{fontenc}

\date{}
\usepackage{amsthm}
\usepackage{amsmath}
\usepackage{amssymb}
\usepackage{graphicx}
\usepackage{dsfont}

\usepackage{algorithm}
\usepackage{algorithmic}
\usepackage{fullpage}
\usepackage{multirow}

\usepackage{cleveref} 
\usepackage{complexity}

\usepackage{tcolorbox}
\newtcolorbox{wbox}
{
	colback  = white,
}

\usepackage{xcolor}
\definecolor{mygray}{RGB}{218,215,203}

\usepackage[round]{natbib}
\graphicspath{{./Figures/}}

\usepackage{todonotes}
\newcommand{\NN}{\ensuremath{\mathbb{N}}}%

\newcommand{\QSAT}{\textsc{2-Quantified 3-Dnf-Sat}}
\newcommand{\ASHG}{\textsc{Ashg-Popularity-Existence}}
\newcommand{\FHG}{\textsc{Fhg-Popularity-Existence}}

\newcommand{\XX}{\ensuremath{\mathcal{X}}}%
\newcommand{\YY}{\ensuremath{\mathcal{Y}}}%
\newcommand{\CC}{\ensuremath{\mathcal{C}}}%

\newcommand{\x}{\ensuremath{a_x}}%
\newcommand{\nx}{\ensuremath{a_{\neg x}}}%
\newcommand{\xl}{\ensuremath{a_{\alpha}}}%
\newcommand{\nxl}{\ensuremath{a_{\neg\alpha}}}%
\newcommand{\xt}{\ensuremath{x_t}}%
\newcommand{\xf}{\ensuremath{x_f}}%
\newcommand{\xfp}{\ensuremath{x_{f'}}}%
\newcommand{\Xt}{\ensuremath{X_t}}%
\newcommand{\Xf}{\ensuremath{X_f}}%
\newcommand{\Xfp}{\ensuremath{X_{f'}}}%

\newcommand{\yfix}{\ensuremath{a_y}}%
\newcommand{\nyfix}{\ensuremath{a_{\neg y}}}%
\newcommand{\y}{\ensuremath{a_{\beta}}}%
\newcommand{\ny}{\ensuremath{a_{\neg\beta}}}%
\newcommand{\yp}{\ensuremath{a_{\beta}'}}%
\newcommand{\ypp}{\ensuremath{a_{\beta}''}}%

\newcommand{\ga}{\ensuremath{a}}%
\newcommand{\gap}{\ensuremath{a'}}%
\newcommand{\ca}{\ensuremath{a_c}}%
\newcommand{\casec}{\ensuremath{a_{c'}}}%
\newcommand{\la}{\ensuremath{\ell_c}}%
\newcommand{\ra}{\ensuremath{r_c}}%

\newcommand{\Tru}{\textsc{True}}
\newcommand{\Fals}{\textsc{False}}

\newtheorem{theorem}{Theorem}[section]

\newtheorem{proposition}[theorem]{Proposition}
\newtheorem{lemma}[theorem]{Lemma}

\theoremstyle{definition}
\newtheorem{definition}[theorem]{Definition}

\usepackage{authblk}

\title{Settling the Complexity of Popularity in Additively Separable and Fractional Hedonic Games}

\author{Martin Bullinger}
\author{Matan Gilboa}
\affil{Department of Computer Science, University of Oxford, UK}

\begin{document}

\maketitle
\begin{abstract}
We study coalition formation in the framework of hedonic games.
There, a set of agents needs to be partitioned into disjoint coalitions, where agents have a preference order over coalitions. 
A partition is called popular if it does not lose a majority vote among the agents against any other partition.
Unfortunately, hedonic games need not admit popular partitions and prior work suggests significant computational hardness.
We confirm this impression by proving that deciding about the existence of popular partitions in additively separable and fractional hedonic games is $\Sigma_2^p$-complete.
This settles the complexity of these problems and is the first work that proves completeness of popularity for the second level of the polynomial hierarchy.
\end{abstract}

\section{Introduction}
We consider the task of partitioning a set of agents, say humans or machines, into disjoint coalitions.
The agents have preferences regarding the coalition they are part of and a reasonable partition should reflect these preferences.
This task is commonly studied in the framework of \emph{coalition formation} and is an intriguing object of study at the intersection of economics and computer science.
The typical economic setting is the formation of teams, such as working groups or political parties, but applications also consider reaching international agreements, establishing research collaboration, or forming customs unions \citep{Ray07a}.
Much related problems are community detection and clustering, two important tasks in social science and machine learning, respectively \citep{Newm04a,CLMP22a}.

The output of a coalition formation scenario is usually measured by means of solution concepts, which often incorporate stability or optimality.
While stability conceptualizes the prospect of agents staying in their own coalition rather than performing deviations to join other coalitions, optimality aims at global guarantees, for instance, with respect to notions of welfare.
We aim at finding an outcome that can be established as a status quo such that its legitimacy cannot be overthrown by another outcome when performing a majority vote against the status quo.
A solution concept that captures this idea is the notion of \emph{popularity} due to \citet{Gard75a}.\footnote{In his original work, \citet{Gard75a} calls popular outcomes ``majority assignments.''}
Informally speaking, an outcome is popular if no other outcome wins a vote against this outcome.
In social choice theory, this corresponds to the well-established notion of weak Condorcet winners \citep{Cond85a}, but popularity can be defined in any context where agents have preferences over outcomes.

\citet{Gard75a} was the first to consider popularity in a coalitional setting.
He considered bipartite matching instances and showed that stable matchings---in the sense of \citet{GaSh62a}---are popular if the agents' preferences are strict.
Interestingly, relaxing either assumption (bipartiteness or strict preferences) may lead to instances in which popular outcomes do not exist, and the corresponding decision problems become \NP-complete \citep{BIM10a,FKPZ19a,GMSZ21a}.
Notably, membership in \NP{} is not trivial for this problem because one has to certify that a matching does not lose a vote against any other matching, of which there are exponentially many.
This task can, however, be performed by transforming the verification of popular matchings to a maximum weight matching problem \citep{BIM10a} or a linear program that can simultaneously handle weak and incomplete preferences as well as nonbipartite instances \citep{KMN11a,BrBu20a}.

When allowing coalitions to be of size greater than~$2$, we reach the classical domain of coalition formation.
We consider the prominent classes of additively separable and fractional hedonic games \citep{BoJa02a,ABB+17a}.
Specifically, we study the following decision problem.
	\begin{wbox}
		\ASHG{} (\FHG)\\
		\textbf{Input:} Additively separable hedonic game (fractional hedonic game)\\
		\textbf{Question:} Does the given game admit a popular partition?
	\end{wbox}

Previous to our work, \citet{ABS11c} and \citet{BrBu20a} already presented evidence that the complexity of popularity increases further in these classes of games:
the verification of popular coalition structures becomes \coNP-complete and the existence problem is both \NP-hard and \coNP-hard.
In addition, the definition of popularity (\emph{existence} of an outcome such that \emph{for all} other outcomes, a vote is not lost) places popularity inside the complexity class $\Sigma_2^p$.
In our work, we complete the complexity picture for additively separable and fractional hedonic games by proving the following two theorems.

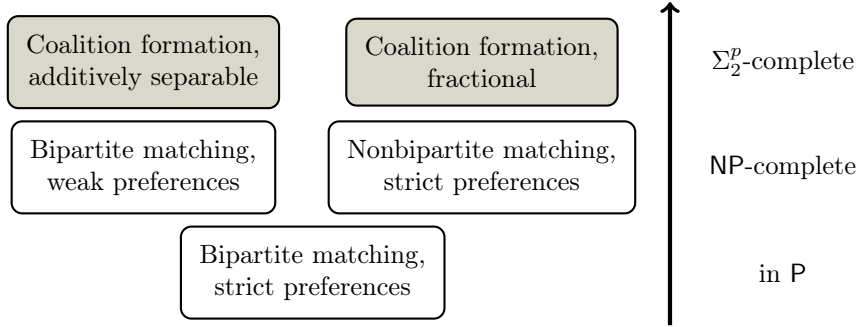
\begin{figure}[t]
    \centering
    \begin{tikzpicture}
        \draw[->,ultra thick] (0,-.1) -- (0,4.2);
        
        \node at (1.5,.6) {in \P};
        \node at (1.5,2) {\NP-complete};
        \node at (1.5,3.4) {$\Sigma_2^p$-complete};

        \node[align = center, draw, thick, rounded corners,inner sep = .7em] at (-4.75,.6) {Bipartite matching,\\ strict preferences};
        \node[align = center, draw, thick, rounded corners,inner sep = .7em] at (-7,2) {Bipartite matching,\\ weak preferences};
        \node[align = center, draw, thick, rounded corners,inner sep = .7em] at (-2.5,2) {Nonbipartite matching,\\ strict preferences};
        \node[align = center, draw, thick, rounded corners,inner sep = .7em, fill = mygray] at (-7,3.4) {Coalition formation,\\ additively separable};
        \node[align = center, draw, thick, rounded corners,inner sep = .7em, fill = mygray] at (-2.5,3.4) {Coalition formation,\\ fractional};
        
    \end{tikzpicture}
    \caption{Complexity hierarchy of popularity in coalitional scenarios.
    Gray boxes refer to our main results.}
    \label{fig:complexity-overview}
\end{figure}

\begin{theorem}
\label{thm_ashg}
\ASHG{} is $\Sigma_2^p$-complete.
\end{theorem}

\begin{theorem}
\label{thm_fhg}
\FHG{} is $\Sigma_2^p$-complete, even if valuations are nonnegative.
\end{theorem}

This completes the characterization of the complexity hierarchy of popularity in coalitional scenarios, as detailed in \Cref{fig:complexity-overview}. These results highlight the significant computational hardness presented by these problems, beyond the previously known results. 
While many NP-hard problems can be addressed in practice using SAT or ILP solvers, our $\Sigma_2^p$-completeness results indicate a higher level of complexity that surpasses these typical approaches.

Note that the nonnegativity assumption in \Cref{thm_fhg} is a strong additional restriction, which is not possible for \Cref{thm_ashg}.
Indeed, additively separable hedonic games define agents' utilities for coalitions based on the sum of valuations of its members.
Hence, forming the grand coalition containing all agents is optimal for all agents if valuations are nonnegative.
By contrast, the sum of valuations is divided by the size of the coalition in fractional hedonic games, which leads to nontrivial preferences, even for nonnegative valuations.

\section{Related Work}

Coalition formation in the framework of hedonic games was first considered by \citet{DrGr80a} and further conceptualized by \citet{BoJa02a}, \citet{BKS01a}, and \citet{CeRo01a}.
An overview on hedonic games can be found in the book chapters by \citet{AzSa15a} and \citet{BER24a}.

In a general model of hedonic games, agents have to rank an exponentially large set of possible coalitions.
Since this causes computational issues, a wide range of succinct preference representations has been proposed in the literature.
Often, this is based on restricting attention to important meta-information about a coalition such as its size \citep{BoJa02a} or its best or worst member \citep{CeRo01a}.
Another way is to aggregate cardinal valuation functions of single agents to a utility for a coalition.
We consider models that follow this latter approach, namely additively separable hedonic games (ASHGs) and fractional hedonic games (FHGs) \citep{BoJa02a,ABB+17a}.

Similar to the landscape of game classes, there exists a variety of solution concepts for hedonic games.
We focus our discussion on the  large body of research on ASHGs and FHGs. 
Much of this literature concerns stability, i.e., the absence of beneficial deviations to join other or form new coalitions.
A common theme is that stability is usually only satisfiable for restricted domains of games, and various computational hardness results have been observed.
Interestingly, there is a difference in complexity dependent on whether single agents or groups of agents perform a deviation.
Whether a single agent can perform a deviation can usually be checked in polynomial time and we obtain \NP-completeness results \citep{SuDi10a,ABS11c,BBS14a,BBW21b,BBT23a}.
By contrast, whether a group deviation exists is itself \NP-complete to check and hence the existence of group stability, e.g., whether there exist partitions in the \emph{core}, becomes $\Sigma_2^p$-complete \citep{Woeg13a,Pete17b,ABB+17a}.
Prior to our complexity results on popularity, these were the only problems known to be $\Sigma_2^p$-complete for hedonic games.

It is possible to achieve more positive results regarding the existence of stable outcomes by considering restricted domains \citep{BoJa02a,DBHS06a}, weakened solution concepts \citep{FMM21a}, stability under randomized deviations \citep{FFKV23a}, or in random games \citep{BuKr24a}.
For instance, symmetric utilities lead to the existence of single-deviation stability in ASHGs \citep{BoJa02a}, but the same is not true in FHGs \citep{BBW21b}, and even in ASHGs, computation is still \PLS-hard \citep{GaSa19a}.

Popularity, our main solution concept, has received less attention.
Most related to our work is the paper by \citet{BrBu20a} who prove \NP-hardness and \coNP-hardness of the existence problem for ASHGs and FHGs.
In addition, they show \coNP-completeness of the verification of popular partitions, a problem that was also considered by \citet{ABS11c} for ASHGs.\footnote{\citet{ABS11c} also consider the existence problem of popularity for ASHGs, but their proof was pointed out to be incomplete by \citet{BrBu20a}.}
Our results improve upon these results by showing $\Sigma_2^p$-completeness of the existence problem, which settles the precise complexity of popularity in ASHGs and FHGs.

Popularity has also been considered in further classes of hedonic games.
\citet{BrBu20a} and \citet{CsPe21a} study it for games with coalitions bounded in size by three, and \citet{KLR+20a} consider a preference model based on the distinction of friends, enemies, and neutrals.
Moreover, \citet{KeRo20a} consider popularity for a nonhedonic class of coalition formation games aimed at modeling altruism. 
All of these papers show \coNP-completeness of the verification problem.
However, while \citet{BrBu20a} and \citet{CsPe21a} at least show \NP-hardness, the complexity of the existence problem remains unresolved in all of these models.

\section{Preliminaries}

In this section, we provide the preliminaries for our work.
We start with defining hedonic games, then define important solution concepts, and finally discuss the computational aspects of these solution concepts.

\subsection{Succinct Classes of Cardinal Hedonic Games}
Let $N$ be a set of agents. 
A \textit{coalition} is a nonempty subset of $N$. 
A coalition of size one is called a \textit{singleton} coalition. 
Denote by $\mathcal{N}_i=\{S\subseteq N\colon i\in S\}$ the set of all coalitions agent $i$ belongs to. 
A \textit{coalition structure}, or a \textit{partition}, is a partition $\pi$ of $N$ into coalitions. 
For an agent $i\in N$, we denote by $\pi(i)$ the coalition $i$ belongs to in $\pi$. 

A \textit{hedonic game} is a pair $(N,\succsim)$, where $\succsim=(\succsim_i)_{i\in N}$ is a preference profile specifying the preferences of each agent $i$ as a complete and transitive preference order $\succsim_i$ over $\mathcal{N}_i$. 
In hedonic games, agents are only concerned with the members of their own coalition which is also reflected in their preference order.
Therefore, we can naturally define an associated preference order over partitions by $\pi\succsim_i\pi'$ if and only if $\pi(i)\succsim_i\pi'(i)$.
For coalitions $S,S'\in \mathcal{N}_i$, we say that agent $i$ \textit{weakly prefers} $S$ over $S'$ if $S\succsim_i S'$.
Moreover, we say that $i$ \textit{prefers} $S$ over $S'$ if $S\succ_i S'$.
We use the same terminology for preferences over partitions.

In this paper, we assume agents rank coalitions (and by extension, partitions) by underlying utility functions $u = (u_i\colon \mathcal{N}_i\to\mathbb{R})_{i\in N}$.
These induce a hedonic game $(N,\succsim)$ where, for every agent $i\in N$ and two coalitions $S, S'\in \mathcal N_i$, we define $S\succsim_i S'$ if and only if $u_i(S)\ge u_i(S')$.
Hence, $i$ prefers $S$ over $S'$ if and only if $u_i(S) > u_i(S')$.
We say that $u_i(S)$ is $i$'s utility for coalition $S$ and extend this to utilities for partitions by setting $u_i(\pi) = u_i(\pi(i))$.
A hedonic game together with its utility-based representation is called a \textit{cardinal hedonic game} and is specified by the pair $(N,u)$.

Hedonic games as introduced so far need every agent to specify a preference order or cardinal values for an exponentially large set of coalitions.
By contrast, we focus on succinctly representable sub-classes of cardinal hedonic games, where the utilities are induced by the aggregation of values that each agent assigns to other members of her coalition. 
These games are specified by a pair $(N,v)$, where $v = (v_i\colon N\to \mathbb R)_{i\in N}$ is a vector of \textit{valuation functions}. 
The quantity $v_i(j)$ denotes the value agent $i$ assigns to agent $j$.

Following \citet{BoJa02a}, an \textit{additively separable hedonic game} (ASHG) given by the pair $(N,v)$ is the cardinal hedonic game $(N,u)$ where 
\[u_i(S)=\sum_{j\in S\setminus \{i\}}v_i(j)\text.\]
Hence, the utility $u_i(S)$ of agent $i$ for coalition $S\in \mathcal N_i$ is defined as the sum of the values agent $i$ assigns to the other members of her coalition.

Following \citet{ABB+17a}, a \textit{fractional hedonic game} (FHG) given by the pair $(N,v)$ is the cardinal hedonic game $(N,u)$ where 
\[u_i(S)=\frac{\sum_{j\in S\setminus \{i\}}v_i(j)}{|S|}\text.\]
Hence, the utility $u_i(S)$ of agent $i$ for coalition $S\in \mathcal N_i$ is defined as the sum of the values agent $i$ assigns to the other members of her coalition divided by the coalition size. 
This quantity can be interpreted as the average value that $i$ assigns to the members of her coalition if we include a value of $0$ for herself.

\subsection{Popular Partitions}
We now move towards defining popularity, our main solution concept, for a given hedonic game $(N,\succsim)$. 
Let $\pi$ and $\pi'$ be two partitions of $N$. 
We denote the set of agents who prefer $\pi$ over $\pi'$ by $N(\pi,\pi')$, i.e., $N(\pi,\pi')=\{i\in N\colon\pi\succ_i\pi'\}$. 
For any subset of agents $M\subseteq N$, we define the \textit{popularity margin} on $M$ with respect to the ordered pair $(\pi,\pi')$ to be $\phi_M(\pi,\pi')=|N(\pi,\pi')\cap M|-|N(\pi',\pi)\cap M|$.
Note that in this definition, agents who are indifferent between the two partitions do not contribute to any of the two terms. 
When $M$ is a singleton containing a single agent $\ga$, we use the abbreviated notation $\phi_{\ga}(\pi,\pi')=\phi_{\{\ga\}}(\pi,\pi')$. 
The definition of popularity margins is useful as sometimes it is convenient to consider restricted subsets of agents separately. 
Further, considering the entire set of agents, we define the \textit{popularity margin} of the ordered pair $(\pi,\pi')$ as $\phi(\pi,\pi')=\phi_N(\pi,\pi')$.
Note that the popularity margin is antisymmetric, i.e., $\phi(\pi,\pi') = -\phi(\pi',\pi)$.
We say that $\pi$ is \textit{more popular} than $\pi'$ if $\phi(\pi,\pi')>0$. 
Moreover, $\pi$ is called \textit{popular} if there exists no partition $\pi'$ that is more popular than $\pi$, i.e., for any partition $\pi'$ it holds that $\phi(\pi,\pi')\geq 0$.

Another useful concept in the context of popularity is Pareto optimality.
We say that $\pi'$ is a \textit{Pareto improvement} from $\pi$ if all agents weakly prefer $\pi'$ over $\pi$, and at least one agent strictly prefers $\pi'$ over $\pi$. 
If there exists no Pareto improvement from $\pi$, we say $\pi$ is \textit{Pareto-optimal}. 
Clearly, popular partitions are Pareto-optimal.
Indeed, every Pareto improvement is a more popular partition.
By contrast, Pareto-optimal partitions need not be popular.
In addition, a useful observation is that it suffices to restrict attention to Pareto-optimal partitions when considering popularity \citep{BrBu20a}.

\begin{proposition}[\citealp{BrBu20a}, Proposition~4]\label{prop:wlogPO}
    A partition $\pi$ is popular if and only if for all Pareto-optimal partitions $\pi'$ it holds that $\phi(\pi,\pi')\geq 0$.
\end{proposition}

As a consequence, whenever we postulate a more popular partition than a given partition, we may assume without loss of generality that this partition is Pareto-optimal.

\subsection{Complexity Theory}
We assume familiarity of the reader with basic notions of complexity theory such as polynomial-time reductions or the classes \P{} (\emph{deterministic polynomial time}) and \NP{} (\emph{nondeterministic polynomial time}).
Here, we focus on the complexity class $\Sigma_2^p$ in the second level of the polynomial hierarchy, which captures the problems considered in this paper.
We refer to the textbooks by \citet{Papa94a} and \citet{ArBa09a} for an introduction to complexity and a more in-depth covering of $\Sigma_2^p$.

The class $\Sigma_2^p$ contains all problems $P$ for which there exists a polynomial-time Turing machine $M$ and a polynomial $q$ such that $x$ is a Yes-instance of $P$ if and only if there exists a $y\in \{0,1\}^{q(|x|)}$ such that for all $z\in \{0,1\}^{q(|x|)}$ it holds that 
$M(x, y, z) = \Tru$.
Informally speaking, this captures problems in which the solutions $y$ of an instance $x$ are challenged by any possible adversary $z$.
The class is thus described by the concatenation of an existential and a universal quantifier.
It therefore contains \NP{}, which is defined by just an existential quantifier (because we can ignore the universal quantifier), and \coNP{}, which is defined by just a universal quantifier (because we can ignore the existential quantifier).
As with other complexity classes, a problem $P$ is said to be $\Sigma_2^p$-\emph{hard} if for every problem in $\Sigma_2^p$, there exists a polynomial-time reduction from this problem to $P$.
A problem is said to be $\Sigma_2^p$-\emph{complete} if it is $\Sigma_2^p$-hard and contained in $\Sigma_2^p$.

As a first example, we define the problem \QSAT, which is a canonical \textsc{Sat} problem for $\Sigma_2^p$.
It is the source problem of our reductions in \Cref{thm_ashg,thm_fhg}. 
	\begin{wbox}
		\QSAT\\
		\textbf{Input:} Two sets $\XX=\{x_1, \dots,x_n\}$ and $\YY=\{y_1, \dots,y_n\}$ of Boolean variables and a Boolean formula $\psi(\XX,\YY)$ over $\XX\cup \YY$ in disjunctive normal form, where each of the conjunctive clauses consists of exactly three distinct literals.\\
		\textbf{Question:} Does there exist a truth assignment $\tau_{\XX}$ to $ x_1, \dots,x_n$ such that for all truth assignments $\tau_{\YY}$ to $y_1, \dots,y_n$ it holds that $\psi(\tau_{\XX},\tau_{\YY})=\Tru$?
	\end{wbox}

\QSAT{} is exactly in the spirit of $\Sigma_2^p$.
Yes-instances are described by the existence of a certificate (the truth assignment to $x_1, \dots,x_n$) such that the output of the formula is \Tru{} regardless of the truth assignment to $y_1, \dots,y_n$.
Even more, \QSAT{} was shown to be $\Sigma_2^p$-complete by \citet{St77}.

As a second example, we argue that \ASHG{} and \FHG{} are contained in $\Sigma_2^p$, as remarked by \citet{BrBu20a}:
One can consider a polynomial-time Turing machine with three inputs that are a hedonic game (say, an ASHG or FHG) and two partitions $\pi$ and $\pi'$ and it outputs \Tru{} if and only if $\phi(\pi,\pi')\ge 0$ in the given hedonic game.
This Turing machine attests membership in $\Sigma_2^p$ of the existence problem of popularity.
In our proofs, we will therefore only consider hardness.

\section{Popularity in ASHGs}
\label{sec_ashg}

In this section, we discuss the proof of \Cref{thm_ashg}.
We start by describing our reduction from \QSAT.
Then, in the subsequent two sections, we give an overview of the proof that satisfiability of the source instance implies the existence of a popular partition and vice versa.
We focus on the key arguments and an illustration while the detailed proof is deferred to \Cref{apndx_ashg}.

\subsection{Setup of the Reduction}
\label{sec_ashg:setup}
We now describe the construction of the reduction.
First, we introduce the following No-instance of \ASHG, on which the reduction relies; it is a symmetric version of the No-instance described by \citet[Example~4]{ABS11c}.
Suppose we have five agents, consisting of three \emph{top} agents $t_1$, $t_2$, and $t_3$, and two \emph{bottom} agents $b_1$ and $b_2$. 
For each $i\in\{1,2,3\}$, $t_i$ assigns value $1$ to $b_1$ and value $2$ to $b_2$. 
Moreover, $b_1$ assigns value $1$ to each top agent and $b_2$ assigns value $2$ to each top agent. 
All other values are set to $-\infty$ (or some large negative value, e.g., $-7$ suffices here) between the agents.
This instance is depicted in \Cref{fig_ashg:No_instance}.

\begin{figure}[t]
\centering
\includegraphics[width=.4\textwidth]{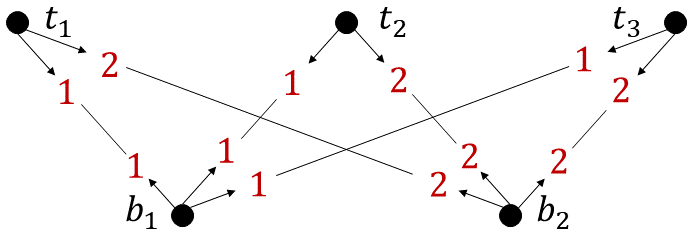}
	\caption{A No-instance of \ASHG. Omitted edges imply value $-\infty$.}
	\label{fig_ashg:No_instance}
\end{figure}

One may verify that there exists no popular partition in this instance:
It is easy to see that it is more popular to dissolve any coalition of size at least three into singletons.
Hence, the interesting case is a partition of the type $\{\{t_1,b_1\},\{t_2,b_2\},\{t_3\}\}$, which, however, is less popular than $\{\{t_1,b_2\},\{t_3,b_1\},\{t_2\}\}$.

In our reduction, we construct a game which has a similar structure to this No-instance (with some additional agents).
However, each top agent $t_i$ is replaced by a set of multiple agents who, intuitively, together function in a similar way as the single agent $t_i$. 
Hence, familiarity with the above No-instance is helpful to understand the reduction as well: when a satisfying assignment to the \QSAT{} instance does not exist, the reduced game simulates a behaviour similar to that of this No-instance.

We proceed by describing our reduction.
Suppose that we are given an arbitrary instance $(\XX,\YY,\psi)$ of \QSAT.
Denote by $\CC$ the set of clauses in $\psi$ and let $m=|\CC|$; without loss of generality, we may assume that $m\geq 2$. 
We construct the following ASHG consisting of $12n+4m-1$ agents, depicted in \Cref{fig_ashg:reduction}.

\begin{figure}[t]
\centering
\includegraphics[width=.8\textwidth]{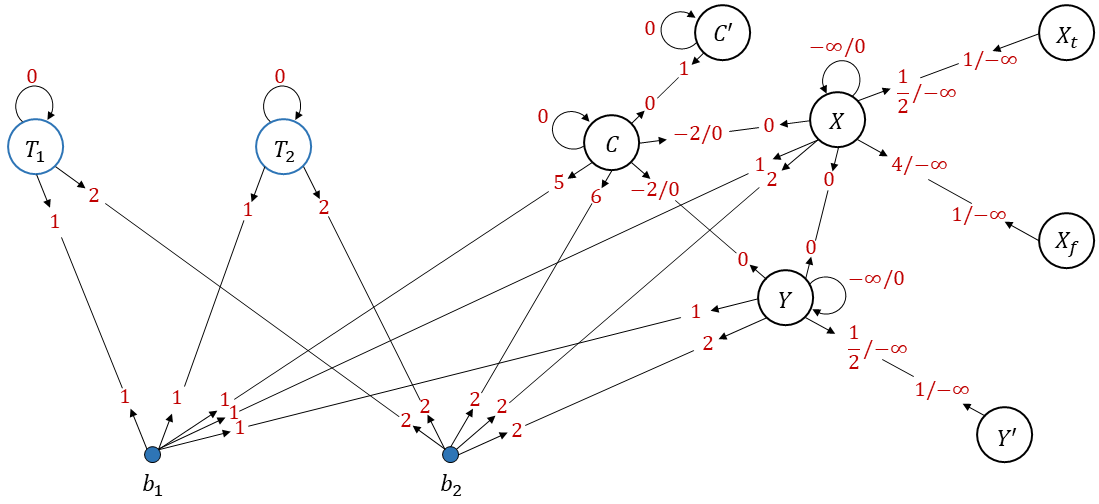}
	\caption{The reduction for the proof of \Cref{thm_ashg}. Omitted edges imply value $-\infty$. When two values $v_1/v_2$ appear, $v_1$ refers to corresponding agents, and $v_2$ to noncorresponding. Left-side agents are marked in blue. $b_1$ and $b_2$ are single agents, while the rest represent sets of agents.}
	\label{fig_ashg:reduction}
\end{figure}

\begin{itemize}
    \item For every variable $x\in \XX$:
        \begin{itemize}
            \item We create two $X$-agents $\x$ and $\nx$, where the former represents the variable and the latter its negation. 
            We will use $\alpha$ to denote any literal over $\XX$, meaning $\xl$ can correspond to a variable or its negation; accordingly, $\nxl$ will simply correspond to the negated literal, e.g., if $\alpha = \lnot x$, then $\nxl = \x$.
            If $\x$ and $\nx$ originate in the same variable, they are called \emph{complementary} agents. 
            \item We create a corresponding $\Xt$-agent and a corresponding $\Xf$-agent, denoted $\xt$ and $\xf$, respectively.
            The subscripts of these agents indicate ``true'' and ``false'' and these agents are used to deduct the satisfying truth assignments from popular partitions (and vice versa).
        \end{itemize}
    \item For every variable $y\in \YY$:
        \begin{itemize}
            \item We create two $Y$-agents $\yfix$ and $\nyfix$, where the former represents the variable and the latter its negation. 
            We will use $\beta$ to denote any literal over $\YY$, meaning $\y$ can correspond to a variable or its negation; $\ny$ will refer to the agent corresponding to the negated literal. 
            If $\y$ and $\ny$ originate in the same variable, they are called \emph{complementary} agents.
            \item We create a $Y'$-agent $a_{y}'$ corresponding to $\yfix$, and a $Y'$-agent $a_{\neg y}'$ corresponding to $\nyfix$ (we emphasize that, in contrast to the $X$-agents, which have corresponding agents as a pair, $\yfix$ and $\nyfix$ each have separate $Y'$-agents).
        \end{itemize}
    \item For every clause $c\in \CC$, we create a $C$-agent $\ca$.
    For a literal $\alpha$ over $\XX$ (or $\beta$ over $\YY$) occurring in $c$, we refer to the $X$-agent $\xl$ (or $Y$-agent $\y$) as corresponding to clause $c$.
    \item We create $m-1$ agents, called $C'$-agents.
    \item We create $2n+m$ agents, called $T_1$ agents and another $2n+m$ agents, called $T_2$ agents.
    \item We create a single agent denoted $b_1$, and a single agent denoted $b_2$. 
\end{itemize}

For each agent type, the set of all agents from that type is denoted by the name of the type (e.g., $T_1$ is the set of all $T_1$-agents).
We use the terms \emph{real agents} to refer to the $X$-, $Y$-, and $C$-agents,
and \emph{structure agents} to refer to all other agents.
In addition, we speak of \emph{left-side agents} to refer to $b_1$, $b_2$, and the $T_1$- and $T_2$-agents, and \emph{right-side agents} to refer to the other agents (this terminology is based on the visualization in \Cref{fig_ashg:reduction}).
We denote by $L$ and $R$ the sets of all left-side and right-side agents, respectively. 

We refer to \Cref{fig_ashg:reduction} for an overview of the valuation functions and to \Cref{apndx_ashg} for a detailed description.
Valuations missing from the figure (as well as some of the depicted ones) correspond to a large negative constant which we indicate by a value of $-\infty$.
For the reduction to work, one can, for instance, set $\infty = 6(12n+4m-1)$.
This completes the description of the reduction. 

When the input \QSAT{} instance is a No-instance, the reduced ASHG mimics the No-instance in \Cref{fig_ashg:No_instance}, where $b_1$ and $b_2$ are still single agents, but $t_1$ and $t_2$ are replaced by the sets $T_1$ and $T_2$, and the real agents correspond to the agent $t_3$. 
The real agents also encode the source instance of \QSAT, as they are representatives of the literals and clauses.
The right-side structure agents provide options for ``good'' coalitions for the real agents.

In essence, a popular partition can only exist if all right-side agents are in good coalition which will allow for a partition corresponding to the partition $\{\{t_1,b_1\},\{t_2,b_2\},\{t_3\}\}$ to be popular.
The good coalitions for the real agents are:
\begin{itemize}
    \item coalitions of the type $\{\x,\xt\}$ and $\{\nx,\xf\}$ or $\{\nx,\xt\}$ and $\{\x,\xf\}$ for the $X$-agents,
    \item coalitions $\{\y,\yp\}$ for the $Y$-agents, and
    \item coalition $C\cup C'$ for the $C$-agents.
\end{itemize}

The crucial part is to determine the exact coalitions of the $X$-agents.
Whether we form $\{\x,\xt\}$ or $\{\nx,\xt\}$ corresponds to a truth assignment to the $\XX$ variables.

To prove \Cref{thm_ashg}, we will show that the logical formula is satisfiable if and only if there exists a popular partition in the constructed ASHG. 
If we have a truth assignment, we can define a partition as described above and prove that it is popular.
Conversely, a popular partition has to be a structure similar to the partition described above and we can use it to extract a truth assignment.
The two directions of the proof will be sketched in \Cref{sec_ashg:if_satisfiable,sec_ashg:if_popular}.

\subsection{Satisfiability Implies Popular Partition}
\label{sec_ashg:if_satisfiable}
Throughout this section, we assume that $(\XX,\YY,\psi)$ is a Yes-instance of \QSAT.
Hence, there is a truth assignment $\tau_{\XX}$ to the variables in $\XX$ such that for all truth assignments $\tau_{\YY}$ to the variables in $\YY$ it holds that $\psi(\tau_{\XX},\tau_{\YY}) = \Tru$. 
Consider the following partition of the agents, denoted by $\pi^*$.
\begin{itemize}
    \item For each $x\in\XX$, if $x$ is assigned \Tru{} by $\tau_{\XX}$ then $\{\{\x,\xt\},\{\nx,\xf\}\}\subseteq\pi^*$, and if $x$ is assigned \Fals{} by $\tau_{\XX}$ then $\{\{\nx,\xt\},\{\x,\xf\}\}\subseteq\pi^*$.
    \item Each $Y$-agent $\y$ forms a coalition with her corresponding $Y'$-agent $\yp$.
    \item The coalition $C\cup C'$ is formed.
    \item Coalitions $T_1\cup\{b_1\}$ and $T_2\cup\{b_2\}$ are formed.
\end{itemize}
Our goal is to show that $\pi^*$ is a popular partition. 
We defer the detailed proof of this statement to \Cref{apndx_ashg:if_satisfiable} and focus on outlining key steps.
Assume towards contradiction that there exists a partition $\pi$ that is more popular than $\pi^*$.
We wish to use $\pi$ to extract a truth assignment $\tau_{\YY}$ to the variables in $\YY$ such that $\psi(\tau_{\XX},\tau_{\YY}) = \Fals$.
For this, we will determine various structural insights about the partition $\pi$ and finally use the coalition $\pi(b_1)$ to find both the assignment $\tau_{\YY}$ as well as a proof that it can be used to evaluate $\psi$ as false.

For determining the structure of $\pi$, it is good to first consider the popularity margin for certain groups of agents.
By using that $\pi^*$ is a very good partition for the $Y'$-, $X_f$-, $X_t$-, and $C'$-agents, we obtain the following facts
\begin{itemize}
    \item For each $\y\in Y$, it holds that $\phi_{\{\y,\yp\}}(\pi^*,\pi)\geq 0$.
    \item For each $x\in\XX$, it holds that $\phi_{\{\x,\nx,\xt,\xf\}}(\pi^*,\pi)\geq 0$.
    \item If for every $\ca\in C$ we have that $\phi_{\ca}(\pi,\pi^*)>0$, then $\phi_{C\cup C'}(\pi^*,\pi)=-1$. 
Otherwise, $\phi_{C\cup C'}(\pi^*,\pi)\geq 0$.
\end{itemize}
Together, the worst-case popularity margin of the right-side agents is thus $\phi_{R}(\pi^*,\pi)\geq -1$.

As a next step, we consider coalitions of left-side agents and show that
\begin{enumerate}
    \item Agents in $T_1$ and $T_2$ cannot be in the same coalition.\label{summary_ashg1:item_TT}
    \item The coalition of $b_1$ contains a right-side agent.\label{summary_ashg1:item_b1}
    \item The coalition of $b_2$ does not contain a right-side agent.\label{summary_ashg1:item_b2}
\end{enumerate}

\Cref{summary_ashg1:item_TT} holds because these agents only gain positive value from $b_1$ and $b_2$, whereas valuations between agents in $T_1$ and $T_2$ are $-\infty$.
This insight can then be leveraged to show that at least one agent of $\{b_1,b_2\}$ has to contain a right-side agent.
Otherwise, it is easy to deduce that the left-side agents have a popularity margin of $\phi_L(\pi^*,\pi)\ge 0$, and furthermore no $C$-agent can gain positive utility in $\pi$, and therefore also $\phi_R(\pi^*,\pi)\ge 0$.
Together, these two facts imply that $\pi$ was not more popular. \Cref{summary_ashg1:item_b1,summary_ashg1:item_b2} follow with little effort from this conclusion.

We can now show that $\phi_{T_1\cup T_2\cup \{b_2\}}(\pi^*,\pi)\geq 0$.
Together with our other insights about the popularity margins, $\pi$ can only be more popular than $\pi^*$ if $\phi_{b_1}(\pi^*,\pi)\leq 0$.

Next, it is easy to see that each agent in $T_1$ or $T_2$ that forms a coalition with a right-side agent would have to be in the coalition with $b_1$.
However, by carefully analysing $\pi(b_1)$, it can then be shown that it cannot contain agents in $T_1$ and $T_2$.

To summarize our knowledge about left-side agents, we know that $b_1$ forms a coalition with right-side agents only, whereas all other left-side agents form coalitions with other left-side agents.

The next step is to analyze the exact coalition of $b_1$ in $\pi$.
It can be shown that $\pi(b_1)$ can only contain real agents (recall that $\phi_{b_1}(\pi^*,\pi)\leq 0$) and that it has to contain exactly $n$ $X$-agents corresponding to the agents forming coalitions with the $X_f$-agents in $\pi^*$, all $C$-agents, and either $\y$ or $\ny$ for every $\YY$ variable.

We can now extract a truth assignment $\tau_{\YY}$ to $\YY$ from the $Y$-agents contained in $\pi(b_1)$.
The only way that $\pi$ is more popular than $\pi^*$ is when all $C$-agents prefer $\pi$ over $\pi^*$ which, due to the valuations by the $C$-agents of the agents corresponding to their respective literals, can only happen if $\tau_{\XX}$ and $\tau_{\YY}$ evaluate every clause to \Fals.
This implies that $\psi(\tau_{\XX},\tau_{\YY}) = \Fals$, a contradiction. We thus conclude this part of the proof.

\subsection{Popular Partition Implies Satisfiability}
\label{sec_ashg:if_popular}
Throughout this section, we assume that there is a popular partition $\pi^*$ in the reduced ASHG.
We will prove that this implies that the source instance is a Yes-instance to \QSAT. We defer the detailed proof of this statement to  \Cref{apndx_ashg:if_popular}.
In this section, we give an overview of the proof.

Our main goal is to show that $\pi^*$ has a structure similar to that of the popular partition defined in \Cref{sec_ashg:if_satisfiable} (up to symmetries), which will enable us to extract a satisfying truth assignment to the variables in $\XX$ by looking at the coalitions of the $X_t$-agents.

As a first step, we show that left-side and right-side agents cannot form a joint coalition.
Suppose a coalition $S\in \pi^*$ contains both a left-side and a right-side agent. 
The only agents who may have a nonnegative utility in such a coalition are $b_1$, $b_2$, and real agents, and thus $S$ must contain some combination of agents $b_1$ and $b_2$. 
If both $b_1$ and $b_2$ are in $S$, then the partition obtained from $\pi^*$ by extracting $b_1$ and $b_2$ from $S$, and forming the coalitions $\{b_1\}\cup T_1$ and $\{b_2\}\cup T_2$, can be shown to be more popular. 
So, only one of $b_1$ and $b_2$ may reside in $S$. 
Denote $b_j\in S$, and $b_i\notin S$, where $i,j\in\{1,2\}$. Hence, it is easy to see that we must have either $\pi^*(b_i)=\{b_i\}\cup T_1$ or $\pi^*(b_i)=\{b_i\}\cup T_2$. Without loss of generality, assume $\pi^*(b_i)=\{b_i\}\cup T_1$. Now, intuitively, we can think of $T_1$, $T_2$, and $S\setminus\{b_j\}$ as the agents $t_1$, $t_2$, and $t_3$ from the No-instance discussed in \Cref{sec_ashg:setup}, respectively. 
A deviation analogous to that discussed in the context of this No-instance shows that this partition is not popular.

Having established that the left and right side are separated, the only possibility for $\pi^*$ to be popular is if agents form coalitions with their corresponding agents, who give them positive utility. 
Specifically, the following must hold.
\begin{enumerate}
    \item For the left side, we have that\label{item_ashg2:left_side} 
	$\{b_1\}\cup T_1\in\pi^*$ and $\{b_2\}\cup T_2\in\pi^*$, or
    $\{b_2\}\cup T_1\in\pi^*$ and $\{b_1\}\cup T_2\in\pi^*$.
    \item We have that $C\cup C'\in\pi^*$.\label{item_ashg2:C}
    \item For each $\y\in Y$, we have that $\{\y,\yp\}\in\pi^*$.\label{item_ashg2:Y}
    \item For each $x\in\XX$, we have that\label{item_ashg2:X} $\{\x,\xt\}\in\pi^*$ and $\{\nx,\xf\}\in\pi^*$, or $\{\x,\xf\}\in\pi^*$ and $\{\nx,\xt\}\in\pi^*$.
\end{enumerate}

This allows us to define the following truth assignment $\tau_{\XX}$ to the $\XX$ variables. 
For each $x\in \XX$, $x$ is assigned \Tru{} if and only if $\pi^*(\x)=\{\x,\xt\}$ (by \Cref{item_ashg2:X}, this is a valid assignment). 
We claim that $\tau_{\XX}$ is a satisfying assignment to the $\QSAT$ instance, i.e., that $\psi(\tau_{\XX},\tau_{\YY}) = \Tru$ for all truth assignments $\tau_{\YY}$ to the $\YY$ variables.

Assume otherwise, namely that there exists a truth assignment $\tau_{\YY}$ to the $\YY$ variables such that $\psi(\tau_{\XX},\tau_{\YY})=\Fals$. 
We will now find a partition that is more popular than $\pi^*$. Recalling \Cref{item_ashg2:left_side}, let us assume without loss of generality that $\{\{b_1\}\cup T_1,\{b_2\}\cup T_2\}\subseteq\pi^*$.
Consider the partition $\pi$ obtained from $\pi^*$ as follows.
\begin{itemize}
    \item Extract the following agents from their respective coalitions, and form a new coalition $S$:
    \begin{itemize}
        \item All $\xl\in X$ such that $\{\xl,\xt\}\in \pi^*$, for some $\xt\in\Xt$.
        \item All $\y\in Y$ such that the literal represented by $\y$ is assigned \Tru{} by $\tau_{\YY}$.
        \item All $C$-agents and agent $b_1$.
    \end{itemize}
    \item Extract $b_2$ from her coalition, and set $\pi(b_2)=\{b_2\}\cup T_1$.    
\end{itemize}
Note that the new coalition $S$ consists of $2n+m+1$ agents.
Moreover, by definition of $\tau_{\XX}$, if $\tau_{\XX}$ assigns \Tru{} to $x$, then $S$ contains $\x$ and if $\tau_{\XX}$ assigns \Fals{} to $x$, then $S$ contains $\nx$.
In addition, for $y\in \YY$, $S$ contains $\yfix$ if $\tau_{\YY}$ assigns \Tru{} to $y$ and $S$ contains $\nyfix$ if $\tau_{\YY}$ assigns \Fals{} to $y$.

We compute the popularity margin between $\pi$ and $\pi^*$. Let $c\in\CC$. Since $\psi(\tau_{\XX},\tau_{\YY})=\Fals$, we have that $c$ has at most two literals in $S$ assigned \Tru{} by $\tau_{\XX}$ and $\tau_{\YY}$. Hence, since the $X$- and $Y$-agents in $S$ correspond to the literals assigned \Tru{} by $\tau_{\XX}$ and $\tau_{\YY}$, there are at most two $X$- or $Y$-agents in $S$ to whom $\ca$ assigns value $-2$ (to the other $X$- or $Y$-agents she assigns $0$). Therefore, all $C$-agents prefer $\pi$ over $\pi^*$. Thus, it is simple to check that $\phi_R(\pi^*,\pi)=-1$ (which stems from the fact that $|C'|-|C|=-1$). Furthermore, we have $\phi_L(\pi^*,\pi)=0$ ($T_1$-agents prefer $\pi$, $T_2$-agents prefer $\pi^*$, and $b_1$ and $b_2$ are indifferent between the partitions).
Altogether, we conclude that $\phi(\pi^*,\pi)=-1$, in contradiction to $\pi^*$ being a popular partition.
Hence, $(\XX,\YY,\psi)$ is a Yes-instance of \QSAT.

\section{Popularity in Nonnegative FHGs}
\label{sec_fhg}

In this section, we discuss the proof of \Cref{thm_fhg}.
We start by describing our reduction from \QSAT.
Subsequently, we give an overview of the proof that satisfiability of the source instance implies the existence of a popular partition and vice versa. 
The detailed proof is deferred to \Cref{apndx_fhg}.

\subsection{Setup of the Reduction}
\label{sec_fhg:setup}

 In this section, we detail the construction of the reduction.
First, we introduce the following No-instance of \FHG, on which the reduction relies.
Consider the FHG induced by a star graph, consisting of a central node $r$ and $k$ leaves $\ell^1, \dots,\ell^k$. 
Agent $r$ assigns value $1$ to all leaves, and the leaves assign value $1$ to $r$, and $0$ to other leaves. 
For $k = 6$, this instance can be seen in the left part of \Cref{fig_fhg:reduction}.
It was shown by \cite{BrBu20a} that the star instance does not admit any popular partitions for $k\geq 6$, whereas it does if $k<6$.
In our reduction, we construct a game which includes multiple such star graphs, each corresponding to a clause from the given logical formula.

Suppose we are given an arbitrary instance $(\XX, \YY, \psi)$ of \QSAT{} with $n\geq 2$, i.e., $\XX$ and $\YY$ contain $n$ variables each. 
As in \Cref{sec_ashg}, denote by $\CC$ the set of clauses in $\psi$, and let $m=|\CC|$; without loss of generality, we may assume that $m\geq 2$. 
We then construct the following FHG, consisting of $11n+8m+3$ agents.
An illustration of the reduction is given in \Cref{fig_fhg:reduction}.

\begin{figure}[t]
\centering
\includegraphics[width=.8\textwidth]{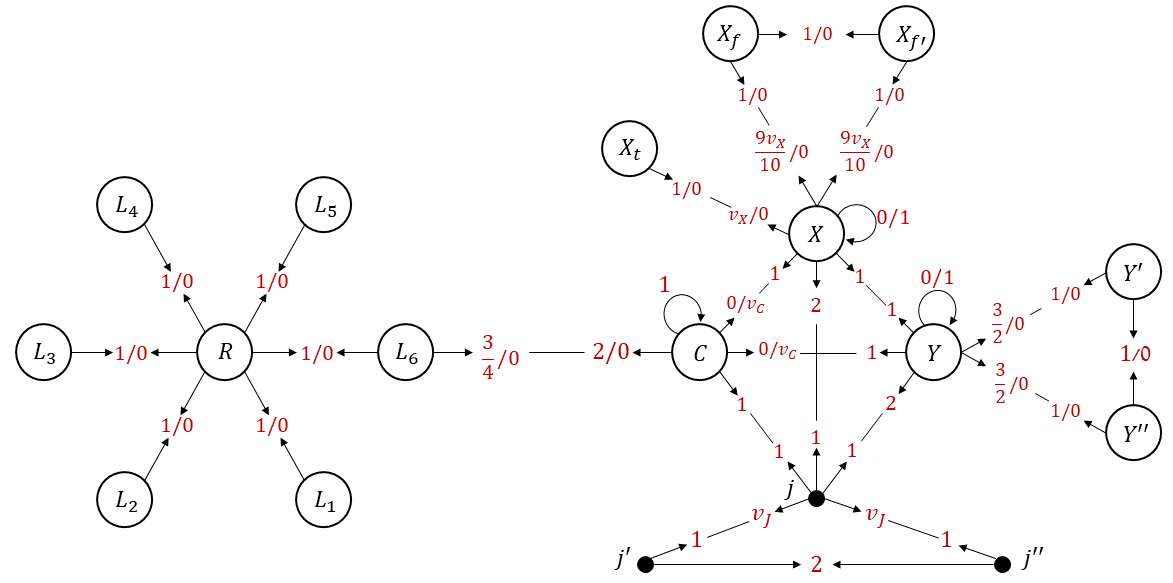}
	\caption{The reduction for the proof of \Cref{thm_fhg}.
	Each node refers to a certain agent type, i.e., to the respective set of agents.
	Edges indicate valuations between all agents in the respective sets.
	When two values $v_1/v_2$ appear, $v_1$ refers to corresponding agents, and $v_2$ to noncorresponding ones.
	Omitted edges imply value $0$.}
	\label{fig_fhg:reduction}
\end{figure}

\begin{itemize}
    \item For every variable $x\in \XX$:
        \begin{itemize}
            \item We create two $X$-agents $\x$ and $\nx$, where the former represents the variable and the latter its negation. 
            We will use $\alpha$ to denote any literal over $\XX$, meaning $\xl$ can correspond to a variable or its negation; accordingly, $\nxl$ will simply correspond to the negated literal.
            If $\x$ and $\nx$ originate in the same variable, they are called \textit{complementary} agents. 
            \item We create a corresponding $\Xt$-agent, $\Xf$-agent, and $\Xfp$-agent, denoted $\xt$, $\xf$, and $\xfp$, respectively.
            The subscripts of these agents indicate ``true'' and ``false'' and these agents are used to deduct the satisfying truth assignments from popular partitions (and vice versa).
        \end{itemize}
    \item For every variable $y\in \YY$:
        \begin{itemize}
            \item We create two $Y$-agents $\yfix$ and $\nyfix$, where the former represents the variable and the latter its negation. 
            We will use $\beta$ to denote any literal over $\YY$, meaning $\y$ can correspond to a variable or its negation; $\ny$ will refer to the agent corresponding to the negated literal. 
            If $\y$ and $\ny$ originate in the same variable, they are called \textit{complementary} agents.            
            \item We create a $Y'$-agent $a_{y}'$ and a $Y''$-agent $a_{y}''$ corresponding to $\yfix$, and a $Y'$-agent $a_{\neg y}'$ and a $Y''$-agent $a_{\neg y}''$ corresponding to $\nyfix$ (we emphasize that, in contrast to the $X$-agents, $\yfix$ and $\nyfix$ have separate $Y'$- and $Y''$-agents).
        \end{itemize}
    \item For every clause $c\in \CC$, we create a corresponding $C$-agent $\ca$, a corresponding $R$-agent $\ra$, and 6 corresponding agents $\la^1, \dots,\la^6$. 
    For each $i\in\{1, \dots,6\}$, $\la^i$ is referred to as an $L_i$-agent. 
    We use the term $L$-agents, or simply ``leaves'', to refer to all $L_1,\dots,L_6$-agents.
    We stress that $C$ always refers to clauses and not to centers of stars, and that $L$ and $R$ no longer refer to left and right side as in \Cref{sec_ashg}.
    For a literal $\alpha$ over $\XX$ (or $\beta$ over $\YY$) occurring in $c$, we refer to the $X$-agent $\xl$ (or $Y$-agent $\y$) as corresponding to clause $c$.
    \item Lastly, we construct three additional agents $j$, $j'$, and $j''$.
    
\end{itemize}
For each agent type, the set of all agents from that type is denoted by the name of the type (e.g., $L_6$ is the set of all $L_6$-agents).
As in \Cref{sec_ashg}, we use the terms \emph{real agents} to refer to the $X$-, $Y$-, and $C$-agents, and \emph{structure agents} to refer to any other agent.
We refer to \Cref{fig_fhg:reduction} for an overview of the valuation functions and to \Cref{apndx_fhg} for a detailed description.
The valuations depend on suitably chosen parameters $v_X=2\frac{4n+m+1}{4n+m+2}$, $v_C=\frac{2n+1}{2n-2.5}$, and $v_J=\frac{3(2n+m-1)}{2(2n+m)}$.
Moreover, valuations missing from the figure correspond to a value of $0$.

Recall that the reduction for ASHGs in \Cref{thm_ashg} used one blown-up No-instance, where the whole combinatorics of the source problem had replaced the role of a single agent.
For FHGs, the combinatorics of the source instance is still encoded similarly.
Literals are still represented by $X$- and $Y$-agents and once again, they have good options to form coalitions together with their corresponding structure agents by forming
\begin{itemize}
    \item coalitions $\{\x,\xt\}$ and $\{\nx,\xf,\xfp\}$ or $\{\nx,\xt\}$ and $\{\x,\xf,\xfp\}$ for the $X$-agents, and 
    \item coalitions $\{\y,\yp,\ypp\}$ for the $Y$-agents.
\end{itemize}
As in the proof for ASHGs, the coalitions of the $X_t$-agents in popular partitions correspond to satisfying truth assignments for the variables in $\XX$.

However, the combinatorics of the source instance is not embedded in a single No-instance, but multiple No-instances are used as gadgets.
This implies that popular partitions have to contain coalitions of some agent in the No-instance with another agent outside the No-instance.
We have one No-instance corresponding to a star for every clause and the only agent that is linked to the gadget by a positive valuation is its corresponding clause agent.
We will see that the only possibility to achieve popularity is to use this single connection by forming coalitions of the type $\{\ca,\la^6\}$.

To prove \Cref{thm_fhg}, we will show that the logical formula is satisfiable if and only if there exists a popular partition in the constructed FHG. 
Every truth assignment lets us define a popular partition along these ideas, while we can show that every popular partition has such a structure and lets us extract a truth assignment.
The two directions of the proof will be sketched in \Cref{sec_fhg:if_satisfiable} and \Cref{sec_fhg:if_popular}.

\subsection{Satisfiability Implies Popular Partition}
\label{sec_fhg:if_satisfiable}
Consider an instance $(\XX,\YY,\psi)$ of \QSAT{} and its reduced FHG as described in the previous section.
Throughout this section, we assume there is a truth assignment $\tau_{\XX}$ to $\XX$ such that for all truth assignments $\tau_{\YY}$ to $\YY$ it holds that $\psi(\tau_{\XX},\tau_{\YY})=\Tru$.
Consider the following partition of the agents, denoted by $\pi^*$:
\begin{itemize}
    \item For each clause $c\in\CC$, $\{\ca,\la^6\}\in\pi^*$ and $\{\ra,\la^1,\la^2,\la^3\}\in\pi^*$.
    \item All remaining $L$-agents (i.e., the $L_4$ and $L_5$-agents) form singleton coalitions.
    \item Agents $j$, $j'$, and $j''$ form a coalition together.
    \item Each $Y$-agent $\y$ forms a coalition with her corresponding agents $\yp\in Y'$ and $\ypp\in Y''$.
    \item For each variable $x\in \XX$, if $x$ is assigned \Tru{} by $\tau_{\XX}$ then $\{\{\x,\xt\},\{\nx,\xf,\xfp\}\}\subseteq\pi^*$, and if $x$ is assigned \Fals{} by $\tau_{\XX}$ then $\{\{\nx,\xt\},\{\x,\xf,\xfp\}\}\subseteq\pi^*$.
\end{itemize}
Our goal is to show $\pi^*$ is a popular partition. We defer the detailed proof of this statement to \Cref{apndx_fhg:if_satisfiable}.
In this section, we give an overview of the proof.

Assume towards contradiction that there exists a partition $\pi$ that is more popular than $\pi^*$. 
An analysis of the sets $S_1=C\cup R\cup L$, $S_2=X\cup \Xt\cup \Xf\cup \Xfp$, $S_3=Y\cup Y'\cup Y''$, and $S_4=\{j',j''\}$ shows that $\phi_{S_i}(\pi^*,\pi)\geq 0$, for each $i\in \{1,2,3,4\}$; namely, there exists an inherent trade-off in each of those sets, where increasing one agent's utility decreases that of another. 
Hence, the only possibility for $\pi$ to be more popular than $\pi^*$ is if $\phi_j(\pi^*,\pi)=-1$, while $\phi_{S_i}(\pi^*,\pi)=0$ for each $i\in \{1,2,3,4\}$. To satisfy $\phi_j(\pi^*,\pi)=-1$, observe that $\pi(j)$ must include some real agents.
We will use $\pi(j)$ to extract an assignment to $\YY$ together with which $\tau_{\XX}$ does not evaluate $\psi$ as \Tru.
Now, a more careful analysis of the sets $S_i$ shows the following facts.
\begin{enumerate}
    \item For all $c\in\CC$, we must have $\phi_{\ca}(\pi^*,\pi)=-1$ (to satisfy $\phi_{S_1}(\pi^*,\pi)=0$).\label{item_fhg1:C}
    \item For all $x\in\XX$ and $\alpha\in\{x,\neg x\}$, if $\pi^*(\xl)=\{\xl,\xf,\xfp\}$ then $\xl\notin\pi(j)$ (to satisfy $\phi_{S_2}(\pi^*,\pi)=0$).\label{item_fhg1:X}
    \item For all $y\in\YY$, we must have $\yfix\notin\pi(\nyfix)$ (to satisfy $\phi_{S_3}(\pi^*,\pi)=0$).\label{item_fhg1:Y}
    \item $j',j''\notin\pi(j)$ (to satisfy $\phi_{S_4}(\pi^*,\pi)=0$, since $\pi(j)$ contains real agents).\label{item_fhg1:J}
\end{enumerate}
\Cref{item_fhg1:X,item_fhg1:Y} imply that $|\pi(j)\cap X|\leq n$ and $|\pi(j)\cap Y|\leq n$. Following this and \Cref{item_fhg1:J}, to satisfy $\phi_j(\pi^*,\pi)=-1$, $\pi(j)$ must consist of $j$, $n$ $X$-agents, $n$ $Y$-agents, all $C$-agents, and no other agent (since $u_j(\pi^*)=\frac{2n+m-1}{2n+m}$). 
Fix clause $c\in \CC$, and let $k_c\in\{0,1,2,3\}$ denote the number of $X$- and $Y$-agents in $\pi(\ca)$ who represent literals from the original clause~$c$. 
By \Cref{item_fhg1:C}, we must have that $u_{\ca}(\pi)>u_{\ca}(\pi^*)$; it can be verified that this holds if and only if $k_c\leq 2$.
 
Consider the following assignment $\tau_{\YY}$ to the $\YY$ variables: variable $y\in \YY$ is assigned \Tru{} if and only if $\yfix\in\pi(j)$ (by \Cref{item_fhg1:Y}, this is a valid assignment). By \Cref{item_fhg1:X}, and since $|\pi(j)\cap X|=n$, the $X$-agents in $\pi(j)$ are exactly the agents $\xl\in X$ for which $\alpha$ is assigned \Tru{} by $\tau_{\XX}$.
Hence, we obtain $\psi(\tau_{\XX},\tau_{\YY})=\Fals$ because for each clause $c\in \CC$ there exists a literal assigned \Fals{} by $(\tau_{\XX},\tau_{\YY})$ (since $k_c\leq 2$).
This gives a contradiction to our choice of $\tau_{\XX}$.
Therefore $\pi^*$ is a popular partition.

\subsection{Popular Partition Implies Satisfiability}
\label{sec_fhg:if_popular}
Throughout this section, we assume that there is a popular partition $\pi^*$ in the constructed FHG.
We will prove that this implies the source instance is a Yes-instance to \QSAT. We defer the detailed proof of this statement to  \Cref{apndx_fhg:if_popular}.
In this section, we give an overview of the proof.

We wish to gain knowledge regarding the structure of $\pi^*$. 
First, since the star instance is a No-instance, some agent in $\{\ra,\la^1,\dots,\la^6\}$ has to form a coalition with an agent outside of this set and we show that this entails that $\{\ca,\la^6\}\in\pi^*$. 
Therefore, it follows that the remaining agents must be partitioned in the ``natural'' way, dictated by the positive relations of corresponding agents.
Specifically, we have that $\{j,j',j''\}\in\pi^*$, for all $\y\in Y$ we have that $\{\y,\yp,\ypp\}\in\pi^*$, and for all $x\in\XX$ we either have that $\{\{\x,\xt\},\{\nx,\xf,\xfp\}\}\subset\pi^*$ or that $\{\{\nx,\xt\},\{\x,\xf,\xfp\}\}\subset\pi^*$. 

This allows us to define the following truth assignment $\tau_{\XX}$ to the $\XX$ variables. 
For each $x\in \XX$, $x$ is assigned \Tru{} if $\pi^*(\x)=\{\x,\xt\}$ and \Fals{} if $\pi^*(\x)=\{\x,\xf,\xfp\}$. 
We claim that $\tau_{\XX}$ is a satisfying assignment to the source instance, i.e., that $\psi(\tau_{\XX},\tau_{\YY}) = \Tru$ for all truth assignments $\tau_{\YY}$ to the $\YY$ variables.

Assume otherwise, namely that there exists a truth assignment $\tau_{\YY}$ to the $\YY$ variables such that $\psi(\tau_{\XX},\tau_{\YY})=\Fals$. 
We will show that this allows us to find a partition that is more popular than $\pi^*$. 
Therefore, consider the partition $\pi$ obtained from $\pi^*$ by extracting the following agents from their respective coalitions to form a new coalition $S$:
\begin{itemize}
    \item All $\xl\in X$ such that $\{\xl,\xt\}\in \pi^*$, for some $\xt\in\Xt$.
    \item All $\y\in Y$ such that the literal represented by $\y$ is assigned \Tru{} by $\tau_{\YY}$.
    \item All $C$-agents and agent $j$.
\end{itemize}
Note that the new coalition consists of $2n+m+1$ agents.
Moreover, by definition of $\tau_{\XX}$, if $\tau_{\XX}$ assigns \Tru{} to $x$, then $S$ contains $\x$ and if $\tau_{\XX}$ assigns \Fals{} to $x$, then $S$ contains $\nx$.
In addition, for $y\in \YY$, $S$ contains $\yfix$ if $\tau_{\YY}$ assigns \Tru{} to $y$ and $S$ contains $\nyfix$ if $\tau_{\YY}$ assigns \Fals{} to $y$.

We compute the popularity margin between $\pi$ and $\pi^*$. All $Y$-agents are indifferent between $\pi^*$ and $\pi$. All $X$-agents in $S$ prefer $\pi$, while their corresponding $\Xt$-agents prefer $\pi^*$.
Let $c\in\CC$. Since $\psi(\tau_{\XX},\tau_{\YY})=\Fals$, we have that $c$ has at most two literals in $S$ assigned \Tru{} by $\tau_{\XX}$ and $\tau_{\YY}$. Hence, since the $X$- and $Y$-agents in $S$ correspond to the literals assigned \Tru{} by $\tau_{\XX}$ and $\tau_{\YY}$, there are at most two $X$- or $Y$-agents in $S$ to whom $\ca$ assigns value $0$. Thus, it may be verified that $u_{\ca}(\pi)>u_{\ca}(\pi^*)$, i.e., $\ca$ prefers $\pi$; this cancels out with $\la^6$, who clearly prefers $\pi^*$. 
Lastly, one may verify that $j$ prefers $\pi$ over $\pi^*$. 
Altogether, we conclude that $\phi(\pi^*,\pi)=-1$, in contradiction to $\pi^*$ being a popular partition. Hence, $(\XX,\YY,\psi)$ is a Yes-instance of \QSAT.

\section{Conclusion}

We considered the complexity of deciding whether popular partitions exist in typical classes of hedonic games.
By showing that this problem is $\Sigma_2^p$-complete, we pinpoint its precise complexity for ASHGs as well as FHGs with nonnegative valuation functions.
Hence, allowing coalitions of size at least three can raise the complexity of popularity from completeness for the first to the second level of the polynomial hierarchy.

Our work is an important step in understanding popularity in coalition formation.
However, there are still various dimensions along which a deeper understanding would be welcome.
Firstly, our methods might aid in resolving the exact complexity of popularity in other classes of coalition formation games for which popularity was considered before \citep{KLR+20a,KeRo20a,BrBu20a,CsPe21a}.
Secondly, it would be interesting to consider popularity in other classes of hedonic games.
Possible candidates include modified fractional hedonic games \citep{Olse12a} and anonymous hedonic games \citep{BoJa02a}.

Finally, popularity has the closely related concepts of strong popularity, where a partition has to strictly win the vote in a pairwise comparison against every other partition, and mixed popularity, which considers probability distributions of partitions that are popular in expectation.
In the domain of matching, strong popularity and mixed popularity were first considered by \citet{Gard75a} and \citet{KMN11a}, respectively, and studied by \citet{BrBu20a} in ASHGs and FHGs.
In these classes, \citet{BrBu20a} show that deciding whether a strongly popular partition exists is \coNP-hard\footnote{This aligns with known complexity results for strong popularity in other classes of coalition formation games \citep{KLR+20a,KeRo20a}.} while computing a mixed popular partition is \NP-hard.
The exact complexity of both problems remains open.
Notably, $\Sigma_2^p$ seems not to be the right complexity class for these problems because mixed popular partitions always exist and strongly popular partitions are unique whenever they exist.
Resolving their complexity could lead to an intriguing complexity picture of concepts of popularity in hedonic games.

\section*{Acknowledgments}
Martin Bullinger was supported by the AI Programme of The Alan Turing Institute. 
Matan Gilboa was supported by an Oxford-Reuben Foundation Graduate Scholarship.
We would like to thank Edith Elkind for many fruitful discussions.

\appendix

\section*{Appendix}

In the appendix, we provide the proof details for \Cref{thm_ashg,thm_fhg}.

\section{Detailed Proof of \Cref{thm_ashg}}
\label{apndx_ashg}
In this section, we provide the proof of \Cref{thm_ashg}.
The proof is based on the reduction from \QSAT{} constructed in \Cref{sec_ashg:setup}. 
To complete the description of the reduced instance, we now specify the exact valuation functions of the agents as detailed in \Cref{fig_ashg:reduction}. 
A summary of the agents' valuations can be seen in \Cref{table_ASHG}.

Whenever a value assigned by some agent $\ga$ to some agent $\gap$ is not specified, it means $\ga$ assigns a value of $-\infty$ to $\gap$. 
We use $-\infty$ to denote a constant negative number of very large magnitude. 
The property that we need this number to satisfy is that no matter which positive values an agent may get in her coalition, she would still get a negative utility if she has even a single $-\infty$ value in that coalition.
For a polynomial-time reduction, one can, for instance, set $\infty = 6(12n+4m-1)$, which corresponds to the product of the largest positive utility and the number of agents.

\begin{itemize}
    \item Any two complementary $X$-agents $\x$ and $\nx$ both assign value $\frac{1}{2}$ to their corresponding $\Xt$-agent, and~$4$ to their corresponding $\Xf$-agent. 
    Furthermore, they assign~$0$ to any real agent, except for their complementary $X$-agent (to whom they assign value $-\infty$). 
    Lastly, they assign value $1$ to $b_1$, and $2$ to $b_2$.
    
    \item Any $Y$-agent $\y$ assigns value $\frac{1}{2}$ to her corresponding $Y'$-agent $\yp$. 
    Furthermore, she assigns~$0$ to any real agent, except for her complementary agent $\ny\in Y$ (to whom she assigns value $-\infty$). 
    Lastly, she assigns value $1$ to $b_1$, and $2$ to $b_2$.
    
    \item Any $C$-agent $\ca$ assigns value $-2$ to any $X$-agent or $Y$-agent corresponding to a literal that is present in the original clause $c$.
    Moreover, she assigns~$0$ to any other real agent and to all $C'$-agents. 
    Furthermore, she assigns~$5$ to $b_1$, and $6$ to $b_2$.
    
    \item Any $\Xt$- and $\Xf$-agents corresponding to variable $x\in \XX$ assign~$1$ to their two corresponding $X$-agents $\x$ and $\nx$.
    
    \item Any $Y'$-agent $\yp$ assigns value $1$ to her corresponding $Y$-agent $\y$.
    
    \item Any $C'$-agent assigns value $0$ to any other $C'$-agent and~$1$ to all $C$-agents.

    \item Any $T_1$-agent assigns value $0$ to any other $T_1$-agent, $1$ to $b_1$ and $2$ to $b_2$.

    \item Any $T_2$-agent assigns value $0$ to any other $T_2$-agent, $1$ to $b_1$ and $2$ to $b_2$.

    \item Agent $b_1$ assigns value $1$ to any real agent, $T_1$-agent, or $T_2$-agent.

    \item Agent $b_2$ assigns value $2$ to any real agent, $T_1$-agent, or $T_2$-agent.
\end{itemize}

\begin{table}[t]
\centering
\caption{Summary of values of ASHG reduction (assigned by row agents to column agents). When two values $v_1/v_2$ appear, $v_1$ refers to corresponding agents, and $v_2$ to noncorresponding.}
\label{table_ASHG}
\setlength{\tabcolsep}{6pt}
\renewcommand{\arraystretch}{1.1}
\begin{tabular}{ c c c c c c c c c c c c } 
\hline
& \textbf{$\x$} & \textbf{$\y$} & \textbf{$\ca$} & \textbf{$T_1$} & \textbf{$T_2$} & \textbf{$b_1$} & \textbf{$b_2$} & \textbf{$\xt$} & \textbf{$\xf$} & \textbf{$\yp$} & \textbf{$c'$} \\
\hline
$\x$ & $-\infty/0$ & $0$ & $0$ & $-\infty$ & $-\infty$ & $1$ & $2$ & $\frac{1}{2}/-\infty$ & $4/-\infty$ & $-\infty$ & $-\infty$ \\ 
 
$\y$ & $0$ & $-\infty/0$ & $0$ & $-\infty$ & $-\infty$ & $1$ & $2$ & $-\infty$ & $-\infty$ & $\frac{1}{2}/-\infty$ & $-\infty$ \\ 
 
$\ca$ & $-2/0$ & $-2/0$ & $0$ & $-\infty$ & $-\infty$ & $5$ & $6$ & $-\infty$ & $-\infty$ & $-\infty$ & $0$ \\ 
 
$T_1$ & $-\infty$ & $-\infty$ & $-\infty$ & $0$ & $-\infty$ & $1$ & $2$ & $-\infty$ & $-\infty$ & $-\infty$ & $-\infty$ \\ 

$T_2$ & $-\infty$ & $-\infty$ & $-\infty$ & $-\infty$ & $0$ & $1$ & $2$ & $-\infty$ & $-\infty$ & $-\infty$ & $-\infty$ \\ 

$b_1$ & $1$ & $1$ & $1$ & $1$ & $1$ & $0$ & $-\infty$ & $-\infty$ & $-\infty$ & $-\infty$ & $-\infty$ \\ 

$b_2$ & $2$ & $2$ & $2$ & $2$ & $2$ & $-\infty$ & $0$ & $-\infty$ & $-\infty$ & $-\infty$ & $-\infty$ \\ 

$\xt$ & $1/-\infty$ & $-\infty$ & $-\infty$ & $-\infty$ & $-\infty$ & $-\infty$ & $-\infty$ & $-\infty$ & $-\infty$ & $-\infty$ & $-\infty$ \\ 

$\xf$ & $1/-\infty$ & $-\infty$ & $-\infty$ & $-\infty$ & $-\infty$ & $-\infty$ & $-\infty$ & $-\infty$ & $-\infty$ & $-\infty$ & $-\infty$ \\ 

$\yp$ & $-\infty$ & $1/-\infty$ & $-\infty$ & $-\infty$ & $-\infty$ & $-\infty$ & $-\infty$ & $-\infty$ & $-\infty$ & $-\infty$ & $-\infty$ \\ 

$c'$ & $-\infty$ & $-\infty$ & $1$ & $-\infty$ & $-\infty$ & $-\infty$ & $-\infty$ & $-\infty$ & $-\infty$ & $-\infty$ & $0$ \\ 
\hline
\end{tabular}
\end{table}

In the subsequent two sections, we will prove that satisfiability of the source instance implies the existence of a popular partition and vice versa.

\subsection{Satisfiability Implies Popular Partition}
\label{apndx_ashg:if_satisfiable}
Throughout this section, we assume that $(\XX,\YY,\psi)$ is a Yes-instance of \QSAT.
Hence, there is a truth assignment $\tau_{\XX}$ to the variables in $\XX$ such that for all truth assignments $\tau_{\YY}$ to the variables in $\YY$ it holds that $\psi(\tau_{\XX},\tau_{\YY}) = \Tru$. 
Consider the following partition of the agents, denoted by $\pi^*$.
\begin{itemize}
    \item For each $x\in\XX$, if $x$ is assigned \Tru{} by $\tau_{\XX}$ then $\{\{\x,\xt\},\{\nx,\xf\}\}\subseteq\pi^*$, and if $x$ is assigned \Fals{} by $\tau_{\XX}$ then $\{\{\nx,\xt\},\{\x,\xf\}\}\subseteq\pi^*$.
    \item Each $Y$-agent $\y$ forms a coalition with her corresponding $Y'$-agent $\yp$.
    \item The coalition $C\cup C'$ is formed.
    \item Agent $b_1$ forms a coalition with all $T_1$-agents.
    \item Agent $b_2$ forms a coalition with all $T_2$-agents.
\end{itemize}
Our goal is to show that $\pi^*$ is a popular partition. 
To this end, assume towards contradiction that there exists a partition $\pi$ that is more popular than $\pi^*$. 
By \Cref{prop:wlogPO}, we may assume that $\pi$ is Pareto-optimal (note that, however, $\pi$ need not be popular).
We work towards deriving a contradiction by computing popularity margins for various agent sets.
We split this into several lemmas.

\begin{lemma}
\label{lem_ashg1:Y_balance}
Let $\y\in Y$. Then $\phi_{\{\y,\yp\}}(\pi^*,\pi)\geq 0$.
\end{lemma}

\begin{proof}
If $\pi(\yp)=\pi^*(\yp)=\{\y,\yp\}$, then the statement is trivial. 
Otherwise, we must have $u_{\yp}(\pi)<1=u_{\yp}(\pi^*)$, and so, even if $u_{\y}(\pi)>u_{\y}(\pi^*)$ we obtain the required result.
\qed
\end{proof}

\begin{lemma}
\label{lem_ashg1:X_balance}
Let $x\in\XX$. Then $\phi_{\{\x,\nx,\xt,\xf\}}(\pi^*,\pi)\geq 0$.
\end{lemma}

\begin{proof}
Assume towards contradiction that $\phi_{\{\x,\nx,\xt,\xf\}}(\pi^*,\pi)<0$.
If $\nx\in\pi(\x)$ then $u_{\x}(\pi)<0<u_{\x}(\pi^*)$ and $u_{\nx}(\pi)<0<u_{\nx}(\pi^*)$, and we are done. 
Otherwise, observe that $u_{\xt}(\pi)\leq 1=u_{\xf}(\pi^*)$ and $u_{\xf}(\pi)\leq 1=u_{\xf}(\pi^*)$. Therefore, $\xt$ and $\xf$ cannot prefer $\pi$ over $\pi^*$.
If both of them prefer $\pi^*$, then we are done.
If both of them are indifferent, then either $\{\{\x,\xt\},\{\nx,\xf\}\}\subseteq\pi^*$ or $\{\{\nx,\xt\},\{\x,\xf\}\}\subseteq\pi^*$; either way we have $\phi_{\{\x,\nx,\xt,\xf\}}(\pi^*,\pi)\geq 0$, and again we are done.

So assume only one of $\xt$ and $\xf$ is indifferent, while the other prefers $\pi^*$ over $\pi$. 
Thus, because of our initial assumption, it must hold that both $\x$ and $\nx$ prefer $\pi$ over $\pi^*$.
Therefore, it is sufficient to show that one of them does not prefer $\pi$. 
Without loss of generality, assume that $\{\nx,\xf\}\in\pi^*$. 
If $\nx$ prefers~$\pi$, then $\pi(\nx)$ must contain $\xf$ and at least one additional agent to whom she assigns a positive value (we observe that without $\xf$, the utility of $\nx$ is less than $4$). 
Therefore, $\xf$ must prefer $\pi^*$, implying $\xt$ must be indifferent between $\pi$ and $\pi^*$. 
But then we must have $\{\x,\xt\}\in\pi$, implying that $\x$ does not prefer $\pi$.
We obtain a contradiction.
\qed
\end{proof}

\begin{lemma}
\label{lem_ashg1:C_balance_at_most_1}
Consider the $C$- and $C'$-agents. 
If for every agent $\ca\in C$ we have $\phi_{\ca}(\pi^*,\pi)<0$, then $\phi_{C\cup C'}(\pi^*,\pi)=-1$. 
Otherwise, $\phi_{C\cup C'}(\pi^*,\pi)\geq 0$.
\end{lemma}

\begin{proof}
Observe that all $C'$-agents obtain maximal utility in $\pi^*$, and so $\phi_{C'}(\pi^*,\pi)\geq 0$. 
If there exists some agent $c'\in C'$ who is indifferent between $\pi^*$ and $\pi$, then it must be that $\pi(c')$ contains the entire set $C$, possibly some $C'$-agents, and no other agent. 
Hence, we have $\phi_{C}(\pi^*,\pi)\geq 0$ and we are in the second case of the lemma.
As $\phi_{C\cup C'}(\pi^*,\pi)\geq 0$, we are done. 

It remains to consider the case where no agent $c'\in C'$ is indifferent, and so, as observed, they all prefer $\pi^*$.
If all $C$-agents prefer $\pi$ then $\phi_{C\cup C'}(\pi^*,\pi)=-1$ (recall that there are $m$ $C$-agents but only $m-1$ $C'$-agents).
Otherwise, we get $\phi_{C\cup C'}(\pi^*,\pi)\geq 0$.
\qed
\end{proof}

By now, we understand the popularity margin of the right-side agents.
Together, the last three lemmas show that $\phi_{R}(\pi^*,\pi)\geq -1$.
For estimating the popularity margin of the left-side agents, we determine the structure of their coalitions.

\begin{lemma}
\label{lem_ashg1:T_separated}
Fix two agents $t_1\in T_1$ and $t_2\in T_2$, and denote $S=\pi(t_1)$. 
Then it holds that $t_2\notin S$.
\end{lemma}

\begin{proof}
Assume otherwise. 
By \Cref{lem_ashg1:Y_balance,lem_ashg1:X_balance,lem_ashg1:C_balance_at_most_1}, we have that $\phi_{R}(\pi^*,\pi)\geq -1$. 
Hence, it is sufficient to show that $\phi_{L}(\pi^*,\pi)\geq 1$ for obtaining a contradiction. 
Clearly $b_1\in S$ or $b_2\in S$, as otherwise it is a Pareto improvement to dissolve $S$ into singletons and $\pi$ is not Pareto-optimal.
In addition, observe that any $T_1$- or $T_2$-agent in $S$ obtains negative utility in $\pi$. 
If $b_1,b_2\in S$, then all $T_1$- and $T_2$-agents obtain utility at most $0$ in $\pi$, and thus prefer $\pi^*$ over $\pi$, and we are done.

So suppose $b_i\in S$ and $b_j\notin S$, where $i,j\in\{1,2\}$.
Consider the coalition $S'=\pi(b_j)$. 
If $S'$ is a singleton, or contains a right-side agent, or a combination of $T_1$- and $T_2$-agents, then all $T_1$- and $T_2$-agents obtain utility at most 0 in $\pi$, and we have the same contradiction as before. 
Otherwise, $S'$ contains either only $T_1$-agents or only $T_2$-agents (apart from $b_j$). 
Therefore, $b_j$ prefers $\pi^*$ (since $S'$ contains at most $2n+m-1$ $T_1$- or $T_2$-agents). 
Furthermore, we have that either all $T_1$-agents and at least one $T_2$-agent prefer $\pi^*$, or all $T_2$-agents and at least one $T_1$-agent prefer $\pi^*$ (considering that $S$ contains both $T_1$- and $T_2$-agents, and $S'$ contains only one of the types $T_1$ or $T_2$). 
Hence, we have a contradiction, since together $\phi_{L}(\pi^*,\pi)\geq 1$.
\qed
\end{proof}

\begin{lemma}
\label{lem_ashg1:B1_or_B2_right_side}
Either $\pi(b_1)$ or $\pi(b_2)$ must contain a right-side agent.
\end{lemma}

\begin{proof}
Assume towards contradiction that neither $\pi(b_1)$ nor $\pi(b_2)$ contain a right-side agent. 
Then we have that all $C$-agents are either indifferent between $\pi^*$ and $\pi$, or prefer $\pi^*$ (since $b_1$ and $b_2$ are the only agents to whom $C$-agents assign positive utility). 
Thus, by \Cref{lem_ashg1:C_balance_at_most_1,lem_ashg1:X_balance,lem_ashg1:Y_balance}, we have that $\phi_R(\pi^*,\pi)\geq 0$. 
So it is enough to show that $\phi_L(\pi^*,\pi)\geq 0$ to reach a contradiction.

By \Cref{lem_ashg1:T_separated} and by the assumption that no right-side agents form a coalition with $b_1$ and $b_2$ in $\pi$, we have that agents $b_1$ and $b_2$ cannot prefer $\pi$ over $\pi^*$. 
It remains to consider the $T_1$- and $T_2$-agents.  
If $b_1\in\pi(b_2)$, then either all $T_1$-agents have utility 0 in $\pi$ (and thus prefer $\pi^*$), or all $T_2$-agents do, and we are done. Hence, agents $b_1$ and $b_2$ are separated.

If $\pi(b_2)\cap T_1\neq\emptyset$, then by \Cref{lem_ashg1:T_separated} all $T_2$-agents obtain a utility of at most $1$ in $\pi$ (and thus prefer $\pi^*$), and again we are done. 
Hence, $\pi(b_2)\cap T_1=\emptyset$.
Now, if $\pi(b_1)\cap T_2\neq\emptyset$, then all $T_1$-agents have utility 0 in $\pi$ (and thus prefer $\pi^*$), and we are done.
Otherwise, all $T_1$-agents can obtain a utility of at most~$1$ in $\pi$, while all $T_2$-agents can obtain a utility of at most~$2$, and so they are all either indifferent between $\pi$ and $\pi^*$, or prefer $\pi^*$. 
In all cases, we obtain $\phi_L(\pi^*,\pi)\geq 0$, a contradiction to $\pi$ being more popular than $\pi^*$.
\qed
\end{proof}

\begin{lemma}
\label{lem_ashg1:B1_right_side}
It holds that $\pi(b_1)$ contains a right-side agent. 
\end{lemma}

\begin{proof}
Assume towards contradiction $\pi(b_1)$ contains no right-side agents. Then, by \Cref{lem_ashg1:B1_or_B2_right_side}, $\pi(b_2)$ contains a right-side agent, implying that any $T_1$- or $T_2$-agent in $\pi(b_2)$ obtains a negative utility. 
Hence, all $T_1$- and $T_2$-agents obtain utility at most~$1$ (as the only other agent to whom they assign positive value is $b_1$), implying all $T_2$-agents prefer $\pi^*$, while all $T_1$-agents are either indifferent or prefer $\pi^*$; therefore, we have $\phi_{T_1\cup T_2}(\pi^*,\pi)\geq 2n+m\geq 3$ (since $|T_2|=2n+m\geq 3$). 
By \Cref{lem_ashg1:X_balance,lem_ashg1:Y_balance,lem_ashg1:C_balance_at_most_1}, we have $\phi_{R\cup\{b_1,b_2\}}(\pi^*,\pi)\geq -3$, and thus we obtain $\phi(\pi^*,\pi)\geq 0$, a contradiction to $\pi$ being more popular that $\pi^*$.
\qed
\end{proof}

\begin{lemma}
\label{lem_ashg1:B2_no_right_side}
It holds that $\pi(b_2)$ contains no right-side agents.
\end{lemma}

\begin{proof}
Assume otherwise. Then by \Cref{lem_ashg1:B1_right_side}, we have that all $T_1$- and $T_2$-agents prefer $\pi^*$ over $\pi$, implying $\phi_{T_1\cup T_2}(\pi^*,\pi)\geq 4n+2m\geq 6$. \Cref{lem_ashg1:X_balance,lem_ashg1:Y_balance,lem_ashg1:C_balance_at_most_1} imply $\phi_{R\cup\{b_1,b_2\}}(\pi^*,\pi)\geq -3$. Hence, we get $\phi(\pi^*,\pi)\geq 0$, in contradiction to $\pi$ being more popular that $\pi^*$.
\qed
\end{proof}

\begin{lemma}
\label{lem_ashg1:B2_TT_balance}
It holds that $\phi_{T_1\cup T_2\cup \{b_2\}}(\pi^*,\pi)\geq 0$.
\end{lemma}

\begin{proof}
By \Cref{lem_ashg1:T_separated,lem_ashg1:B2_no_right_side}, $b_2$ cannot prefer $\pi$ over $\pi^*$. 
Furthermore, by \Cref{lem_ashg1:B1_right_side,lem_ashg1:T_separated}, either all $T_1$-agents obtain utility $0$ in $\pi$ (and thus prefer $\pi^*$), or all $T_2$-agents do. Hence, we have $\phi_{T_1\cup T_2\cup \{b_2\}}(\pi^*,\pi)\geq 0$.
\qed
\end{proof}

\begin{lemma}
\label{lem_ashg1:b1_at_least_indifferent}
It holds that $\phi_{b_1}(\pi^*,\pi)\leq 0$.
\end{lemma}
\begin{proof}
If not, by \Cref{lem_ashg1:C_balance_at_most_1,lem_ashg1:B2_TT_balance,lem_ashg1:X_balance,lem_ashg1:Y_balance} we obtain a contradiction to our initial assumption of $\phi(\pi^*,\pi)<0$.
\end{proof}

\begin{lemma}
\label{lem_ashg1:T_right_side}
Let $\ga$ be a right-side agent, and denote $S=\pi(\ga)$. Then it holds that $T_1\cap S=\emptyset$ and $T_2\cap S=\emptyset$.
\end{lemma}

\begin{proof}
Assume towards contradiction that $S$ contains a $T_1$- or $T_2$-agent. 
Then it must hold that $b_1\in S$, as otherwise it is a Pareto improvement to remove any $T_1$- or $T_2$-agents from $S$.
Hence, by \Cref{lem_ashg1:B2_no_right_side}, no $C$-agent prefers $\pi$ over $\pi^*$, because $b_1$ and $b_2$ are the only agents to whom they assign a positive value.
Therefore, by \Cref{lem_ashg1:C_balance_at_most_1}, we have that $\phi_{C\cup C'}(\pi^*,\pi)\geq 0$. 
Thus, by \Cref{lem_ashg1:B2_TT_balance,lem_ashg1:X_balance,lem_ashg1:Y_balance}, in order for $\pi$ to be more popular than $\pi^*$, it must be the case that $b_1$ prefers $\pi$ over $\pi^*$, while the following hold:
\begin{enumerate}
    \item We have $\phi_{C\cup C'}(\pi^*,\pi)=0$. \label{lem_ashg1:T_right_side_item_C}
    \item For all $\y\in Y$, we have $\phi_{\{\y,\yp\}}(\pi^*,\pi)=0$.\label{lem_ashg1:T_right_side_item_Y}
    \item For all $x\in\XX$, we have $\phi_{\{\x,\nx,\xt,\xf\}}(\pi^*,\pi)=0$.\label{lem_ashg1:T_right_side_item_X}
\end{enumerate}
Since $S$ contains both a right-side agent and a $T_1$- or $T_2$-agent, it is clear that all agents in $S$ apart from $b_1$ obtain negative utility, and thus prefer $\pi^*$. 
Furthermore, since $b_1$ prefers $\pi$, $S$ contains no structure agents from the right side. 
We will show $S$ contains no real agents either, contradicting the assumption that $S$ contains a right-side agent.

If $S$ contains a $C$-agent, then all $C'$-agents as well as this agent prefer $\pi^*$, and therefore $\phi_{C\cup C'}(\pi^*,\pi)\ge 1$, contradicting \Cref{lem_ashg1:T_right_side_item_C}.
If $S$ contains agent $\y\in Y$, then $\phi_{\{\y,\yp\}}(\pi^*,\pi)=2$, in contradiction to \Cref{lem_ashg1:T_right_side_item_Y}.
If, for some $x\in\XX$ and $\alpha\in\{x,\neg x\}$, $S$ contains $\xl$, then \Cref{lem_ashg1:T_right_side_item_X} is contradicted: Agent $\xl$ prefers $\pi^*$ over $\pi$.
In addition, one of $\xt$ and $\xf$ must prefer $\pi^*$, while the other is either indifferent, or also prefers $\pi^*$.

Thus, we conclude that $T_1\cap S=\emptyset$ and $T_2\cap S=\emptyset$. 
\qed
\end{proof}

In the next few lemmas, we wish to characterize the exact structure of the coalition of $b_1$. 
For this purpose, we denote $S=\pi(b_1)$ for the remainder of this section.
Our goal is to use this coalition to extract a truth assignment to $\YY$ to contradict that $\tau_{\XX}$ is a satisfying truth assignment for all truth assignments to $\YY$.

\begin{lemma}
\label{lem_ashg1:b1_reals}
Apart from $b_1$, $S$ contains only real agents.
\end{lemma}
\begin{proof}
By \Cref{lem_ashg1:B1_right_side}, $S$ contains a right-side agent.
Hence, by \Cref{lem_ashg1:B2_no_right_side,lem_ashg1:T_right_side}, it contains no other left-side agent.
Moreover, by \Cref{lem_ashg1:b1_at_least_indifferent}, $b_1$ must weakly prefer $\pi$ and, therefore, cannot contain a right-side structure agent.
\end{proof}

\begin{lemma}
\label{lem_ashg1:b1_x_nx}
Let $x\in\XX$. Then $\x\notin S$ or $\nx\notin S$.
\end{lemma}
\begin{proof}
Assume towards contradiction that $\x,\nx\in S$. Then we have $u_{\x}(\pi)<0$, $u_{\nx}(\pi)<0$, $u_{\xt}(\pi)\leq 0$ and $u_{\xf}(\pi)\leq 0$, and therefore all four of those agents prefer $\pi^*$ over $\pi$. 
Hence, by \Cref{lem_ashg1:B2_TT_balance,lem_ashg1:X_balance,lem_ashg1:Y_balance,lem_ashg1:C_balance_at_most_1}, we have a contradiction to the assumption that $\pi$ is more popular than $\pi^*$.
\end{proof}

\begin{lemma}
\label{lem_ashg1:b1_y_ny}
Let $y\in \YY$.
Then $\yfix\notin S$ or $\nyfix\notin S$.
\end{lemma}
\begin{proof}
The proof is identical to the proof of \Cref{lem_ashg1:b1_x_nx}, substituting $X$, $\x$, $\nx$, $\xt$, and $\xf$ by $Y$, $\yfix$, $\nyfix$, $a_{y}'$, and $a_{\neg y}'$, respectively.
\end{proof}

\begin{lemma}
\label{lem_ashg1:b1_coalition}
It holds that $|S\cap X|=|S\cap Y|=n$ and $|S\cap C|=m$.
\end{lemma}
\begin{proof}
By \Cref{lem_ashg1:b1_x_nx}, it holds that $|S\cap X|\leq n$. Similarly, by \Cref{lem_ashg1:b1_y_ny}, we have that $|S\cap Y|\leq n$.
By \Cref{lem_ashg1:b1_at_least_indifferent}, we know that $u_{b_1}(\pi) \ge u_{b_1}(\pi^*) = 2n + m$.
Since, by \Cref{lem_ashg1:b1_reals}, it holds that $S\setminus \{b_1\}\subseteq X\cup Y\cup C$ and $b_1$ obtains a utility of~$1$ from all of these agents, we conclude that $S\cap X = n$, $S\cap Y = n$, and $S$ contains all $m$ $C$-agents.
\end{proof}

\begin{lemma}
\label{lem_ashg1:which_x_can_deviate}
Let $x\in\XX$ and $\alpha\in\{x,\neg x\}$, such that $\xl\in\pi^*(\xt)$ and $\nxl\in\pi^*(\xf)$. Then it must be that $\xl\in S$. 
\end{lemma}
\begin{proof}
Assume towards contradiction that $\xl\notin S$.
Then, by \Cref{lem_ashg1:b1_coalition,lem_ashg1:b1_x_nx}, we have that $\nxl\in S$. 
By \Cref{lem_ashg1:b1_reals} we have that $\xt,\xf\notin S$, and therefore agent $\nxl$ prefers $\pi^*$ over $\pi$.
If $\{\xl,\xt\}\in\pi$, then $\xl$ and $\xt$ are indifferent between $\pi^*$ and $\pi$, while $u_{\xf}(\pi)\leq 0$, and so $\xf$ prefers $\pi^*$. 
Otherwise, if $\{\xl,\xf\}\in\pi$, then $\xf$ is indifferent, and $\xt$ prefers $\pi^*$. 
In any other case, we have that both $\xt$ and $\xf$ prefer $\pi^*$. 
Hence, either way, we have that $\phi_{\{\xl,\nxl,\xt,\xf\}}(\pi^*,\pi)>0$. Therefore, by \Cref{lem_ashg1:C_balance_at_most_1,lem_ashg1:Y_balance,lem_ashg1:X_balance,lem_ashg1:B2_TT_balance}, we have a contradiction to the assumption that $\pi$ is more popular than $\pi^*$.
\end{proof}

\begin{lemma}
\label{lem_ashg1:all_c_prefer_pi}
For every agent $\ca\in C$ we have $\phi_{\ca}(\pi^*,\pi)<0$.
\end{lemma}

\begin{proof}
By \Cref{lem_ashg1:b1_coalition}, we have that $u_{b_1}(\pi)=2n+m$, and so it holds that $\phi_{b_1}(\pi^*,\pi)=0$. Hence, by \Cref{lem_ashg1:C_balance_at_most_1,lem_ashg1:Y_balance,lem_ashg1:X_balance,lem_ashg1:B2_TT_balance}, the statement must hold in order to satisfy the assumption that $\pi$ is more popular than $\pi^*$. 
\end{proof}

We have now gained enough information about the structure of $\pi$ to contradict the fact that it is more popular than $\pi^*$. 
Consider the truth assignment $\tau_{\YY}$ to the $\YY$ variables, where $y\in\YY$ is assigned \Tru{} if and only if $\yfix\in S$ (by \Cref{lem_ashg1:b1_y_ny,lem_ashg1:b1_coalition}, this is a valid assignment). 
By \Cref{lem_ashg1:which_x_can_deviate} and by definition of $\tau_{\YY}$, the $X$- and $Y$-agents in $S$ represent exactly the literals assigned \Tru{} by $\tau_{\XX}$ and $\tau_{\YY}$. 
Furthermore, observe that according to \Cref{lem_ashg1:all_c_prefer_pi}, for every $c\in\CC$, $S$ contains at most $2$ $X$- or $Y$-agents representing literals from~$c$ (otherwise, $\ca$ would have negative utility in $\pi$). Therefore, we have that every $c\in\CC$ contains at most $2$ literals assigned \Tru{} by $\tau_{\XX}$ and $\tau_{\YY}$, implying that $\psi(\tau_{\XX},\tau_{\YY}) = \Fals$. 
We thus obtain a contradiction to our choice of $\tau_{\XX}$, implying there cannot exist a partition $\pi$ more popular than $\pi^*$.

This concludes the proof that the existence of a satisfying assignment to the \QSAT{} instance implies the existence of a popular partition in the constructed ASHG. \qed

\subsection{Popular Partition Implies Satisfiability}
\label{apndx_ashg:if_popular}
Throughout this section, we assume that there is a popular partition $\pi^*$ in the reduced ASHG.
We will prove that this implies that the source instance is a Yes-instance to \QSAT.
We observe that $\pi^*$ must be Pareto-optimal, as otherwise, applying a Pareto improvement will yield a more popular partition, a contradiction.
To extract a truth assignment to the variables in $\XX$, we have to understand the structure of~$\pi^*$. 
The first few lemmas lead up to the conclusion that there is a separation between left-side agents and right-side agents. 

\begin{lemma}
\label{lem_ashg2:TT_separate}
Let $S\in\pi^*$ be a coalition. If $S\cap T_1\neq \emptyset$, then $S\cap T_2=\emptyset$.
\end{lemma}

\begin{proof}
Assume towards contradiction that $S\in\pi^*$ contains a $T_1$- and a $T_2$-agent.
Observe that the only agents who may have positive utility in $S$ are $b_1$ and $b_2$. 
Hence, the partition $\pi$ obtained from $\pi^*$ by dissolving all members of $S$ into singletons is more popular: if $b_1,b_2\in S$, then all members of $S$ prefer $\pi$ over $\pi^*$; otherwise, at most one agent prefers $\pi^*$, while at least two prefer $\pi$ (since $S$ contains a $T_1$- and a $T_2$-agent). 
This is a contradiction to the popularity of $\pi^*$.
\qed
\end{proof}

\begin{lemma}
\label{lem_ashg2:BBT}
Denote $S=\pi^*(b_1)$. If $b_2\in S$, then we have that $S\cap T_1=S\cap T_2=\emptyset$.
\end{lemma}

\begin{proof}
Assume towards contradiction that $b_2\in S$ and that, without loss of generality, we have that $S\cap T_1\neq\emptyset$.
By \Cref{lem_ashg2:TT_separate}, we have that $S\cap T_2=\emptyset$, and thus for all $t_2\in T_2$ it holds that $\pi^*(t_2)\subseteq T_2$.
Otherwise, every agent in $\pi^*(t_2)$ obtains a negative utility (recall that $b_1,b_2\notin \pi^*(t_2)$) and therefore dissolving $\pi^*(t_2)$ into singletons would be a Pareto improvement.
This also implies that $u_{t_2}(\pi^*)\leq 0$. 
Observe that $S$ contains no agents apart from $b_1$, $b_2$, and $T_1$-agents, as otherwise all members of $S$ obtain a negative utility and it would be a Pareto improvement to dissolve $S$ into singletons. 
Therefore, consider the partition $\pi$ obtained from $\pi^*$ by extracting $b_2$ from $S$, and setting $\pi(b_2)=\{b_2\}\cup T_2$.
This affects the coalitions of no agents other than the agents in $S$ and $T_2$.
We have that $b_1$, $b_2$, and all $T_2$-agents prefer $\pi$ over $\pi^*$, while only $T_1$-agents may prefer $\pi^*$, a contradiction to the popularity of $\pi^*$.
\qed
\end{proof}

\begin{lemma}
\label{lem_ashg2:BB_right}
If $\pi^*(b_1)$ contains a right-side agent, then $\pi^*(b_2)$ does not.
\end{lemma}

\begin{proof}
Assume towards contradiction that both $\pi^*(b_1)$ and $\pi^*(b_2)$ contain a right-side agent. Then all $T_1$- and $T_2$-agents have utility at most $0$. 
Furthermore, for $i\in\{1,2\}$, we have that for any agent $t_i \in T_i$, either $t_i\in\pi^*(b_1)\cup\pi^*(b_2)$, or $\pi^*(t_i)$ contains only $T_i$-agents, as otherwise it is a Pareto improvement to remove $t_i$ from $\pi^*(t_i)$.

Thus, consider the partition $\pi$ obtained from $\pi^*$ by extracting $b_1$ and $b_2$ from their coalitions and setting $\pi(b_1)=\{b_1\}\cup T_1$ and $\pi(b_2)=\{b_2\}\cup T_2$. 
We have that all $T_1$- and $T_2$-agents prefer $\pi$ over $\pi^*$ (which there are $4n+2m$ of). 
Observe that the only agents who may prefer $\pi^*$ are $b_1$, $b_2$, and the real agents (of which there are a total of $4n+m+2$ agents).
Therefore, since $m\geq 2$, it is enough to show that one of those $4n+m+2$ agents prefers $\pi$ to establish that $\pi$ is more popular that $\pi^*$, thus reaching a contradiction. 
Observe that if $\pi^*(b_1)$ or $\pi^*(b_2)$ contains a right-side structure agent, then $b_1$ or $b_2$ prefer $\pi$, and we are done. Hence, we may assume this is not the case, implying $\pi^*(b_1)$ and $\pi^*(b_2)$ both contain real agents.
Similarly, if $b_1\in\pi^*(b_2)$ we reach the same contradiction, so we may assume $b_1\notin\pi^*(b_2)$.
If either $\pi^*(b_1)$ or $\pi^*(b_2)$ contains a $T_1\cup T_2$-agent, then clearly any real agent present in that coalition (of which there exists at least one) prefers $\pi$ over $\pi^*$, and we are done. 
Otherwise, $b_1$ and $b_2$ form coalitions with real agents only.
Hence, at least one of $b_1$ and $b_2$ must prefer $\pi$ since she has at most $2n+\frac{m}{2}$ real agents in her coalition in $\pi^*$. 
Again, we are done.
\qed
\end{proof}

\begin{lemma}
\label{lem_ashg2:TT_right}
Let $t_1\in T_1$ and $t_2\in T_2$. Then $\pi^*(t_1)$ and $\pi^*(t_2)$ contain no right-side agents.
\end{lemma}

\begin{proof}
Assume towards contradiction $\pi^*(t_1)$ contains a right-side agent $\ga$. 
Then $u_{t_1}(\pi^*)<0$ and $u_{\ga}(\pi^*)<0$. 
Furthermore, by \Cref{lem_ashg2:BB_right}, only one of $b_1$ and $b_2$ may be in $\pi^*(t_1)$. 
Consider the partition $\pi$ obtained when removing $t_1$ from $\pi^*(t_1)$ and forming a singleton $\pi(t_1)=\{t_1\}$: 
All agents in $\pi^*(t_1)$ except possibly one of $b_1$ and $b_2$ prefer $\pi$ over $\pi^*$.
Since there are at least two agents (namely $t_1$ and $\ga$) preferring $\pi$, we obtain a contradiction.
The proof for $\pi^*(t_2)$ is analogous.
\qed
\end{proof}

\begin{lemma}
\label{lem_ashg2:BB_separated}
Denote $S=\pi^*(b_1)$. Then it holds that $b_2\notin S$.
\end{lemma}

\begin{proof}
Assume otherwise. By \Cref{lem_ashg2:BB_right,lem_ashg2:BBT}, $S$ contains no right-side agent and no $T_1$- or $T_2$-agents. Hence, we have $S=\{b_1,b_2\}$, and so it is clearly a Pareto improvement to dissolve $S$ into singletons, a contradiction.
\qed
\end{proof}

\begin{lemma}
\label{lem_ashg2:right_side_arrangement}
If $\pi^*(b_1)$ (resp., $\pi^*(b_2)$) contains a right-side agent, then either $\pi^*(b_2)=\{b_2\}\cup T_1$ or $\pi^*(b_2)=\{b_2\}\cup T_2$ (resp., $\pi^*(b_1)=\{b_1\}\cup T_1$ or $\pi^*(b_1)=\{b_1\}\cup T_2$).
If $\pi^*(b_1)$ and $\pi^*(b_2)$ contain no right-side agents, then either $\{\{b_1\}\cup T_1,\{b_2\}\cup T_2\}\subseteq\pi^*$ or $\{\{b_2\}\cup T_1,\{b_1\}\cup T_2\}\subseteq\pi^*$.
\end{lemma}

\begin{proof}
Suppose without loss of generality that $\pi^*(b_1)$ contains a right-side agent. Then by \Cref{lem_ashg2:BB_right}, $\pi^*(b_2)$ contains no right-side agent.
Moreover, by \Cref{lem_ashg2:TT_right}, the coalitions of all $T_1$- or $T_2$-agents contain no right-side agents. 
Furthermore, by \Cref{lem_ashg2:TT_separate}, $T_1$- and $T_2$- agents are in different coalitions of $\pi^*$. 
Hence, it is simple to verify that if none of the suggested coalitions $\pi^*(b_2)$ exists, then constructing one of them must be a Pareto improvement.

Now, suppose that $\pi^*(b_1)$ and $\pi^*(b_2)$ contain no right-side agents. 
Using \Cref{lem_ashg2:TT_separate,lem_ashg2:BB_separated}, it is easy to verify that the suggested sub-partitions of $\{b_1,b_2\}\cup T_1\cup T_2$ are more popular than any other.
\qed
\end{proof}

We are ready to complete our first main step of proving that coalitions can contain only left-side or only right-side agents.

\begin{lemma}
\label{lem_ashg2:left_right}
Let $S\in\pi^*$ be a coalition. If $S$ contains a left-side agent, it cannot contain a right-side agent.
\end{lemma}

\begin{proof} 
Assume that $S$ contains both a left-side and a right-side agent. By \Cref{lem_ashg2:BB_right,lem_ashg2:TT_right}, $S$ contains only one left-side agent, which must be $b_j$, where $j\in\{1,2\}$. 
If $j = 1$, set $i = 2$ and if $j = 2$, set $i = 1$.
By \Cref{lem_ashg2:right_side_arrangement}, the remaining left-side agents are partitioned as follows: either $\{b_i\}\cup T_1\in\pi^*$ or $\{b_i\}\cup T_2\in\pi^*$.
Without loss of generality, assume that $\{b_i\}\cup T_1\in\pi^*$. 
Observe that if $S$ contains no real agents, then it is a Pareto improvement to remove $b_j$ from $S$. 
Hence, $S$ contains some real agent $\ga$. 
We make a case distinction according to the identity of $b_j$.

\textit{Case 1:} Suppose $b_1\in S$ (i.e., $j=1$ and $i=2$). Then consider the partition $\pi$ obtained from $\pi^*$ by extracting $b_1$ and $b_2$ from their respective coalitions, and setting $\pi(b_1)=\{b_1\}\cup T_2$ and $\pi(b_2)=\{b_2\}\cup (S\setminus \{b_1\})$.
All $T_2$-agents clearly prefer $\pi$ over $\pi^*$ while all $T_1$-agents prefer $\pi^*$ (so those cancel out in terms of popularity).
Any structure agent in $S$ is indifferent between the partitions, since they value $b_1$ and $b_2$ equally. 
Every real agent in $S$ (of which there is at least one of) prefers $\pi$ over $\pi^*$, having a greater value assigned to $b_2$ than to $b_1$. 
Furthermore, we must have $\phi_{\{b_1,b_2\}}(\pi^*,\pi)=0$ because the valuation functions of $b_1$ and $b_2$ are identical for all agents except themselves.
Hence, $\pi$ is more popular than $\pi^*$, a contradiction.

\textit{Case 2:} Suppose $b_2\in S$ (i.e., $j=2$ and $i=1$). Then consider the partition $\pi$ obtained from $\pi^*$ by extracting $b_1$ and $b_2$ from their respective coalitions, and setting $\pi(b_1)=\{b_1\}\cup T_2$ and $\pi(b_2)=\{b_2\}\cup T_1$. All $T_1$- and $T_2$-agents clearly prefer $\pi$ over $\pi^*$. Agent $b_1$ is indifferent between $\pi^*$ and $\pi$. Any structure agent in $S$ prefers $\pi$ (having removed $b_2$ from $S$). Even if $b_2$ and all real players prefer $\pi^*$ over $\pi$, we have that at least $4n+2m$ agents prefer $\pi$ whereas at most $4n + m + 1$ agents prefer $\pi^*$.
Hence, $\pi$ is more popular (recall that $m\geq 2$), a contradiction.
\qed
\end{proof}

We are ready to determine the exact coalitions in $\pi^*$.

\begin{lemma} 
\label{lem_ashg2:C_coalition}
It holds that $C\cup C'\in\pi^*$.
\end{lemma}

\begin{proof}
By \Cref{lem_ashg2:left_right}, all $C$-agents have utility at most $0$ in $\pi^*$.
In addition, no agent apart from $C'$-agents assigns positive value to $C$-agents.
Hence if $C\cup C'\notin \pi^*$, then removing all agents in $C\cup C'$ from their coalitions and forming $C\cup C'$ would be a Pareto improvement.
This contradicts Pareto optimality of $\pi^*$.
\qed
\end{proof}

\begin{lemma} 
\label{lem_ashg2:Y_coalition}
Fix agent $\y\in Y$. Then $\pi^*(\y)=\{\y,\yp\}$.
\end{lemma}

\begin{proof} 
By \Cref{lem_ashg2:left_right}, agent $\yp$ is the only agent to whom $\y$ assigns a positive value, and vice versa.
Furthermore, $\yp$ assigns value $-\infty$ to any player apart from $\y$. 
Hence, if $\{\y,\yp\}\notin \pi^*$, then it is a Pareto improvement to form $\{\y,\yp\}$. 
We obtain a contradiction to the Pareto optimality of $\pi^*$.
\qed
\end{proof}

\begin{lemma} 
\label{lem_ashg2:X_coalition}
Fix $x\in\XX$. Then we have that either $\{\{\x,\xt\},\{\nx,\xf\}\}\subseteq\pi^*$ or $\{\{\x,\xf\},\{\nx,\xt\}\}\subseteq\pi^*$.
\end{lemma}

\begin{proof} 
By \Cref{lem_ashg2:C_coalition,lem_ashg2:Y_coalition,lem_ashg2:left_right}, the only coalitions of right-side agent that were not characterized, yet, are $X$-, $X_t$-, and $X_f$-agents. 
If a coalition contains noncorresponding agents, then splitting the coalition is a Pareto improvement.
Hence, agents in $\{\x,\nx,\xt,\xf\}$ form coalitions among themselves.

Of these, every coalition of size at least $3$ can contain at most one agent with nonnegative utility.
Hence, dissolving this coalition into singletons is more popular.
Moreover, $\{\x,\nx\},\{\xt,\xf\}\notin \pi^*$ as dissolving these coalitions would yield a Pareto improvement.
Finally, in any situation different from the two possibilities detailed in the lemma, we would have a singleton $X$-agent and a corresponding singleton $X_t$- or $X_f$-agent, and merging them would yield a Pareto improvement.
\qed
\end{proof}

We now have enough information about the structure of $\pi^*$ to deduce a satisfying assignment to the $\QSAT$ instance from it.
First, we summarize what we know about $\pi^*$.

\begin{itemize}
    \item For the left side, we have that
    \begin{itemize}
        \item $\{b_1\}\cup T_1\in\pi^*$ and $\{b_2\}\cup T_2\in\pi^*$ or 
        \item $\{b_2\}\cup T_1\in\pi^*$ and $\{b_1\}\cup T_2\in\pi^*$ (\Cref{lem_ashg2:right_side_arrangement,lem_ashg2:left_right}).
    \end{itemize}
    \item We have that $C\cup C'\in\pi^*$ (\Cref{lem_ashg2:C_coalition}).
    \item For $\y\in Y$, we have that $\{\y,\yp\}\in\pi^*$ (\Cref{lem_ashg2:Y_coalition}).
    \item For $x\in\XX$, we have that
\begin{itemize}
    \item $\{\x,\xt\}\in\pi^*$ and $\{\nx,\xf\}\in\pi^*$, or 
    \item $\{\x,\xf\}\in\pi^*$ and $\{\nx,\xt\}\in\pi^*$ (\Cref{lem_ashg2:X_coalition}).
\end{itemize}
\end{itemize}

This allows us to define the following truth assignment $\tau_{\XX}$ to the $\XX$ variables. 
For each $x\in \XX$, $x$ is assigned \Tru{} if and only if $\pi^*(\x)=\{\x,\xt\}$ (by \Cref{lem_ashg2:X_coalition}, this is a valid assignment). 
We claim that $\tau_{\XX}$ is a satisfying assignment to the $\QSAT$ instance, i.e., that $\psi(\tau_{\XX},\tau_{\YY}) = \Tru$ for all truth assignments $\tau_{\YY}$ to the $\YY$ variables.

Assume otherwise, namely that there exists a truth assignment $\tau_{\YY}$ to the $\YY$ variables such that $\psi(\tau_{\XX},\tau_{\YY})=\Fals$. 
We will show that this allows us to find a partition that is more popular than $\pi^*$. Recalling \Cref{lem_ashg2:right_side_arrangement}, let us assume without loss of generality that $\{\{b_1\}\cup T_1,\{b_2\}\cup T_2\}\subseteq\pi^*$.
Consider the partition $\pi$ obtained from $\pi^*$ as follows.
\begin{itemize}
    \item Extract the following agents from their respective coalitions, and place them all together in a new coalition $S$:
    \begin{itemize}
        \item All $\xl\in X$ such that $\{\xl,\xt\}\in \pi^*$, for some $\xt\in\Xt$.
        \item All $\y\in Y$ such that the literal represented by $\y$ is assigned \Tru{} by $\tau_{\YY}$.
        \item All $C$-agents.
        \item Agent $b_1$.
    \end{itemize}
    \item Extract $b_2$ from her coalition, and set $\pi(b_2)=\{b_2\}\cup T_1$.    
\end{itemize}
Note that the new coalition $S$ consists of $2n+m+1$ agents.
Moreover, by definition of $\tau_{\XX}$, if $\tau_{\XX}$ assigns \Tru{} to $x$, then $S$ contains $\x$ and if $\tau_{\XX}$ assigns \Fals{} to $x$, then $S$ contains $\nx$.
In addition, for $y\in \YY$, $S$ contains $\yfix$ if $\tau_{\YY}$ assigns \Tru{} to $y$ and $S$ contains $\nyfix$ if $\tau_{\YY}$ assigns \Fals{} to $y$.

We compute the popularity margin between $\pi$ and $\pi^*$.

\begin{itemize}
	\item Fix $x\in\XX$, and let $\alpha\in\{x,\neg x\}$ such that $\xl\in X\cap S$. We have $u_{\xl}(\pi)=1>\frac{1}2=u_{\xl}(\pi^*)$, and $u_{\xt}(\pi)=0<1=u_{\xt}(\pi^*)$. Hence $\phi_{\{\xl,\xt\}}(\pi^*,\pi)=0$.
	\item Let $\y\in Y\cap S$. We have $u_{\y}(\pi)=1>\frac{1}2=u_{\y}(\pi^*)$, and $u_{\yp}(\pi)=0<1=u_{\yp}(\pi^*)$. Hence $\phi_{\{\y,\yp,\ypp\}}(\pi^*,\pi)=0$.
	\item Let $c\in\CC$. Since $\psi(\tau_{\XX},\tau_{\YY})=\Fals$, we have that $c$ has at most two literals in $S$ assigned \Tru{} by $\tau_{\XX}$ and $\tau_{\YY}$. Hence, since the $X$- and $Y$-agents in $S$ correspond to the literals assigned \Tru{} by $\tau_{\XX}$ and $\tau_{\YY}$, there are at most two $X$- or $Y$-agents in $S$ to whom $\ca$ assigns value $-2$ (to the other $X$- or $Y$-agents she assigns $0$). 
	By contrast, for all $c'\in C'$ we have $u_{c'}(\pi)=0<m=u_{c'}(\pi^*)$. 
	Therefore, $\phi_{C\cup C'}(\pi^*,\pi)=-1$ (since there are $m-1$ $C'$-agents and $m$ $C$-agents). 
	\item For all $t_1\in T_1$ we have $u_{t_1}(\pi)=2>1=u_{t_1}(\pi^*)$, and for all $t_2\in T_2$ we have $u_{t_2}(\pi)=0<2=u_{t_2}(\pi^*)$. Thus, we have $\phi_{T_1\cup T_2}(\pi^*,\pi)=0$.
        \item Agents $b_1$ and $b_2$ are indifferent, since in both partitions their coalitions contain the same amount of agents who they value positively (which all contribute equally to their utility), and no agents who they value negatively.
	\item All other agents are in the same coalition in $\pi$ and $\pi^*$ and they are therefore indifferent between the two partitions. 
\end{itemize}

	Altogether, we conclude that $\phi(\pi^*,\pi)=-1$, in contradiction to $\pi^*$ being a popular partition.
	Hence, $(\XX,\YY,\psi)$ is a Yes-instance of \QSAT.
\qed

\section{Detailed Proof of \Cref{thm_fhg}}
\label{apndx_fhg}
In this section, we provide the proof of \Cref{thm_fhg}.
The proof is based on the reduction from \QSAT{} constructed in \Cref{sec_fhg:setup}. 

To complete the description of the reduced instance, we now specify the exact valuation functions of the agents. 
A summary of the agents' valuations can be seen in \Cref{table_FHG}.
Whenever a value assigned by some agent $\ga$ to some agent $\gap$ is not specified, it means $\ga$ assigns a value of $0$ to $\gap$.

\begin{itemize}
    \item For any variable $x\in\XX$:
    \begin{itemize}
        \item $\x$ and $\nx$ both assign value $v_X$ to $\xt$, and $\frac{9}{10}v_X$ to $\xf$ and $\xfp$, where $v_X=2\frac{4n+m+1}{4n+m+2}$. 
        Furthermore, they assign 1 to any real agent, except to their counterpart (to whom they assign value $0$). 
        Lastly, they assign value $2$ to~$j$.
        \item $\xt,\xf$ and $\xfp$ assign value $1$ to $\x$ and $\nx$. In addition, $\xf$ and $\xfp$ assign value $1$ to each other.
    \end{itemize}
    
    \item For any $Y$-agent $\y$:
    \begin{itemize}
        \item Agent $\y$ assigns value $\frac{3}{2}$ to $\yp$ and $\ypp$. Furthermore, she assigns value $1$ to any real agent, except for her counterpart $\ny$ (to whom she assigns value $0$). Lastly, she assigns value $2$ to $j$.
        \item $\yp$ and $\ypp$ assign value $2$ to each other; furthermore, they assign value $1$ to $\y$.
    \end{itemize}
    
    \item For any clause $c\in\CC$:
    \begin{itemize}
        \item Agent $\ca$ assigns value $0$ to any $X$-agent or $Y$-agent corresponding to a literal that was in the clause $c$; to any other $X$- or $Y$-agent, she assigns $v_C=\frac{2n+1}{2n-2.5}$; to all $C$-agents and to $j$, she assigns value $1$; and lastly, to her corresponding $L_6$-agent $\la^6$, she assigns value $2$. 
        \item For all $i\in\{1, \dots,5\}$, $\la^i$ assigns value $1$ to $\ra$.
        \item Agent $\la^6$ assigns value $1$ to $\ra$, and value $\frac{3}{4}$ to $\ca$.
        \item Agent $\ra$ assigns value $1$ to her six corresponding leaves $\la^1\dots\la^6$.
    \end{itemize}
    \item $j$ assigns value $1$ to any real agent, and value $v_J=\frac{3(2n+m-1)}{2(2n+m)}$ to both $j'$ and $j''$.
    \item $j'$ and $j''$ assign value $2$ to each other; furthermore, they assign value $1$ to $j$.
\end{itemize}

\begin{table}[t]
\centering
\caption{Summary of values of FHG reduction (assigned by row agents to column agents). When two values $v_1/v_2$ appear, $v_1$ refers to corresponding agents, and $v_2$ to noncorresponding.}
\label{table_FHG}
\setlength{\tabcolsep}{6pt}
\renewcommand{\arraystretch}{1.1}
\begin{tabular}{ c c c c c c c c c c c c c c c } 
\hline
& \textbf{$\xl$} & \textbf{$\y$} & \textbf{$\ca$} & \textbf{$j$} & \textbf{$j'$} & \textbf{$j''$} & \textbf{$\xt$} & \textbf{$\xf$} & \textbf{$\xfp$} & \textbf{$\yp$} & \textbf{$\ypp$} & \textbf{$\la^{1\to 5}$} & \textbf{$\la^6$} & \textbf{$\ra$} \\
\hline
$\xl$ & $0/1$ & $1$ & $1$ & $2$ & $0$ & $0$ & $v_X/0$ & $\frac{9}{10}v_X/0$ & $\frac{9}{10}v_X/0$ & $0$ & $0$ & $0$ & $0$ & $0$ \\ 
 
$\y$ & $1$ & $0/1$ & $1$ & $2$ & $0$ & $0$ & $0$ & $0$ & $0$ & $\frac{3}{2}/0$ & $\frac{3}{2}/0$ & $0$ & $0$ & $0$ \\ 
 
$\ca$ & $0/v_C$ & $0/v_C$ & $1$ & $1$ & $0$ & $0$ & $0$ & $0$ & $0$ & $0$ & $0$ & $0$ & $2/0$ & $0$ \\ 
 
$j$ & $1$ & $1$ & $1$ & $0$ & $v_J$ & $v_J$ & $0$ & $0$ & $0$ & $0$ & $0$ & $0$ & $0$ & $0$ \\ 

$j'$ & $0$ & $0$ & $0$ & $1$ & $0$ & $2$ & $0$ & $0$ & $0$ & $0$ & $0$ & $0$ & $0$ & $0$ \\ 

$j''$ & $0$ & $0$ & $0$ & $1$ & $2$ & $0$ & $0$ & $0$ & $0$ & $0$ & $0$ & $0$ & $0$ & $0$ \\ 

$\xt$ & $1/0$ & $0$ & $0$ & $0$ & $0$ & $0$ & $0$ & $0$ & $0$ & $0$ & $0$ & $0$ & $0$ & $0$ \\ 

$\xf$ & $1/0$ & $0$ & $0$ & $0$ & $0$ & $0$ & $0$ & $0$ & $1/0$ & $0$ & $0$ & $0$ & $0$ & $0$ \\

$\xfp$ & $1/0$ & $0$ & $0$ & $0$ & $0$ & $0$ & $0$ & $1/0$ & $0$ & $0$ & $0$ & $0$ & $0$ & $0$ \\

$\yp$ & $0$ & $1/0$ & $0$ & $0$ & $0$ & $0$ & $0$ & $0$ & $0$ & $0$ & $2/0$ & $0$ & $0$ & $0$ \\ 

$\ypp$ & $0$ & $1/0$ & $0$ & $0$ & $0$ & $0$ & $0$ & $0$ & $0$ & $2/0$ & $0$ & $0$ & $0$ & $0$ \\ 

$\la^{1\to 5}$ & $0$ & $0$ & $0$ & $0$ & $0$ & $0$ & $0$ & $0$ & $0$ & $0$ & $0$ & $0$ & $0$ & $1/0$ \\ 

$\la^6$ & $0$ & $0$ & $\frac{3}{4}/0$ & $0$ & $0$ & $0$ & $0$ & $0$ & $0$ & $0$ & $0$ & $0$ & $0$ & $1/0$ \\ 

$\ra$ & $0$ & $0$ & $0$ & $0$ & $0$ & $0$ & $0$ & $0$ & $0$ & $0$ & $0$ & $1/0$ & $1/0$ & $0$ \\ 
\hline
\end{tabular}
\end{table}
In the subsequent two sections, we will prove that satisfiability of the source instance implies the existence of a popular partition and vice versa.

Before we proceed, we start with a general proposition that is useful for understanding the structure of the preferences in FHGs.
It states that an agent prefers a coalition containing an additional agent if and only if this agent yields a value greater than her utility for the current coalition.
This will be useful when computing popularity margins. 
We note that this proposition will usually be used implicitly, and therefore an intuitive understanding of it may be helpful throughout the proof.

\begin{proposition}\label{prop:FHGpreferences}
	Let $(N,v)$ be an FHG, $C\subseteq N$, $i\in C$, and $j\in N\setminus C$.
	Then it holds that $u_i(C\cup\{j\})\ge u_i(C)$ if and only if $v_i(j)\ge u_i(C)$.
\end{proposition}

\begin{proof}
	It holds that $u_i(C\cup\{j\}) = \frac{u_i(C)|C|+v_i(j)}{|C|+1}$.
	Hence, $u_i(C\cup\{j\})\ge u_i(C)$ if and only if $u_i(C)|C|+v_i(j)\ge (|C|+1)u_i(C)$.
	This holds if and only if $v_i(j)\ge u_i(C)$.
    \qed
\end{proof}

We proceed with the consideration of FHGs with the goal of proving \Cref{thm_fhg}.

\subsection{Satisfiability Implies Popular Partition}
\label{apndx_fhg:if_satisfiable}
Consider an instance $(\XX,\YY,\psi)$ of \QSAT{} and its reduced FHG as described in the previous section.
Throughout this section, we assume there is a truth assignment $\tau_{\XX}$ to $\XX$ such that for all truth assignments $\tau_{\YY}$ to $\YY$ it holds that $\psi(\tau_{\XX},\tau_{\YY})=\Tru$. 
Consider the following partition of the agents, denoted by $\pi^*$:
\begin{itemize}
    \item For each clause $c\in\CC$, $\{\ca,\la^6\}\in\pi^*$ and $\{\ra,\la^1,\la^2,\la^3\}\in\pi^*$.
    \item All remaining $L$-agents (i.e., the $L_4$ and $L_5$-agents) form singleton coalitions.
    \item Agents $j$, $j'$, and $j''$ form a coalition together.
    \item Each $Y$-agent $\y$ forms a coalition with her corresponding agents $\yp\in Y'$ and $\ypp\in Y''$.
    \item For each variable $x\in \XX$, if $x$ is assigned \Tru{} by $\tau_{\XX}$ then $\{\{\x,\xt\},\{\nx,\xf,\xfp\}\}\subseteq\pi^*$, and if $x$ is assigned \Fals{} by $\tau_{\XX}$ then $\{\{\nx,\xt\},\{\x,\xf,\xfp\}\}\subseteq\pi^*$.
\end{itemize}

Our goal is to show $\pi^*$ is a popular partition. To that end, assume towards contradiction there exists a partition $\pi$ that is more popular than $\pi^*$. By \Cref{prop:wlogPO}, we may assume that $\pi$ is Pareto-optimal (note that, however, $\pi$ need not be popular).
We begin by proving the following three lemmas regarding the popularity margin on certain agent sets. 

\begin{lemma}
\label{lem_fhg1:CLR_balance}
Let $c\in \CC$. Then it holds that $\phi_{\{\ca,\ra,\la^1, \dots,\la^6\}}(\pi^*,\pi)\geq 0$.
\end{lemma}

\begin{proof}
We make a case distinction based on whether $\la^6$ is in $\pi(\ra)$ or not.

\textit{Case 1:} Suppose $\la^6\notin\pi(\ra)$. 
In this case, it is easy to verify that we cannot rearrange the coalitions containing $\ra$ and $\la^1, \dots, \la^5$ in a more popular way than in $\pi^*$; namely, $\phi_{\{\ra, \la^1, \dots, \la^5\}}(\pi^*,\pi)\geq 0$.
Furthermore, the only agent apart from $\ra$ that $\la^6$ has a positive value for is $\ca$, and therefore her maximal utility in $\pi$ is~$\frac{3}{8}$ (which is what she obtains in $\pi^*$ as well). 
Hence, if $\pi(\la^6)=\pi^*(\la^6)=\{\la^6,\ca\}$ then $\la^6$ is clearly indifferent between $\pi$ and $\pi^*$, but in any other case $\la^6$ prefers $\pi^*$. 
Therefore, if $\ca$ prefers $\pi$, $\la^6$ prefers $\pi^*$, implying $\phi_{\{\la^6,\ca\}}(\pi^*,\pi)\geq 0$, and we are done.

\textit{Case 2:} Suppose $\la^6\in\pi(\ra)$. Observe that if $\pi(\ra)$ contains 4 or more leaves, then $|\pi(\ra)|>|\pi^*(\ra)|$ and so $\la^1$, $\la^2$, and $\la^3$ all prefer $\pi^*$ over $\pi$ (if they are not in $\pi(\ra)$ they get utility $0$ in $\pi$). 
In addition, $\la^6$ prefers $\pi^*$, since she gains utility at most $\frac{7}{24}$ being in $\pi(\ra)$ with 3 additional leaves (even if additionally $\ca\in\pi(\ra)$). 
So even if $\la^4$, $\la^5$, and $\ca$ prefer $\pi$, the statement of the lemma holds. 

If $\pi(\ra)$ contains at most $2$ leaves (one of which is $\la^6$), it is easy to see that $\phi_{\{\ra, \la^1, \dots, \la^5\}}(\pi^*,\pi)\geq 2$, so we are done even if $\ca$ and $\la^6$ prefer $\pi$.

Lastly, suppose $\pi(\ra)$ contains $3$ leaves (one of which is $\la^6$). If $\ca\in\pi(\ra)$, then $\ra, \la^1, \la^2$, and $\la^3$ prefer $\pi^*$, and we are done. 
Otherwise, the utility of $\la^6$ in $\pi$ is at most $\frac{1}{4}$, which implies she prefers $\pi^*$. 
Moreover, it is simple to check that $\phi_{\{\ra, \la^1, \dots, \la^5\}}(\pi^*,\pi)\geq 1$ in this case (by dividing into cases according to which leaves are in $\pi(\ra)$). 
Hence, even if $\ca$ prefers $\pi$ over $\pi^*$, we obtain the desired result.
\qed
\end{proof}

\begin{lemma}
\label{lem_fhg1:X_balance}
Let $x\in \XX$. Then it holds that $\phi_{\{a_x,a_{\lnot x},\xt,\xf,\xfp\}}(\pi^*,\pi)\geq 0$.
\end{lemma}

\begin{proof}
Assume towards contradiction that 
\begin{equation}
\label{eq_lem_fhg1:X_balance}
\phi_{\{\x,\nx,\xt,\xf,\xfp\}}(\pi^*,\pi)<0\text.
\end{equation}
Without loss of generality, assume that $\pi^*(\x)=\{\x,\xt\}$ and $\pi^*(\nx)=\{\nx,\xf,\xfp\}$. 
Observe that the only way in which $\xf$ or $\xfp$ prefers $\pi$ is if $\pi(\xf)=\{\xf,\xfp,\x,\nx\}$.
In this case, it holds that $u_{\xt}(\pi)=0$ and $u_{\x}(\pi)=u_{\nx}(\pi)=\frac{9}{20}v_X$. 
Hence, $\x$, $\nx$, and $\xt$ prefer $\pi^*$, contradicting Inequality~\ref{eq_lem_fhg1:X_balance}. 
Thus, $\xf$ and $\xfp$ are either indifferent or prefer $\pi^*$. 
We consider two cases accordingly.

\textit{Case 1:} Assume that $\xf$ or $\xfp$ are indifferent between $\pi$ and $\pi^*$. Without loss of generality assume that $\xf$ is indifferent, i.e., $u_{\xf}(\pi)=\frac{2}{3}$. 
Then it must be that $\pi(\xf)\in \{S_1, S_2, S_3\}$ where $S_1=\{\xf,\xfp,\nx\}$, $S_2=\{\xf,\xfp,\x\}$, and $S_3=\{\xf,\x,\nx\}$. It is easy to verify that for all three of these possible coalitions, no matter how we partition the remaining two agents, we get a contradiction to Inequality~\ref{eq_lem_fhg1:X_balance}.

\textit{Case 2:} Assume that $\xf$ and $\xfp$ both prefer $\pi^*$ over $\pi$. Then $\x$, $\nx$, and $\xt$ must all prefer $\pi$ in order for Inequality~\ref{eq_lem_fhg1:X_balance} to hold. 
However, in order for $\xt$ to prefer $\pi$, it must be that $\pi(\xt)=\{\xt,\x,\nx\}$, which implies that $\x$ and $\nx$ prefer $\pi^*$. This is again a contradiction.
\qed\end{proof}

\begin{lemma}
\label{lem_fhg1:Y_balance}
Fix agent $\y\in Y$. 
Then it holds that $\phi_{\{\y,\yp,\ypp\}}(\pi^*,\pi)\geq 0$.
\end{lemma}

\begin{proof}
Assume towards contradiction that 
\begin{equation}
\label{eq_lem_fhg1:Y_balance}
\phi_{\{\y,\yp,\ypp\}}(\pi^*,\pi) < 0\text.
\end{equation}
First, we observe that $\yp$ and $\ypp$ cannot prefer $\pi$ over $\pi^*$, as they already get maximal utility in their coalition in $\pi^*$. 
So it must be that $\y$ prefers $\pi$ over $\pi^*$, while both $\yp$ and $\ypp$ must be indifferent (if one of them prefers $\pi^*$, we already get a contradiction to Inequality~\ref{eq_lem_fhg1:Y_balance}). 
Notice that $\y$ assigns value $2$ only to $j$, $\frac{3}{2}$ to $\yp$ and $\ypp$, and at most $1$ to any other agent. 
Hence, it is easy to verify that the only possibility where $\y$ has utility greater than $1$ is if she is with $j$ and with at least one of $\yp$, $\ypp$, in contradiction to $\yp$ and $\ypp$ being indifferent.
\qed
\end{proof}

So far, we have seen that the popularity margin with respect to $(\pi^*,\pi)$ for all agents apart from $j$, $j'$, and $j''$ is at least $0$. 
It is clear that $j'$ and $j''$ have maximal utility in $\pi^*$, and thus $\phi_{\{j',j''\}}(\pi^*,\pi)\geq 0$. 
Therefore, in order for $\pi$ to be more popular than $\pi^*$, $j$ must prefer $\pi$ over $\pi^*$; furthermore, according to \Cref{lem_fhg1:CLR_balance,lem_fhg1:X_balance,lem_fhg1:Y_balance}, we cannot afford any of the subsets of agents discussed in those lemmas to have a strictly positive popularity margin with respect to $(\pi^*,\pi)$: it must be exactly 0, as otherwise $\pi$ ``loses'' the advantage over $\pi^*$ it gains from $j$. 
By the same reasoning, $j'$ and $j''$ must also be indifferent between $\pi$ and $\pi^*$. 
The following proposition summarizes all of these observations.

\begin{proposition}
\label{prop:fhg1_observations}
The following statements are true.
\begin{enumerate}
    \item $\phi_j(\pi^*,\pi)<0$.\label{prop:fhg1_observations:j}
    \item $\phi_{j'}(\pi^*,\pi)=\phi_{j''}(\pi^*,\pi)=0$.\label{prop:fhg1_observations:j'_j''}
    \item For each $x\in\XX$, it holds that $\phi_{\{\x,\nx,\xt,\xf,\xfp\}}(\pi^*,\pi)=0$.\label{prop:fhg1_observations:X}
    \item For any $\y\in Y$, it holds that $\phi_{\{\y,\yp,\ypp\}}(\pi^*,\pi)=0$.\label{prop:fhg1_observations:Y}
    \item For any $c\in\CC$, it holds that $\phi_{\{\ca,\ra, \la^1, \dots, \la^6\}}(\pi^*,\pi)=0$.\label{prop:fhg1_observations:C}
\end{enumerate}
\end{proposition}

By \Cref{prop:fhg1_observations}(\ref{prop:fhg1_observations:j'_j''}), we have that $\{j',j''\}\subseteq \pi(j'')\subseteq \{j,j',j''\}$.
However, if $j\in\pi(j'')$ we would get a contradiction to \Cref{prop:fhg1_observations}(\ref{prop:fhg1_observations:j}), and we can therefore assume that $j\notin\pi(j'')$.

We now proceed with showing that $\pi(j)$ cannot contain more than $n$ $X$-agents or more than $n$ $Y$-agents; we will even have more specific restrictions as to which $X$- and $Y$-agents can reside in a coalition with $j$, as formalized in the next two lemmas. 

\begin{lemma}
\label{lem_fhg1:xf_j}
Let $x\in\XX$. 
If $\{\nx, \xf, \xfp\}\in \pi^*$ then it holds that $\nx\notin\pi(j)$, and if $\{\x, \xf, \xfp\}\in \pi^*$ then it holds that $\x\notin\pi(j)$.
\end{lemma}

\begin{proof}
Without loss of generality assume that $\{\nx, \xf, \xfp\}\in \pi^*$ (and therefore also $\{\x, \xt\}\in \pi^*$).
Assume towards contradiction that $\nx\in\pi(j)$. 
Observe that $\xf$ and $\xfp$ can only be indifferent between $\pi$ and $\pi^*$ if $\pi(\xf)=\{\xf,\xfp,\x\}$.
In any other case, they would prefer $\pi^*$ (if $\pi(\xf)=\{\xf,\xfp,\x,\nx,j\}$ then $u_{\xf}(\pi)=u_{\xfp}(\pi)=\frac{3}{5}$, whereas $u_{\xf}(\pi^*)=u_{\xfp}(\pi^*)=\frac{2}{3}$; other cases are simpler to rule out).
We make a case distinction to complete the proof.

\textit{Case 1:} If indeed $\pi(\xf)=\{\xf,\xfp,\x\}$, then $u_{\xt}(\pi)<\frac{1}{2}$ and so $\xt$ prefers $\pi^*$ over $\pi$. 
Furthermore, the maximal utility for $\nx$ would be achieved if $\pi(\nx)=\{\nx,\xt,j\}$, and therefore $u_{\nx}(\pi)\leq \frac{2+v_X}{3}<\frac{6}{5}v_X=u_{\nx}(\pi^*)$. Hence, we have a contradiction to \Cref{prop:fhg1_observations}(\ref{prop:fhg1_observations:X}).

\textit{Case 2:} If $\pi(\xf)\neq\{\xf,\xfp,\x\}$, then as mentioned above, $\xf$ and $\xfp$ both prefer $\pi^*$ over $\pi$. 
If $\pi(\xt)=\pi^*(\xt)=\{\xt,\x\}$, then we get a contradiction to \Cref{prop:fhg1_observations}(\ref{prop:fhg1_observations:X}). If not, we again have $u_{\xt}<\frac{1}{2}$ and so $\xt$ prefers $\pi^*$ over $\pi$, which once again contradicts \Cref{prop:fhg1_observations}(\ref{prop:fhg1_observations:X}).
\qed
\end{proof}

\begin{lemma}
\label{lem_fhg1:complementary_Y}
Let $\y\in Y$. Then it holds that $\ny\notin\pi(\y)$.
\end{lemma}

\begin{proof} 
Assume towards contradiction that $\ny\in\pi(\y)$. 
If $\yp$ or $\ypp$ are in $\pi(\y)$, then $\yp$ and $\ypp$ both prefer $\pi^*$ over $\pi$, in contradiction to \Cref{prop:fhg1_observations}(\ref{prop:fhg1_observations:Y}). 
By the same reasoning we have that $\neg \yp,\neg \ypp\notin\pi(\y)$. 
Hence, we get $u_{\y}(\pi)<1=u_{\y}(\pi^*)$: 
indeed, apart from $\yp$ and $\ypp$, $\y$ assigns positive value only to the real agents (except for $\ny$), and $j$.
Thus, since $\ny\in\pi(\y)$, no matter which subset of them are in $\pi(\y)$,  $\y$ will have utility less than $1$. 
Since $\yp$ and $\ypp$ already get maximal utility in $\pi^*$, this gives a contradiction to \Cref{prop:fhg1_observations}(\ref{prop:fhg1_observations:Y}).
\qed
\end{proof}

Now, let us consider $\pi(j)$. As discussed earlier, we have that $j',j''\notin\pi(j)$, and yet by \Cref{prop:fhg1_observations}(\ref{prop:fhg1_observations:j}), we have $u_j(\pi)>u_j(\pi^*)=\frac{2n+m-1}{2n+m}$. Note that $j$ assigns $1$ to all real agents and $0$ to any other agent (excluding $j'$ and $j''$). By \Cref{lem_fhg1:xf_j,lem_fhg1:complementary_Y}, there are at most $n$ $X$-agents and at most $n$ $Y$-agents in $\pi(j)$. 
Hence, in order for $u_j(\pi)>\frac{2n+m-1}{2n+m}$ to hold, we must have that $n$ $X$-agents, $n$ $Y$-agents, all $C$-agents, and no other agent are in $\pi(j)$ (in particular, note that this implies $L$- and $R$- agents have no $C$-agents in their coalitions). 
For each clause $c\in \CC$, let $k_c\in\{0,1,2,3\}$ denote the number of $X$- and $Y$-agents in $\pi(\ca)$ who represent literals from the original clause~$c$. Then we have that:
\begin{equation}
\label{eq_fhg_d1_c_utility}
u_{\ca}(\pi)=\frac{(2n-k_c)\cdot v_C+m}{2n+m+1}=\frac{(2n-k_c)\cdot (\frac{2n+1}{2n-2.5})+m}{2n+m+1}
\end{equation}
Our next lemma shows that all $C$-agents must strictly prefer $\pi$ over $\pi^*$.

\begin{lemma}
\label{lem_fhg1:all_c_prefer_pi}
Let $c\in\CC$. 
Then it holds that $\phi_{\ca}(\pi^*,\pi)<0$.
\end{lemma}

\begin{proof}
First, it is easy to see from \Cref{eq_fhg_d1_c_utility} that we cannot have $\phi_{\ca}(\pi^*,\pi)=0$ (we would need $k_c=2.5$ for that, but $k_c\in\NN$). 
Therefore, assume towards contradiction that $\phi_{\ca}(\pi^*,\pi)>0$, namely $\ca$ prefers $\pi^*$ over~$\pi$. Consider the corresponding agents of $\ca$. We have that $\phi_{\{\ra, \la^1, \dots, \la^6\}}(\pi^*,\pi)\geq 1$. One may verify this by considering separate cases based on which leaves out of $\la^1, \dots, \la^6$ are included in $\pi(\ra)$; 
for this, notice that if $|\pi(\ra)|\geq 4$ then $\la^6$ prefers $\pi^*$ even if $\la^6\in\pi(\ra)$, since $\frac{3}{8}>\frac{1}{4}$.
(Note that we cannot assume that $\ra,\la^1, \dots,\la^6$ do not have any external agents in their coalitions in $\pi$, but clearly that would only strengthen the arguments as it weakens those coalitions).

Hence, we get a contradiction to \Cref{prop:fhg1_observations}(\ref{prop:fhg1_observations:C}).
\qed
\end{proof}

We have now gathered enough knowledge about $\pi$ to contradict the assumption that it is more popular than $\pi^*$. Observe that, for each $c\in\CC$, $u_{\ca}(\pi^*)=1$. Using \Cref{eq_fhg_d1_c_utility}, some basic algebra shows that $u_{\ca}(\pi)>1$ if and only if $k_c<2.5$ (and since $k_c\in \NN$, we have $k_c\leq 2$).
By \Cref{lem_fhg1:all_c_prefer_pi}, this must hold for all $c\in \CC$. 
Consider the following assignment $\tau_{\YY}$ to the $\YY$ variables: variable $y\in \YY$ is assigned \Tru{} if and only if $\yfix\in\pi(j)$ (by \Cref{lem_fhg1:complementary_Y}, this is a valid assignment). By \Cref{lem_fhg1:xf_j}, and since $|\pi(j)\cap X|=n$ (as observed above), the $X$-agents in $\pi(j)$ are exactly the agents $\xl\in X$ for which $\alpha$ is assigned \Tru{} by $\tau_{\XX}$.
Hence, we have that the assignments $\tau_{\XX},\tau_{\YY}$ yield $\psi(\XX,\YY)=\Fals$, since for each clause $c\in \CC$ there exists a literal assigned \Fals{} by $\tau_{\XX},\tau_{\YY}$ (since $k_c\leq 2$). 
This gives a contradiction to our choice of $\tau_{\XX}$. 

Hence, we have completed the proof showing that if $(\XX,\YY,\psi)$ is a Yes-instance of \QSAT{}, then its reduced FHG admits a popular partition.

\subsection{Popular Partition Implies Satisfiability}
\label{apndx_fhg:if_popular}
Throughout this section, we assume that there is a popular partition $\pi^*$ in the constructed FHG.
We will prove that this implies the source instance is a Yes-instance to \QSAT.
We observe that $\pi^*$ must be Pareto-optimal, as otherwise, applying a Pareto improvement will yield a more popular partition, a contradiction. 

We call a coalition $S\subseteq N$ a \textit{null coalition} if $|S|>1$ and all of its members get utility $0$. Observe that if there exist any null coalition $S$ in $\pi^*$, we can construct another popular partition $\pi^{**}$ by dissolving all null coalitions into singletons: 
Clearly, a partition $\pi$ is more popular than $\pi^*$ if and only if it is more popular than $\pi^{**}$. 
Thus, if there exists a popular partition, then there must exist one with no null coalitions.
Hence, without loss of generality, we may assume that $\pi^*$ contains no null coalitions.
We refer to this assumption as the \textit{null-coalition assumption}. 
We will now begin analyzing some properties of $\pi^*$, aiming to eventually derive a satisfying assignment to the \QSAT{} instance.

Our first lemma derives that only agents in singletons can receive a utility of~$0$.
\begin{lemma}
\label{lem_fhg2:0_implies_singleton}
Let $\ga$ be an agent such that $u_{\ga}(\pi^*)=0$. Then $\ga$ is in a singleton coalition.
\end{lemma}

\begin{proof}
Assume towards contradiction that there exists an agent $\ga$ with $u_{\ga}(\pi^*)=0$ and $|\pi^*(\ga)|\ge 2$. 
Observe that for any two agents $\ga_1$ and $\ga_2$, we have that $\ga_1$ assigns a positive value to $\ga_2$ if and only if $\ga_2$ assigns a positive value to $\ga_1$ (simply by our specific game construction). 
Now, since $u_{\ga}(\pi^*)=0$, there is no agent in $\pi^*(\ga)$ who assigns a positive value to $\ga$. 
Furthermore, by the null-coalition assumption, there must exist some agent $\gap\in\pi^*(\ga)$ with a positive utility. 
But then, it is a Pareto improvement to remove $\ga$ from her coalition, as $\gap$ will clearly prefer the resulting partition, while no agent can prefer the original $\pi^*$. 
\qed
\end{proof}

The next lemmas concern the structure of agents in the star gadgets corresponding to clauses.

\begin{lemma}
\label{lem_fhg2:l6_c}
Let $c\in\CC$ be a clause. Then it holds that $\ca\in\pi^*(\la^6)$.
\end{lemma}

\begin{proof}
Assume towards contradiction that $\ca\notin\pi^*(\la^6)$. 
Intuitively, the problem is that $\ra,\la^1, \dots,\la^6$ form a No-instance as discussed earlier, implying that $\pi^*$ is not popular. Formally, let us consider separate cases based on the number of leaves in $\pi^*(\ra)$.

\textit{Case 1:} Suppose that $\pi^*(\ra)$ contains at least $4$ leaves of $\la^1, \dots,\la^6$. 
Then consider the partition $\pi$ which is identical to $\pi^*$ except we remove one of those leaves $\la^i$ from $\pi^*(\ra)$ and make it a singleton $\{\la^i\}$. 
The only agents who may prefer $\pi^*$ over $\pi$ are $\ra$ and $\la^i$, while the other leaves in $\pi^*(\ra)$ (of which there are at least $3$) prefer $\pi$. Hence, we obtain a contradiction to $\pi^*$ being popular. 

\textit{Case 2:} Suppose that $\pi^*(\ra)$ contains at most $3$ leaves of $\la^1, \dots,\la^6$. 
Let $S=\pi^*(\ra)\setminus\{\ra,\la^1, \dots,\la^6\}$. By \Cref{lem_fhg2:0_implies_singleton}, for all $s\in S$ we have $u_s(\pi^*)>0$, implying it is more popular to split $\pi^*(\ra)$ into $S$ and $\pi^*(\ra)\setminus S$.
If $S\neq \emptyset$, this is more popular since no two agents from $S$ and $\pi^*(\ra)\setminus S$ respectively assign positive utility to each other (under the assumption that $\ca\notin\pi^*(\la^6)$). Hence, $S=\emptyset$.

Moreover, if $\la^i\in\{\la^1, \dots,\la^6\}$ was not in $\pi^*(\ra)$, then $u_{\la^i}(\pi^*)=0$, and by \Cref{lem_fhg2:0_implies_singleton}, $\la^i$ is in a singleton coalition.
Now, consider the partition $\pi$ which is identical to $\pi^*$ except that we merge coalitions to form $\{\ra,\la^1, \dots,\la^6\}$.
Observe that $\ra$ prefers $\pi$ over $\pi^*$, as her utility is clearly monotonically increasing in the number of leaves (out of $\la^1, \dots,\la^6$) that are in her coalition. 
Furthermore, if $\la^i\in\{\la^1, \dots,\la^6\}$ was not in $\pi^*(\ra)$, then $u_{\la^i}(\pi)>0 = u_{\la^i}(\pi^*)$ (under the assumption that $\ca\notin\pi^*(\la^6)$).
As this concerns at least three leaf agents, we obtain $\phi(\pi,\pi^*) \ge 1$, a contradiction. 
\qed
\end{proof}

\begin{lemma}
\label{lem_fhg2:LRC}
Let $c\in\CC$. Then it holds that $\ra,\la^1, \dots,\la^5\notin\pi^*(\ca)$.
\end{lemma}

\begin{proof}
We begin with the agents $\la^1, \dots,\la^5$. Assume towards contradiction that, without loss of generality, $\la^1\in \pi^*(\ca)$. 
Since $\la^1$ is not in a singleton coalition, 
\Cref{lem_fhg2:0_implies_singleton} implies that $\ra\in\pi^*(\ca)$. 
We consider two cases.

\textit{Case 1:} Assume at least three of $\la^1, \dots,\la^6$ are in $\pi^*(\ca)$ (including $\la^1$ and $\la^6$, which we already know are in $\pi^*(\ca)$). Then the partition $\pi$ obtained from $\pi^*$ by removing $\la^1$ from $\pi^*(\ca)$ is more popular than $\pi^*$: Only $\ra$ and $\la^1$ may prefer $\pi^*$, while (at least) two other leaves and $\ca$ prefer $\pi$. 

\textit{Case 2:} Assume $\la^1$ and $\la^6$ are the only agents of $\la^1, \dots,\la^6$ in $\pi^*(\ca)$. 
Then the utility of $\la^2,\la^3,\la^4,\la^5$ is $0$, and these are therefore in singleton coalitions (by \Cref{lem_fhg2:0_implies_singleton}). 
Hence, we can construct the partition $\pi$ obtained from $\pi^*$ by extracting $\ra,\la^1,\la^2,\la^3,\la^4,\la^5$ from their coalitions and forming a new coalition consisting only of them. 
Only $\la^1$ and $\la^6$ may prefer $\pi^*$ over $\pi$ (since we removed $\ra$ from their coalition), while $\ca,\ra,\la^2,\la^3,\la^4$ and $\la^5$ all prefer $\pi$. 
This is a contradiction.

So far, we proved that $\la^1, \dots,\la^5\notin\pi^*(\ca)$. 
It remains to show that $\ra\notin\pi^*(\ca)$. 
Assume towards contradiction that $\ra\in \pi^*(\ca)$.
By \Cref{lem_fhg2:0_implies_singleton}, $\la^1, \dots,\la^5$ are in singletons.
Consider the partition $\pi$ obtained from $\pi^*$ by removing $\ra$ from her coalition and forming $\{\ra,\la^1, \dots,\la^5\}$. 
This is preferred by $\la^1, \dots,\la^5$.
Also, $\ra$ prefers $\pi$ over $\pi^*$ since $u_r(\pi^*)\leq\frac{1}{3}$ while $u_r(\pi)=\frac{5}{6}$.
By contrast, 
the only agent that may prefer $\pi^*$ is $\la^6$.
We again obtain a contradiction to the popularity of $\pi^*$. 
\qed
\end{proof}

\begin{lemma}
\label{lem_fhg2:LRC_general}
Let $c, c'\in\CC$. 
Then it holds that $\ra,\la^1, \dots,\la^5\notin\pi^*(\casec)$.
\end{lemma}

\begin{proof}
Consider $S=\pi^*(\casec)\cap \{\ra,\la^1, \dots,\la^5\}$. 
Assume towards contradiction that $S\neq\emptyset$.
By \Cref{lem_fhg2:LRC}, we know that $c'\neq c$.
By \Cref{lem_fhg2:l6_c,lem_fhg2:LRC}, since $S\neq\emptyset$, we have that $\ca,\la^6\notin\pi^*(\casec)$. 
Therefore, agents in $S$ assign value $0$ to any agent in $\pi^*(\casec)\setminus S$, and vice versa. 
Hence, it is clearly a Pareto improvement to split $S$ from $\pi^*(\casec)$, in contradiction to Pareto optimality of $\pi^*$ (the null-coalition assumption ensures that at least one agent in $\pi^*(\casec)$, and by \Cref{lem_fhg2:0_implies_singleton} all of them, prefer $\pi$).
\qed
\end{proof}

\begin{lemma}
\label{lem_fhg2:RL_everyone}
Let $c\in\CC$. 
Fix some agent $\ga\notin \{\ra,\la^1, \dots,\la^5\}$. 
Then $\ra,\la^1, \dots,\la^5\notin\pi^*(\ga)$.
\end{lemma}

\begin{proof}
Let $S=\pi^*(\ga)\cap\{\ra,\la^1, \dots,\la^5\}$, and assume towards contradiction that $S\neq\emptyset$. 
Notice that it is a Pareto improvement to remove $S$ from $\pi^*(\ga)$, obtaining another partition denoted $\pi$: By \Cref{lem_fhg2:l6_c,lem_fhg2:LRC_general}, there are no $L_6$-agents in $\pi^*(\ga)$, and thus all agents in $\pi^*(\ga)\setminus S$ must assign value $0$ to all agents in $S$, and vice versa. 
Therefore, it is a Pareto improvement to split $S$ from $\pi^*(\ga)$, in contradiction to Pareto optimality of $\pi^*$ (again, the null-coalition assumption ensures at least one agent in $\pi^*(\ga)$ prefers $\pi$).
\qed
\end{proof}

\begin{lemma}
\label{lem_fhg2:star_partition}
Let $c\in\CC$. 
Then $\pi^*(\ra)$ consists exactly of $\ra$ and three agents from $\la^1, \dots,\la^5$ (without loss of generality $\la^1,\la^2$ and $\la^3$), while the remaining two leaves ($\la^4$ and $\la^5$) are singletons.
\end{lemma}

\begin{proof}
\Cref{lem_fhg2:RL_everyone} tells us that none of the agents $\ra,\la^1, \dots,\la^5$ can belong to a coalition containing an agent who is not from $\ra,\la^1, \dots,\la^5$. 
Hence, it is easy to verify that if $\pi^*(\ra)$ contained more than three of $\la^1, \dots,\la^5$, it would be more popular to remove one of them; and if it contained less, it would be more popular to form a coalition $\{\ra,\la^1, \dots,\la^5\}$.
\qed
\end{proof}

So far we have derived the structure of the coalitions of the agents in the star gadgets.
We now wish to gain more insight on the structure of the coalitions of the remaining agents.
\begin{lemma}
\label{lem_fhg2:j'_j''}
Fix some agent $\ga\notin\{j,j',j''\}$. Then $j',j''\notin\pi^*(\ga)$. Furthermore, $j'\in\pi^*(j'')$.
\end{lemma}

\begin{proof}
Let $\ga\notin\{j,j',j''\}$.
Assume towards contradiction that, without loss of generality, $j'\in\pi^*(\ga)$. Observe that, regardless of whether or not $j$ or $j''$ are also in $\pi^*(\ga)$, we have that $u_{j'}(\pi^*)\leq \frac{3}{4}$ and $u_{j''}(\pi^*)\leq \frac{3}{4}$. 
Consider the partition $\pi$ obtained from $\pi^*$ by removing $j'$ and $j''$ from their coalitions and forming a new coalition consisting of the two of them. 
Then $u_{j'}(\pi)=u_{j''}(\pi)=1>\frac{3}{4}$, and therefore $j'$ and $j''$ prefer $\pi$ over $\pi^*$.
By contrast, the only agent that may prefer $\pi^*$ is $j$ (since she is the only agent who assigns positive value to $j'$ and $j''$). 
Hence, $\pi$ is more popular than $\pi^*$, a contradiction.

To prove the second part of the statement, we observe that if $j'\notin\pi^*(j'')$ then, according to the first part, $j'$ and $j''$ are in singleton coalitions. Hence, it would be a Pareto improvement to group them together, a contradiction to popularity of $\pi^*$.
\qed
\end{proof}

\begin{lemma}
\label{lem_fhg2:Y'_Y''}
Fix agents $\y\in Y$ and $\ga\notin\{\y,\yp,\ypp\}$. Then $\yp,\ypp\notin\pi^*(\ga)$. Furthermore, $\yp\in\pi^*(\ypp)$.
\end{lemma}

\begin{proof}
The proof is identical to the proof of \Cref{lem_fhg2:j'_j''}, substituting $j$ with $\y$, $j'$ with $\yp$ and $j''$ with $\ypp$.
\qed
\end{proof}

\begin{lemma}
\label{lem_fhg2:y_max}
Fix agent $\y\in Y$. Then $u_{\y}(\pi^*)\leq 1$.
\end{lemma}
\begin{proof}
If $j\in\pi^*(\y)$, then by \Cref{lem_fhg2:Y'_Y''} we have $\yp,\ypp\notin\pi^*(\y)$; hence, since $\y$ assigns value $2$ to $j$ and at most value 1 to all remaining agents, her utility is at most 1. 

If $j\notin\pi^*(\y)$, since $\y$ assigns value $\frac{3}{2}$ to $\yp,\ypp$, and at most value $1$ to all remaining agents, her utility is at most 1.
\qed\end{proof}
 
\begin{lemma}
\label{lem_fhg2:xf_xf'_together}
Let $x\in\XX$. Then $\xfp\in\pi^*(\xf)$.
\end{lemma}
\begin{proof}
Assume towards contradiction that $\xf$ and $\xfp$ are members of different coalitions. 
Clearly if both have utility~0, and therefore are singletons, it is a Pareto improvement to pair them together.
The only agents that they receive positive utility from are $\x$ and $\nx$, so one of $\xf$ and $\xfp$ has to reside in a coalition with one of them.
Without loss of generality, we assume that $\x\in \pi^*(\xf)$. 
We consider two cases.

\textit{Case 1:} Assume that $u_{\xfp}(\pi^*)>0$. Then $\nx\in\pi^*(\xfp)$. 
If $|\pi^*(\xf)|\geq |\pi^*(\xfp)|$, then consider the partition $\pi$ obtained from $\pi^*$ by moving $\xf$ from $\pi^*(\xf)$ to $\pi^*(\xfp)$.
The partition $\pi$ is more popular than $\pi^*$: First note that $v_{\xfp}(\xf)\ge u_{\xf}(\pi^*)$ and $v_{\nx}(\xf)\ge u_{\nx}(\pi^*)$.
Hence, by \Cref{prop:FHGpreferences}, $\xfp$ and $\nx$ prefer $\pi$.
In addition, $\xf$ has a valuation of $1$ for $\xfp$ and $\nx$ in $\pi(\xf)$ whereas she only had a valuation of $1$ for $\x$ in $\pi^*$.
As her coalition size also shrinks (and therefore the number of agents for which she has a valuation of $0$), she also prefers $\pi$.
Moreover, all additional members of $\pi^*(\xf)$ (apart from $\x$) prefer $\pi$.
Since $|\pi^*(\xf)|\geq |\pi^*(\xfp)|$, we conclude that $\pi$ is more popular. 
If $|\pi^*(\xf)|<|\pi^*(\xfp)|$, a symmetric argument shows it is more popular to move $\xfp$ to $\pi^*(\xf)$.

\textit{Case 2:} Assume $u_{\xfp}(\pi^*)=0$. If $\nx\in\pi^*(\xf)$, then if $|\pi^*(\xf)|\leq 5$ it is more popular to include $\xfp$ in this coalition, and if $|\pi^*(\xf)|>5$ it is more popular to extract $\xf$ from it and pair it with $\xfp$. 
So suppose $\nx\notin\pi^*(\xf)$. Then if $|\pi^*(\xf)|\leq 3$ it is more popular to include $\xfp$ in this coalition, and if $|\pi^*(\xf)|>3$ it is more popular to extract $\xf$ from it and pair it with $\xfp$.

In all cases, we have a contradiction to the popularity of $\pi^*$.
\qed
\end{proof}

We now narrow down the structure of possible coalitions of agents of type $\xt$, $\xf$, and $\xfp$ if there exists a corresponding $X$-agent that does not form a coalition with any of them.

\begin{lemma}
\label{lem_fhg2:xt_xf_xf'}
Let $x\in\XX$, and let $\alpha\in\{x,\neg x\}$. If $\xt,\xf,\xfp\notin\pi^*(\xl)$, then exactly one of the following holds:
\begin{enumerate}
    \item $\pi^*(\xt)=\{\xt\}$ and $\pi^*(\xf)=\pi^*(\xfp)=\{\xf,\xfp\}$.\label{lem_fhg2:xt_xf_xf':both_free}
    \item $\xf,\xfp\in\pi^*(\nxl)$ and $\pi^*(\xt)=\{\xt\}$.\label{lem_fhg2:xt_xf_xf':xt_free}
    \item $\xt\in\pi^*(\nxl)$ and $\pi^*(\xf)=\pi^*(\xfp)=\{\xf,\xfp\}$.\label{lem_fhg2:xt_xf_xf':xf_free}
    \item $\xt,\xf,\xfp\in\pi^*(\nxl)$, and thus $u_{\xf}=u_{\xfp}\leq \frac{1}{2}$.\label{lem_fhg2:xt_xf_xf':all_together}
\end{enumerate}
\end{lemma}
\begin{proof}
By \Cref{lem_fhg2:0_implies_singleton}, if $\xt\notin\pi^*(\nxl)$ then $\pi^*(\xt)=\{\xt\}$.
By \Cref{lem_fhg2:xf_xf'_together} and by Pareto optimality, if $\xf,\xfp\notin\pi^*(\nxl)$ then $\pi^*(\xf)=\pi^*(\xfp)=\{\xf,\xfp\}$. Hence, parts 1 through 4 describe the only possibilities.  
\qed
\end{proof}

Using \Cref{lem_fhg2:xt_xf_xf'}, we can see that an $X$-agent who is separated from her corresponding agents always has a ``good'' deviation, resulting in a partition that is liked at least as much as $\pi^*$ by these agents. 
Similarly, we can find a ``good'' deviation for $Y$-agents and $j$ if they are not with their corresponding agents.
We formalize this idea as it plays an important role in the remainder of the proof.

\begin{definition}[Isolated Agent]
\label{def:fhg2_isolated_agent}
 Let $x\in\XX$, and let $\alpha\in\{x,\neg x\}$. 
 If $\xt,\xf,\xfp\notin \pi^*(\xl)$, we say that $\xl$ is \emph{isolated}. 
 Note that $\xl$ and $\nxl$ may be together in the same coalition, and still be considered isolated.
 
Let $\y\in Y$. If $\yp,\ypp\notin\pi^*(\y)$, we say that $\y$ is \emph{isolated}.

If $j',j''\notin \pi^*(j)$, we say that $j$ is \emph{isolated}.
\end{definition}

\begin{definition}[Isolated Deviation]
\label{def:fhg2_isolated_deviations}
Let $x\in\XX$, and let $\alpha\in\{x,\neg x\}$. 
Denote $S=\pi^*(\xl)$, and assume $\xl$ is isolated. 
Then the following deviation from $\pi^*$, resulting in a partition denoted $\pi$, is denoted \emph{isolated deviation}, performed by $\xl$.
\begin{itemize}
    \item If $\nxl\in S$, extract both $\x$ and $\nx$, and set $\pi(\x)=\{\x,\xt\},\pi(\nx)=\{\nx,\xf,\xfp\}$ (note that it holds in this situation $\pi^*(\xt)=\{\xt\}$ and $\pi^*(\xf)=\pi^*(\xfp)=\{\xf,\xfp\}$ by \Cref{lem_fhg2:xt_xf_xf'}(\ref{lem_fhg2:xt_xf_xf':both_free})). 
    \item If $\nxl\notin S$, extract $\xl$ from $S$; if \Cref{lem_fhg2:xt_xf_xf'}(\ref{lem_fhg2:xt_xf_xf':xt_free}) holds, set $\pi(\xl)=\{\xl,\xt\}$; otherwise, set $\pi(\xl)=\{\xl,\xf, \xfp\}$.
\end{itemize}

Let $\y\in Y$, and denote $S=\pi^*(\y)$. Assume $\yp,\ypp\notin S$. Consider the deviation from $\pi^*$, resulting in a partition denoted $\pi$, where we extract $\y$ from $S$ and set $\pi(\y)=\{\y,\yp,\ypp\}$. This deviation is denoted an \emph{isolated deviation}, performed by $\y$.

If $j',j''\notin \pi(j)$, then setting $\pi(j)=\{j,j',j''\}$ is called an \emph{isolated deviation}, performed by $j$.
\end{definition}

\begin{lemma}
\label{lem_fhg2:X_deviation_outcome}
Let $x\in\XX$, and let $\alpha\in\{x,\neg x\}$. If $\xl$ is isolated, and she performs an isolated deviation resulting in some partition $\pi$, then $\phi_{\{\xl,\nxl,\xt,\xf,\xfp\}}(\pi^*,\pi)\leq 0$. 
Furthermore, if $\phi_{\xl}(\pi^*,\pi)\leq 0$, then $\phi_{\{\xl,\nxl,\xt,\xf,\xfp\}}(\pi^*,\pi)<0$.

Lastly, if \Cref{lem_fhg2:xt_xf_xf'}(\ref{lem_fhg2:xt_xf_xf':all_together}) holds, then for any agent $\ga\in\pi^*(\xf) \setminus \{\xl,\nxl,\xt,\xf,\xfp\}$, it holds that $\phi_{\ga}(\pi^*,\pi)=-1$.
\end{lemma}
\begin{proof}
If \Cref{lem_fhg2:xt_xf_xf'}(\ref{lem_fhg2:xt_xf_xf':both_free}), \Cref{lem_fhg2:xt_xf_xf'}(\ref{lem_fhg2:xt_xf_xf':xf_free}), or \Cref{lem_fhg2:xt_xf_xf'}(\ref{lem_fhg2:xt_xf_xf':all_together}) holds, then it is straightforward to verify that $\phi_{\{\xl,\nxl,\xt,\xf,\xfp\}}(\pi^*,\pi)<0$. 
If \Cref{lem_fhg2:xt_xf_xf'}(\ref{lem_fhg2:xt_xf_xf':xt_free}) holds, then, if $\xl$ prefers $\pi^*$, we get $\phi_{\{\xl,\nxl,\xt,\xf,\xfp\}}(\pi^*,\pi)=0$, and otherwise  $\phi_{\{\xl,\nxl,\xt,\xf,\xfp\}}(\pi^*,\pi)<0$.

For the last part, observe that the only agents removed from $\pi^*(\xf)$ are $\xf$ and $\xfp$.
These yield a valuation of zero to all agents except $\xl$ and $\nxl$.
Hence, since other agents had positive utility (by \Cref{lem_fhg2:0_implies_singleton}), any agent $\ga\in\pi^*(\xf) \setminus \{\xl,\nxl,\xt,\xf,\xfp\}$ improves her utility upon the removal of $\xf$ and $\xfp$.
\end{proof}

Note that the only agents not considered in the previous lemma who may additionally be affected by the considered deviation are agents in $\pi^*(\xl)\setminus \{\xl,\nxl\}$.
We will deal with them when we apply the lemma.

\begin{lemma}
\label{lem_fhg2:Y_deviation_outcome}
Let $\y\in Y$. If $\y$ is isolated, and she performs an isolated deviation resulting in some partition $\pi$, then $\phi_{\{\y,\yp,\ypp\}}(\pi^*,\pi)\leq 0$; furthermore, if $\y$ prefers $\pi$ over $\pi^*$, then $\phi_{\{\y,\yp,\ypp\}}(\pi^*,\pi)=-1$.
\end{lemma}
\begin{proof}
The result is given immediately by \Cref{lem_fhg2:Y'_Y'',lem_fhg2:y_max}.
\end{proof}

\begin{lemma}
\label{lem_fhg2:j_deviation_outcome}
If $j$ is isolated, and she performs an isolated deviation resulting in some partition $\pi$, then $j'$ and $j''$ are indifferent between $\pi^*$ and $\pi$.
\end{lemma}
\begin{proof}
By \Cref{lem_fhg2:j'_j''}, it is easy to verify that $u_{j'}(\pi^*)=u_{j''}(\pi^*)=u_{j'}(\pi)=u_{j''}(\pi)=1$.
\end{proof}

We have collected enough insights to precisely determine the coalitions of agents $\ca$.

\begin{lemma}
\label{lem_fhg2:l6_coalition}
Let $c\in\CC$, and denote $S=\pi^*(\ca)$. Then $S=\{\ca,\la^6\}$.
\end{lemma}

\begin{proof}
\Cref{lem_fhg2:l6_c} tells us $\la^6\in S$, so it is left to show no other agent can be in $S$. Note that $\ra\notin S$ by \Cref{lem_fhg2:RL_everyone}, and therefore $\ca$ is the only agent in $S$ to whom $\la^6$ assigns positive value.
Assume towards contradiction that $|S|\geq 3$. 
We make a case distinction. 

\textit{Case 1:} Assume that $|S|\ge 4$, i.e., there are at least two more agents in $S$ apart from $\ca$ and $\la^6$. 
Then $u_{\la^6}(\pi^*)\leq\frac{3}{16}<\frac{1}{4}$. Consider the partition $\pi$ obtained from $\pi^*$ by extracting $\la^6$ from $S$, and having $\la^6$ join the coalition of $\ra$, while $\la^3$ is removed from that coalition and becomes a singleton (recall the arrangement of $\ra,\la^1, \dots,\la^5$ from \Cref{lem_fhg2:star_partition}). 
We have that $\la^3$ prefers $\pi^*$ over $\pi$, and $\la^6$ prefers $\pi$ over $\pi^*$ (as $u_{\la^6}(\pi)=\frac{1}{4}>u_{\la^6}(\pi^*)$), while $\ra,\la^1,\la^2,\la^4$ and $\la^5$ are clearly indifferent between $\pi,\pi^*$. 
The only other agents who were affected by this change are agents in $S\setminus \{\la^6\}$.
Of these, all except $\ca$ prefer $\pi$. 
Since $S$ contains at least two other agents, it follows that $\pi$ is more popular than $\pi^*$, a contradiction. 

\textit{Case 2:} Assume $|S|=3$. i.e., there is just one agent $\ga$ in $S\setminus \{\ca,\la^6\}$. 
By \Cref{lem_fhg2:RL_everyone}, we have that $\ga\notin (\bigcup_{i=1}^5 L_i)\cup R$, and by \Cref{lem_fhg2:j'_j'',lem_fhg2:Y'_Y''}, we have that $\ga\notin Y'\cup Y''\cup\{j',j''\}$. 
Furthermore, since $|S|=3$, $\ga$ cannot be a $C$- or $L_6$-agent, as those always come in pairs (by \Cref{lem_fhg2:l6_c}), which would yield $|S|\geq 4$. 
Hence, $\ga$ must be either an $X$-agent, a $Y$-agent, or $j$; either way, $\ga$ is isolated. 
Consider the partition $\pi$ obtained from $\pi^*$ by having $\ga$ perform an isolated deviation. 
Clearly, $\la^6$ prefers $\pi$ over $\pi^*$, while $\ca$ may prefer $\pi^*$. By \Cref{lem_fhg2:X_deviation_outcome,lem_fhg2:Y_deviation_outcome,lem_fhg2:j_deviation_outcome}, it is enough to show that $\ga$ prefers $\pi$ to establish that $\pi$ is more popular than $\pi^*$. 
Observe that $u_{\ga}(\pi^*)=\frac{1}{3}$ (regardless of whether $\ga$ is $j$, an $X$-agent, or a $Y$-agent).
If $\ga=j$, then $u_{\ga}(\pi)=\frac{2n+m-1}{2n+m}\geq \frac{2}{3}$;
if $\ga\in Y$, then $u_{\ga}(\pi)=1$; 
and if $\ga\in X$, then any of the possible isolated deviations yield $u_{\ga}(\pi)\geq\frac{v_X}{2}\geq\frac{11}{12}$. 
Hence, $\ga$ prefers $\pi$ over $\pi^*$, concluding the proof.
\qed
\end{proof}

The next three lemmas precisely settle the coalition of $j$.

\begin{lemma}
\label{lem_fhg2:Xf_j}
Let $x\in\XX$. Denote $S=\pi^*(j)$. Then $\xf,\xfp\notin S$.
\end{lemma}

\begin{proof}
Assume otherwise. 
Then, by \Cref{lem_fhg2:xf_xf'_together}, it holds that $\{\xf,\xfp\}\subseteq S$.
Observe that $S$ must contain $\x$ or $\nx$, as otherwise it is a Pareto improvement to split $\{\xf,\xfp\}$ from $S$.
Without loss of generality, assume that $\x\in S$. 
We consider the following cases.

\textit{Case 1:} If $|S|\geq 7$, consider the partition $\pi$ obtained from $\pi^*$ by extracting $\xf$ and $\xfp$ from $S$ and setting $\pi(\xf)=\pi(\xfp)=\{\xf,\xfp\}$. 
We have that $u_{\xf}(\pi)=u_{\xfp}(\pi)=\frac{1}{2}$, whereas $u_{\xf}(\pi^*)=u_{\xfp}(\pi^*)\leq \frac{3}{7}$. 
Hence, $\xf$ and $\xfp$ prefer $\pi$ over $\pi^*$. 
Considering the rest of the agents in $S$, at most two of them might prefer $\pi^*$ over~$\pi$ (only $\x$ and $\nx$ assign a positive value to $\xf$ and $\xfp$), 
while the rest must prefer $\pi$, a contradiction to the popularity of $\pi^*$.

\textit{Case 2:} Assume $|S|\leq 6$, namely $S$ contains at most 2 agents in addition to $\xf,\xfp,\x$ and $j$. Consider the partition $\pi$ obtained from $\pi^*$ as follows.
\begin{itemize}
    \item $j$ performs an isolated deviation.
    \item Any $Y$-agent in $S$ performs an isolated deviation.
    \item Any isolated $X$-agent in $S$ performs an isolated deviation.
\end{itemize}
Consider the following agents.
Agent $j$ prefers $\pi$ over $\pi^*$, since $u_j(\pi)=\frac{2n+m-1}{2n+m}>\frac{1}{2}\geq u_j(\pi^*)$.
Agents $\xf$ and $\xfp$ clearly prefer $\pi$ over $\pi^*$.
Agent $\x$ may prefer $\pi^*$ over $\pi$, and so may $\nx$ (if $\nx\in S$).
By \Cref{lem_fhg2:X_deviation_outcome,lem_fhg2:Y_deviation_outcome,lem_fhg2:j_deviation_outcome}, any agent who performed an isolated deviation, along with her corresponding agents, will have a popularity margin of at most 0 with respect to $\pi^*$ and $\pi$. 
Lastly, if an $X$-agent $\ga\neq \xf$ and her corresponding $\Xt$-agent $\gap$ were in $S$, then, since $|S| = 6$, it holds that $S = \{\xf,\xfp,\x,j,\ga,\gap\}$.
Hence, $j$ is the only agent who was removed from $S$, and so $\ga$ and $\x$ prefer $\pi^*$ while $\gap$, $\xf$, $\xfp$, and $j$ prefer $\pi$. 
It follows that $\pi$ is more popular than $\pi^*$, a contradiction.
\qed
\end{proof}

\begin{lemma}
\label{lem_fhg2:Xt_j}
Let $x\in\XX$. Denote $S=\pi^*(j)$. Then $\xt\notin S$.
\end{lemma}

\begin{proof}
Assume otherwise. Observe that $S$ must contain $\x$ or $\nx$, as otherwise $u_{\xt}=0$, in contradiction to \Cref{lem_fhg2:0_implies_singleton}. 
Without loss of generality, assume that $\x\in S$. 
We divide into cases.

\textit{Case 1:} Suppose $\nx\notin S$. If $|S|\geq 5$, consider the partition $\pi$ obtained from $\pi^*$ by removing $\xt$ from $S$ and forming a singleton.
At least 3 members of $S$ prefer $\pi$ over $\pi^*$, and only $\x$ and $\xt$ prefer $\pi^*$, a contradiction. 

If $|S|=3$, then the partition $\pi$ where $j$ performs an isolated deviation is more popular, a contradiction. 

Therefore, $|S|=4$, namely $S=\{\xt,\x,j,\ga\}$ for some agent $\ga\neq\nx$. 
By \Cref{lem_fhg2:0_implies_singleton}, $u_{\ga}(\pi^*)$ must be positive, and so by \Cref{lem_fhg2:l6_coalition,lem_fhg2:xf_xf'_together,lem_fhg2:j'_j''}, $\ga$ can only be an (isolated) $X$- or $Y$-agent. 
Consider the partition $\pi$ obtained from $\pi^*$ where $j$ and $\ga$ perform isolated deviations.
Agent $j$ prefers $\pi$ over $\pi^*$, since $u_j(\pi)=\frac{2n+m-1}{2n+m}>\frac{1}{2}\geq u_j(\pi^*)$.
Agents $j'$ and $j''$ are indifferent between the partitions, by \Cref{lem_fhg2:j_deviation_outcome}.
$\xt$ clearly prefers $\pi$ over $\pi^*$.
Lastly, by \Cref{lem_fhg2:X_deviation_outcome,lem_fhg2:Y_deviation_outcome}, $\ga$ and her corresponding agents will have a popularity margin of at most 0 with respect to $\pi^*$ and $\pi$. 
Hence, $\pi$ is more popular then $\pi^*$, a contradiction.

\textit{Case 2:} Suppose $\nx\in S$. If $|S|\geq 7$, consider the partition $\pi$ obtained from $\pi^*$ by removing $\xt$ from $S$. 
At least 4 members of $S$ prefer $\pi$ over $\pi^*$, and only $\x,\nx$ and $\xt$ prefer $\pi^*$, a contradiction. 
So assume $|S|\leq 6$, namely $S$ contains at most 2 agents in addition to $\xt$, $\x$, $\nx$, and $j$. 
Note that these two agents cannot be a pair of corresponding $X$- and $\Xt$-agents, as that would lead us back to Case 1
Hence, no $\Xt$-agent apart from $\xt$ is in $S$. 
Together with \Cref{lem_fhg2:Xf_j}, this means that every $X$-agent in $S$ is isolated.

Consider the partition $\pi$ obtained from $\pi^*$ as follows.
\begin{itemize}
    \item Agent $j$ performs an isolated deviation.
    \item Agent $\nx$ performs an isolated deviation (which, by definition, results in $\pi(\nx)=\{\nx,\xf,\xfp\}$).
    \item Any additional $X$-agent in $S$ apart from $\x$ performs an isolated deviation.
    \item Any $Y$-agent in $S$ performs an isolated deviation.
\end{itemize}
Agent $j$ prefers $\pi$ over $\pi^*$, since $u_j(\pi)=\frac{2n+m-1}{2n+m}>\frac{1}{2}\geq u_j(\pi^*)$.
By \Cref{lem_fhg2:j_deviation_outcome}, agents $j'$ and $j''$ are indifferent between the partitions.
Agents $\xf$ and $\xfp$ clearly prefer $\pi$ over $\pi^*$.
Agent $\xt$ is either indifferent between $\pi$ and $\pi^*$, or prefers $\pi$, since $u_{\xt}(\pi^*)\leq \frac{1}{2}=u_{\xt}(\pi)$.
Agent $\x$ and $\nx$ may prefer $\pi^*$ over $\pi$.
Lastly, by \Cref{lem_fhg2:X_deviation_outcome,lem_fhg2:Y_deviation_outcome}, any $X$- or $Y$-agent in $S$ and her corresponding agents will have a popularity margin of at most 0 with respect to $\pi^*$ and $\pi$.
Hence, $\pi$ is more popular then $\pi^*$, a contradiction.
\qed
\end{proof} 

\begin{lemma}
\label{lem_fhg2:j_j'_j''}
It holds that $\{j,j',j''\}\in \pi^*$.
\end{lemma}

\begin{proof}
Denote $S=\pi^*(j)$.
Assume towards contradiction that one of $j'$ and $j''$ is not in $S$. 
Then by \Cref{lem_fhg2:j'_j''}, $\pi^*(j')=\pi^*(j'')=\{j',j''\}$. 
Moreover, by \Cref{lem_fhg2:RL_everyone,lem_fhg2:Y'_Y'',lem_fhg2:l6_coalition,lem_fhg2:Xt_j,lem_fhg2:Xf_j}, $S$ can contain only isolated $X$- and $Y$-agents, apart from $j$. 
Consider the partition $\pi$ obtained from $\pi^*$ as follows: 
\begin{itemize}
    \item $j$ performs an isolated deviation.
    \item Any $Y$-agent in $S$ performs an isolated deviation.
    \item Any $X$-agent in $S$ performs an isolated deviation.
\end{itemize}
By \Cref{lem_fhg2:X_deviation_outcome,lem_fhg2:Y_deviation_outcome}, any $X$- and $Y$-agent and her corresponding agents will have a popularity margin of at most 0 with respect to $(\pi^*,\pi)$. 
Hence, if $j$ prefers $\pi$ over $\pi^*$, then $\pi$ is more popular than $\pi^*$, a contradiction. 
So assume $u_j(\pi^*)\geq u_j(\pi)$. Then, since $u_j(\pi^*)=\frac{|S|-1}{|S|}$ and $u_j(\pi)=\frac{2n+m-1}{2n+m}\geq \frac{2n+1}{2n+2}$ (recall $m\geq 2$), we must have that $|S|\geq 2n+2$, i.e., $S$ contains at least $2n+1$ $X$- and $Y$-agents. 
Thus, by the Pigeonhole principle, $S$ must contain at least one pair of complementary agents, either $\y,\ny\in Y$ or $\x,\nx\in X$ in $S$.

If $S$ contains complementary agents $\y,\ny\in Y$, we have $u_{\y}(\pi^*)=u_{\ny}(\pi^*)\leq\frac{4n}{4n+1}<1=u_{\y}(\pi)=u_{\ny}(\pi)$, and so both $\y$ and $\ny$ prefer $\pi$, while their corresponding $Y'$-agents and $Y''$-agents are indifferent. 
Hence, by \Cref{lem_fhg2:X_deviation_outcome,lem_fhg2:Y_deviation_outcome}, even if $j$ prefers $\pi^*$ over $\pi$ we have that $\pi$ is more popular than $\pi^*$, a contradiction. 

If $S$ contains complementary agents $\x,\nx\in X$, we have $u_{\x}(\pi^*)=u_{\nx}(\pi^*)\leq\frac{4n}{4n+1}<\frac{4n+m+1}{4n+m+2}$, and so both $\x$ and $\nx$ prefer $\pi$.
Furthermore, their corresponding $\xt\in \Xt,\xf\in \Xf$ and $\xfp\in \Xfp$ all prefer~$\pi$ as well, so we have $\phi_{\{\x,\nx,\xt,\xf,\xfp\}}(\pi^*,\pi)=-5$, which is clearly enough to contradict the popularity of~$\pi^*$. 

We have shown that $j',j''\in S$. Hence, by \Cref{lem_fhg2:j'_j''}, it holds that $S = \{j,j',j''\}$.
\qed
\end{proof}

\begin{lemma}
\label{lem_fhg2:xy_separate}
Let $S$ be some coalition in $\pi^*$. Then $S$ cannot contain both an $X$- and a $Y$-agent.
\end{lemma}

\begin{proof}
Assume towards contradiction that $S$ contains both an $X$-agent and a $Y$-agent. 
If $S$ contains only $X$- and $Y$-agents, by \Cref{lem_fhg2:X_deviation_outcome,lem_fhg2:Y_deviation_outcome}, it is more popular to have all of them perform isolated deviations (here, since $j\notin S$, we used the strict version of \Cref{lem_fhg2:Y_deviation_outcome}, and the fact that there is at least one $Y$-agent in $S$). 
Hence, $S$ additionally contains an agent of another type.
Even more, by \Cref{lem_fhg2:Y'_Y'',lem_fhg2:j_j'_j'',lem_fhg2:l6_coalition,lem_fhg2:RL_everyone}, $S$ contains an $\Xt$- ,$\Xf$-, or $\Xfp$-agent. 
By \Cref{lem_fhg2:0_implies_singleton}, $S$ also contains one of her corresponding $X$-agents.
Let, therefore, $x\in \XX$ and $\alpha\in\{x,\neg x\}$, such that $\xl$ is a nonisolated $X$-agent in $S$. 
Fix a $Y$-agent $\y$ in $S$. 
Notice that we must have $|S|\leq 6$. 
Indeed, if $|S|\geq 7$, then we can either remove $\xt$, or remove $\{\xf,\xfp\}$ as a pair (this is well defined by \Cref{lem_fhg2:xf_xf'_together}).
Denote the resulting partition by $\pi$. 
If $\xt$ was removed from $S$ to create $\pi$, then only $\xl$, $\nxl$, and $\xt$ may prefer $\pi^*$ while the remaining members of $S$ prefer $\pi$. 
If $\xf$ and $\xfp$ were removed from $S$ to create $\pi$, then we have that $u_{\xf}(\pi)=u_{\xfp}(\pi)=\frac{1}{2}$, whereas $u_{\xf}(\pi^*)=u_{\xfp}(\pi^*)\leq \frac{3}{7}$. 
Hence, $\xf$ and $\xfp$ prefer $\pi$ over $\pi^*$, and so do the remaining agents in $S$ apart from $\xl$ and $\nxl$. 
Thus, $\pi$ is more popular than $\pi$, a contradiction. 
It remains to consider the case where $|S|\leq 6$.

Another useful observation is that, since \Cref{lem_fhg2:j_j'_j''} implies $j\notin S$, any isolated $X$ or $Y$-agent in $S$ has utility at most $\frac{|S|-1}{|S|}$. 
Therefore, since $|S|\leq 6$, every such agent performing an isolated deviation will end up with a higher utility (an $X$-agent will end up with utility at least $\frac{11}{12}$, and a $Y$-agent will end up with utility $1$).
This enables us to use the strict-inequality versions of \Cref{lem_fhg2:X_deviation_outcome,lem_fhg2:Y_deviation_outcome}. 
We divide into cases dependent on whether $\xf,\xfp\in S$ or $\xt\in S$.

\textit{Case 1:} Suppose that $\xf,\xfp\in S$. If either $\nxl\in S$ or $\xt\in S$, then 
$S$ contains $\xf$, $\xfp$, $\xl$, $\x$, $\y$, one of $\nxl$ or $\xt$, and possibly another agent $\ga$.
Hence, $\y$ is an isolated agent.
Also, $\ga$ must be an isolated $X$- or $Y$-agent.
Hence, having all isolated agents in $S$ perform an isolated deviation results in a more popular partition, a contradiction. 

If $\nxl,\xt\notin S$, then apart from $\xl,\xf,\xfp$ and $\y$, there is space for at most two other agents. 
If those two are corresponding $X$- and $\Xt$-agents, then we are back in Case 1; otherwise they can only be isolated $X$- or $Y$-agents. Once again, having all isolated agents in $S$ perform isolated deviations results in a more popular partition than $\pi^*$, a contradiction.

\textit{Case 2:} Suppose $\xt\in S$ while $\nxl\notin S$ and $\xf,\xfp\notin S$. 
Then $|S|\geq 4$ as otherwise removing $\xt$ from $S$ results in a more popular partition. 
Hence, $S$ contains $\xl,\xt,\y$ and possibly another agent $\ga$ who must be an isolated $X$- or $Y$-agent. 
Thus, the partition where all isolated agents in $S$ perform isolated deviations is more popular than $\pi^*$, a contradiction. 

\textit{Case 3:} Suppose that $\xt,\nxl\in S$ while $\xf,\xfp\notin S$. 
Apart from $\xl,\nxl,\xt$ and $\y$, there is room for at most two other agents. 
If those two are corresponding $X$- and $\Xt$-agents, then we are back in Case~2 from the perspective of this $X$-agent. 
Otherwise they are both isolated $X$- or $Y$-agents. 
Hence, removing $\xl$ and allocating her with $\{\xf,\xfp\}$ and having all isolated agents in $S$ perform isolated deviations results in a more popular partition, a contradiction.
\qed
\end{proof}

We can now fully determine the coalitions of $Y$-agents.

\begin{lemma}
\label{lem_fhg2:y_coalition}
Let $\y\in Y$. Then $\pi^*(\y)=\{\y,\yp,\ypp\}$.
\end{lemma}

\begin{proof}
If $\pi^*(\y)\neq\{\y,\yp,\ypp\}$, then by previous lemmas $\pi^*(\y)$ contains only $Y$-agents, and thus all agents in this coalition have utility less than $1$. 
Hence, it would be a Pareto improvement to have all agents in $\pi^*(\y)$ perform isolated deviations, a contradiction.
\qed
\end{proof}

Moreover, we can fully determine the coalitions of $X$-agents.

\begin{lemma}
\label{lem_fhg2:x_coalition}
Let $x\in\XX$. Then one of the following holds.
\begin{itemize}
    \item $\pi^*(\x)=\{\x,\xt\}$ and $\pi^*(\nx)=\{\nx,\xf,\xfp\}$.
    \item $\pi^*(\x)=\{\x,\xf,\xfp\}$ and $\pi^*(\nx)=\{\nx,\xt\}$.
\end{itemize}
\end{lemma}

\begin{proof}
By previous lemmas, $\pi^*(\x)$ and $\pi^*(\nx)$ can only contain $X$-, $\Xt$-, $\Xf$-, and $\Xfp$-agents. 
Our main claim is that $\x$ and $\nx$ are the only $X$-agents in $\pi^*(\x)$ and $\pi^*(\nx)$. 
Assume towards contradiction that there exists a coalition $S$ containing $X$-agents originating in two different variables of $\XX$.
By the same reasoning as in the proof of \Cref{lem_fhg2:xy_separate}, $S$ must contain some nonisolated $X$-agent, $|S|\leq 6$, and every isolated agent performing an isolated deviation will end up with a strictly higher utility, enabling us to use the strict-inequality versions of \Cref{lem_fhg2:X_deviation_outcome,lem_fhg2:Y_deviation_outcome}.
Without loss of generality, assume that $\x$ is nonisolated. Consider the following cases.

\textit{Case 1:} Suppose $\xt\in S$ while $\xf,\xfp\notin S$ (and by \Cref{lem_fhg2:xt_xf_xf'}, we have $\pi^*(\xf)=\pi^*(\xfp)=\{\xf,\xfp\}$). 
If $\nx\in S$, then the partition $\pi$ obtained from $\pi^*$ by extracting $\x$, $\nx$, and $\xt$ from $S$ and setting $\pi(\x)=\{\x,\xt\}$ and $\pi(\nx)=\{\nx,\xf,\xfp\}$ results in a more popular partition: 
Agents $\x$, $\nx$, $\xt$, $\xf$, and $\xfp$ prefer $\pi$ while at most three agents (the remaining agents in $S$) may prefer $\pi^*$, a contradiction.
If $\nx\notin S$, then we have $S\leq 4$, as otherwise removing $\xt$ from $S$ results in a more popular partition. Then, if $S$ is composed of two separate pairs of corresponding $X$- and $\Xt$-agents, it is a Pareto improvement to separate it into two pairs:
Clearly, the $\Xt$-agents prefer this, and furthermore, since $v_X>1$, it holds that $\frac {v_X}2 > \frac {v_X}4 + \frac 14 = \frac{v_X + 1}4$ and therefore the $X$-agents also prefer this.
Otherwise, the (at most two) remaining agents in $S$ are isolated $X$-agents, and having them perform isolated deviations results in a more popular partition, a contradiction.

\textit{Case 2:} Suppose $\xt,\xf,\xfp\in S$. 
If $S$ contains another $X$-agent and corresponding $\Xt$-agent, we are back in Case~1 from their perspective. 
Otherwise, $S$ contains an isolated $X$-agent, and having her perform an isolated deviation results in a more popular partition.

\textit{Case 3:} 
Suppose $\xt\notin S$ (and thus, no $\Xt$-agent is in $S$, as otherwise we are back in Case 1 or 2). 
Then since $\x$ is not isolated, $\xf,\xfp\in S$. 
If $\nx\in S$, then, since no $\Xt$-agent is in $S$, all other agents are isolated $X$-agents; 
having all of them perform isolated deviations results in a more popular partition.
So assume that $\nx\notin S$. 
Since $|S|\leq 6$, there is room for at most three other agents. 
If those three agents are another triple of $X$-,$\Xf$-, and $\Xfp$-agents, it is more popular to separate $S$ into two triplets (the $\Xf$-agents clearly prefer this, and, since $v_X>\frac{5}{9}$, it is easy to see that this is even true for the $X$-agents).
Otherwise, all other agents are isolated $X$-agents, and having them all perform isolated deviations results in a more popular partition, a contradiction.

Together, we have shown that $\x$ and $\nx$ are the only $X$-agents in $\pi^*(\x)$ and $\pi^*(\nx)$.
Hence, by \Cref{lem_fhg2:0_implies_singleton}, these coalitions cannot contain other $\Xt$-agents and since removing a pair of $\Xf$- and $\Xfp$-agents yields a Pareto improvement, there can also be no other such agents.
Hence, $\pi^*(\x),\pi^*(\nx)\subseteq \{\x,\nx,\xt,\xf,\xfp\}$.

Moreover, neither $\x$ nor $\nx$ is isolated, as then they would be in a singleton (the coalition $\{\x,\nx\}$ yields a utility of $0$ which is not possible), and could perform a strictly improving isolated deviation. 
Hence, we only have to exclude that $\nx\in \pi^*(\x)$ to end up in one of the cases described in the lemma.

Assume towards contradiction that $\nx\in \pi^*(\x)$.
If $\pi^*(\x) = \{\x,\nx,\xt,\xf,\xfp\}$, then forming $\{\x,\xt\}$ and $\{\nx,\xf,\xfp\}$ is more preferred by $\xt$ (since $\frac 12 > \frac 25$) and $\xf$ and $\xfp$ (since $\frac 23 > \frac 35$).
If $\pi^*(\x) = \{\x,\nx,\xf,\xfp\}$ or $\pi^*(\x) = \{\x,\nx,\xt\}$ then $u_{\x}(\pi^*) = u_{\nx}(\pi^*) \le \max\left\{\frac {18}{40} v_X, \frac 13 v_X\right\} = \frac {18}{40} v_X$.
Forming $\pi$ by dissolving $\pi^*(\x)$ into $\{\x,\xt\}$ and $\{\nx,\xf,\xfp\}$ yields $u_{\x}(\pi) = \frac 12 v_X$ and $u_{\nx}(\pi) = \frac {18}{30} v_X$, which is higher for $\x$ and $\nx$.
Moreover $\pi$ is preferred by the agent(s) among $\xt$, $\xf$, and $\xfp$ not contained in $\pi^*(\x)$.
Hence, $\pi$ is more popular, a contradiction.
\qed
\end{proof}

We have collected enough knowledge about the structure of the popular partition $\pi^*$ to extract a truth assignment.
First, we summarize what we know about $\pi^*$.

\begin{itemize}
	\item For every clause $c\in\CC$, we have the coalition $\{\ca,\la^6\}$ (\Cref{lem_fhg2:l6_coalition}).
Moreover, without loss of generality, we have the coalitions $\{\ra,\la^1,\la^2,\la^3\}$, $\{\la^4\}$, and $\{\la^5\}$ (\Cref{lem_fhg2:star_partition}).
	\item We have the coalition $\{j,j',j''\}$ (\Cref{lem_fhg2:j_j'_j''}).
	\item For $\y\in Y$, we have the coalition $\{\y,\yp,\ypp\}$ (\Cref{lem_fhg2:y_coalition}).
	\item For $x\in\XX$, we have that
\begin{itemize}
    \item $\pi^*(\x)=\{\x,\xt\}$ and $\pi^*(\nx)=\{\nx,\xf,\xfp\}$ or 
    \item $\pi^*(\x)=\{\x,\xf,\xfp\}$ and $\pi^*(\nx)=\{\nx,\xt\}$ (\Cref{lem_fhg2:x_coalition}).
\end{itemize}
\end{itemize}

This allows us to define the following truth assignment $\tau_{\XX}$ to the $\XX$ variables. 
For each $x\in \XX$, $x$ is assigned \Tru{} if and only if $\pi^*(\x)=\{\x,\xt\}$ (by \Cref{lem_fhg2:x_coalition}, this is a valid assignment). 
We claim that $\tau_{\XX}$ is a satisfying assignment to the $\QSAT$ instance, i.e., that $\psi(\tau_{\XX},\tau_{\YY}) = \Tru$ for all truth assignments $\tau_{\YY}$ to the $\YY$ variables.

Assume otherwise, namely that there exists a truth assignment $\tau_{\YY}$ to the $\YY$ variables such that $\psi(\tau_{\XX},\tau_{\YY})=\Fals$. 
We will show that this allows us to find a partition that is more popular than $\pi^*$. 
Therefore, consider the partition $\pi$ obtained from $\pi^*$ by extracting the following agents from their respective coalitions, and placing them all together in a new coalition $S$:
\begin{itemize}
    \item All $\xl\in X$ such that $\{\xl,\xt\}\in \pi^*$, for some $\xt\in\Xt$.
    \item All $\y\in Y$ such that the literal represented by $\y$ is assigned \Tru{} by $\tau_{\YY}$.
    \item All $\ca\in C$.
    \item Agent $j$.
\end{itemize}
Note that the new coalition consists of $2n+m+1$ agents.
Moreover, by definition of $\tau_{\XX}$, if $\tau_{\XX}$ assigns \Tru{} to $x$, then $S$ contains $\x$ and if $\tau_{\XX}$ assigns \Fals{} to $x$, then $S$ contains $\nx$.
In addition, for $y\in \YY$, $S$ contains $\yfix$ if $\tau_{\YY}$ assigns \Tru{} to $y$ and $S$ contains $\nyfix$ if $\tau_{\YY}$ assigns \Fals{} to $y$.

We compute the popularity margin between $\pi$ and $\pi^*$.
\begin{itemize}
	\item Fix $x\in\XX$, and let $\alpha\in\{x,\neg x\}$ such that $\xl\in X\cap S$. We have $u_{\xl}(\pi)=1>\frac{4n+m+1}{4n+m+2}= \frac {v_X}2 = u_{\xl}(\pi^*)$, and $u_{\xt}(\pi)=0<\frac{1}{2}=u_{\xt}(\pi^*)$. Hence $\phi_{\{\xl,\xt\}}(\pi^*,\pi)=0$.
	\item Let $\y\in Y\cap S$. The utilities of $\y,\yp$ and $\ypp$ are 1 in both partitions. Hence $\phi_{\{\y,\yp,\ypp\}}(\pi^*,\pi)=0$.
	\item Let $c\in\CC$. 
	We have that $u_{\la^6}(\pi)=0<\frac{3}{8}=u_{\la^6}(\pi^*)$, and so $\la^6$ prefers $\pi^*$ over $\pi$. 
	However, since $\psi(\tau_{\XX},\tau_{\YY})=\Fals$, we have that $c$ has at most two literals in $S$ assigned \Tru{} by $\tau_{\XX}$ and $\tau_{\YY}$.
	Hence, since the $X$- and $Y$-agents in $S$ correspond to the literals assigned \Tru{} by $\tau_{\XX}$ and $\tau_{\YY}$, there are at most two $X$- or $Y$-agents in $S$ to whom $\ca$ assigns value $0$ (to the other $X$- or $Y$-agents she assigns $v_C$). 
	Thus, we have $u_{\ca}(\pi)\geq \frac{(2n-2)v_C+m}{2n+m+1} > \frac{2n+1+m}{2n+m+1} = 1=u_{\ca}(\pi^*)$. 
	There, we use that $v_C = \frac{2n+1}{2n-2.5}$ and therefore $(2n-2)\frac{2n+1}{2n-2.5} > 2n+1$.
	Therefore, $\phi_{\{\ca,\la^6\}}(\pi^*,\pi)=0$. 
	\item Consider agents $j,j',j''$. Agents $j'$ and $j''$ obtain utility $1$ in both partitions, and are thus indifferent. As for $j$, we have $u_j(\pi)=\frac{2n+m}{2n+m+1}>\frac{2n+m-1}{2n+m}=u_j(\pi^*)$, and thus $\phi_{\{j,j',j''\}}(\pi^*,\pi)=-1$.
	\item All other agents are in the same coalition in $\pi$ and $\pi^*$ and they are therefore indifferent between the two partitions. 
\end{itemize}

	Altogether, we conclude that $\phi(\pi^*,\pi)=-1$, in contradiction to $\pi^*$ being a popular partition.
	Hence, $(\XX,\YY,\psi)$ is a Yes-instance of \QSAT.
\qed


\begin{thebibliography}{37}
\providecommand{\natexlab}[1]{#1}
\providecommand{\url}[1]{\texttt{#1}}
\expandafter\ifx\csname urlstyle\endcsname\relax
  \providecommand{\doi}[1]{doi: #1}\else
  \providecommand{\doi}{doi: \begingroup \urlstyle{rm}\Url}\fi

\bibitem[Arora and Barak(2009)]{ArBa09a}
S.~Arora and B.~Barak.
\newblock \emph{Computational Complexity: {A} Modern Approach}.
\newblock Cambridge University Press, 2009.

\bibitem[Aziz and Savani(2016)]{AzSa15a}
H.~Aziz and R.~Savani.
\newblock Hedonic games.
\newblock In F.~Brandt, V.~Conitzer, U.~Endriss, J.~Lang, and A.~D. Procaccia,
  editors, \emph{Handbook of Computational Social Choice}, chapter~15.
  Cambridge University Press, 2016.

\bibitem[Aziz et~al.(2013)Aziz, Brandt, and Seedig]{ABS11c}
H.~Aziz, F.~Brandt, and H.~G. Seedig.
\newblock Computing desirable partitions in additively separable hedonic games.
\newblock \emph{Artificial Intelligence}, 195:\penalty0 316--334, 2013.

\bibitem[Aziz et~al.(2019)Aziz, Brandl, Brandt, Harrenstein, Olsen, and
  Peters]{ABB+17a}
H.~Aziz, F.~Brandl, F.~Brandt, P.~Harrenstein, M.~Olsen, and D.~Peters.
\newblock Fractional hedonic games.
\newblock \emph{ACM Transactions on Economics and Computation}, 7\penalty0
  (2):\penalty0 1--29, 2019.

\bibitem[Banerjee et~al.(2001)Banerjee, Konishi, and S{\"o}nmez]{BKS01a}
S.~Banerjee, H.~Konishi, and T.~S{\"o}nmez.
\newblock Core in a simple coalition formation game.
\newblock \emph{Social Choice and Welfare}, 18:\penalty0 135--153, 2001.

\bibitem[Bir{\'o} et~al.(2010)Bir{\'o}, Irving, and Manlove]{BIM10a}
P.~Bir{\'o}, R.~W. Irving, and D.~F. Manlove.
\newblock Popular matchings in the marriage and roommates problems.
\newblock In \emph{Proceedings of the 7th Italian Conference on Algorithms and
  Complexity (CIAC)}, pages 97--108, 2010.

\bibitem[Bogomolnaia and Jackson(2002)]{BoJa02a}
A.~Bogomolnaia and M.~O. Jackson.
\newblock The stability of hedonic coalition structures.
\newblock \emph{Games and Economic Behavior}, 38\penalty0 (2):\penalty0
  201--230, 2002.

\bibitem[Brandl et~al.(2015)Brandl, Brandt, and Strobel]{BBS14a}
F.~Brandl, F.~Brandt, and M.~Strobel.
\newblock Fractional hedonic games: {I}ndividual and group stability.
\newblock In \emph{Proceedings of the 14th International Conference on
  Autonomous Agents and Multiagent Systems (AAMAS)}, pages 1219--1227, 2015.

\bibitem[Brandt and Bullinger(2022)]{BrBu20a}
F.~Brandt and M.~Bullinger.
\newblock Finding and recognizing popular coalition structures.
\newblock \emph{Journal of Artificial Intelligence Research}, 74:\penalty0
  569--626, 2022.

\bibitem[Brandt et~al.(2023)Brandt, Bullinger, and Wilczynski]{BBW21b}
F.~Brandt, M.~Bullinger, and A.~Wilczynski.
\newblock Reaching individually stable coalition structures.
\newblock \emph{ACM Transactions on Economics and Computation}, 11\penalty0
  (1--2):\penalty0 4:1--65, 2023.

\bibitem[Brandt et~al.(2024)Brandt, Bullinger, and Tappe]{BBT23a}
F.~Brandt, M.~Bullinger, and L.~Tappe.
\newblock Stability based on single-agent deviations in additively separable
  hedonic games.
\newblock \emph{Artificial Intelligence}, 334, 2024.

\bibitem[Bullinger and Kraiczy(2024)]{BuKr24a}
M.~Bullinger and S.~Kraiczy.
\newblock Stability in random hedonic games.
\newblock In \emph{Proceedings of the 25th ACM Conference on Economics and
  Computation (ACM-EC)}, 2024.

\bibitem[Bullinger et~al.(2024)Bullinger, Elkind, and Rothe]{BER24a}
M.~Bullinger, E.~Elkind, and J.~Rothe.
\newblock Cooperative game theory.
\newblock In J.~Rothe, editor, \emph{Economics and Computation: An Introduction
  to Algorithmic Game Theory, Computational Social Choice, and Fair Division},
  chapter~3, pages 139--229. Springer, 2024.

\bibitem[Cechl{\'a}rov{\'a} and Romero-Medina(2001)]{CeRo01a}
K.~Cechl{\'a}rov{\'a} and A.~Romero-Medina.
\newblock Stability in coalition formation games.
\newblock \emph{International Journal of Game Theory}, 29:\penalty0 487--494,
  2001.

\bibitem[Cohen-Addad et~al.(2022)Cohen-Addad, Lattanzi, Maggiori, and
  Parotsidis]{CLMP22a}
V.~Cohen-Addad, S.~Lattanzi, A.~Maggiori, and N.~Parotsidis.
\newblock Online and consistent correlation clustering.
\newblock In \emph{Proceedings of the 39th International Conference on Machine
  Learning (ICML)}, pages 4157--4179, 2022.

\bibitem[Condorcet(1785)]{Cond85a}
M.~Condorcet.
\newblock \emph{Essai sur l'application de l'analyse {\`a} la probabilit{\'e}
  des d{\'e}cisions rendues {\`a} la pluralit{\'e} des voix}.
\newblock Imprimerie Royale, 1785.
\newblock Facsimile published in 1972 by Chelsea Publishing Company, New York.

\bibitem[Cseh and Peters(2022)]{CsPe21a}
{\'A}.~Cseh and J.~Peters.
\newblock Three-dimensional popular matching with cyclic preferences.
\newblock In \emph{Proceedings of the 21st International Conference on
  Autonomous Agents and Multiagent Systems (AAMAS)}, pages 309--317, 2022.

\bibitem[Dimitrov et~al.(2006)Dimitrov, Borm, Hendrickx, and Sung]{DBHS06a}
D.~Dimitrov, P.~Borm, R.~Hendrickx, and S.~C. Sung.
\newblock Simple priorities and core stability in hedonic games.
\newblock \emph{Social Choice and Welfare}, 26\penalty0 (2):\penalty0 421--433,
  2006.

\bibitem[Dr{\`e}ze and Greenberg(1980)]{DrGr80a}
J.~H. Dr{\`e}ze and J.~Greenberg.
\newblock Hedonic coalitions: Optimality and stability.
\newblock \emph{Econometrica}, 48\penalty0 (4):\penalty0 987--1003, 1980.

\bibitem[Faenza et~al.(2019)Faenza, Kavitha, Power, and Zhang]{FKPZ19a}
Y.~Faenza, T.~Kavitha, V.~Power, and X.~Zhang.
\newblock Popular matchings and limits to tractability.
\newblock In \emph{Proceedings of the 30th Annual ACM-SIAM Symposium on
  Discrete Algorithms (SODA)}, pages 2790--2809, 2019.

\bibitem[Fanelli et~al.(2021)Fanelli, Monaco, and Moscardelli]{FMM21a}
A.~Fanelli, G.~Monaco, and L.~Moscardelli.
\newblock Relaxed core stability in fractional hedonic games.
\newblock In \emph{Proceedings of the 30th International Joint Conference on
  Artificial Intelligence (IJCAI)}, pages 182--188, 2021.

\bibitem[Fioravanti et~al.(2023)Fioravanti, Flammini, Kodric, and
  Varricchio]{FFKV23a}
S.~Fioravanti, M.~Flammini, B.~Kodric, and G.~Varricchio.
\newblock $\varepsilon$-fractional core stability in hedonic games.
\newblock In \emph{Proceedings of the 37th Annual Conference on Neural
  Information Processing Systems (NeurIPS)}, 2023.

\bibitem[Gairing and Savani(2019)]{GaSa19a}
M.~Gairing and R.~Savani.
\newblock Computing stable outcomes in symmetric additively separable hedonic
  games.
\newblock \emph{Mathematics of Operations Research}, 44\penalty0 (3):\penalty0
  1101--1121, 2019.

\bibitem[Gale and Shapley(1962)]{GaSh62a}
D.~Gale and L.~S. Shapley.
\newblock College admissions and the stability of marriage.
\newblock \emph{The American Mathematical Monthly}, 69\penalty0 (1):\penalty0
  9--15, 1962.

\bibitem[G{\"a}rdenfors(1975)]{Gard75a}
P.~G{\"a}rdenfors.
\newblock Match making: {A}ssignments based on bilateral preferences.
\newblock \emph{Behavioral Science}, 20\penalty0 (3):\penalty0 166--173, 1975.

\bibitem[Gupta et~al.(2021)Gupta, Misra, Saurabh, and Zehavi]{GMSZ21a}
S.~Gupta, P.~Misra, S.~Saurabh, and M.~Zehavi.
\newblock Popular matching in roommates setting is np-hard.
\newblock \emph{ACM Transactions on Computation Theory}, 13\penalty0
  (2):\penalty0 9:1--9:20, 2021.

\bibitem[Kavitha et~al.(2011)Kavitha, Mestre, and Nasre]{KMN11a}
T.~Kavitha, J.~Mestre, and M.~Nasre.
\newblock Popular mixed matchings.
\newblock \emph{Theoretical Computer Science}, 412\penalty0 (24):\penalty0
  2679--2690, 2011.

\bibitem[Kerkmann and Rothe(2020)]{KeRo20a}
A.~M. Kerkmann and J.~Rothe.
\newblock Altruism in coalition formation games.
\newblock In \emph{Proceedings of the 29th International Joint Conference on
  Artificial Intelligence (IJCAI)}, pages 461--467, 2020.

\bibitem[Kerkmann et~al.(2020)Kerkmann, Lang, Rey, Rothe, Schadrack, and
  Schend]{KLR+20a}
A.~M. Kerkmann, J.~Lang, A.~Rey, J.~Rothe, H.~Schadrack, and L.~Schend.
\newblock Hedonic games with ordinal preferences and thresholds.
\newblock \emph{Journal of Artificial Intelligence Research}, 67:\penalty0
  705--756, 2020.

\bibitem[Newman(2004)]{Newm04a}
M.~E.~J. Newman.
\newblock Detecting community structure in networks.
\newblock \emph{The European Physical Journal B - Condensed Matter and Complex
  Systems}, 38\penalty0 (2):\penalty0 321--330, 2004.

\bibitem[Olsen(2012)]{Olse12a}
M.~Olsen.
\newblock On defining and computing communities.
\newblock In \emph{Proceedings of the 18th Computing: The Australasian Theory
  Symposium (CATS)}, volume 128 of \emph{Conferences in Research and Practice
  in Information Technology (CRPIT)}, pages 97--102, 2012.

\bibitem[Papadimitriou(1994)]{Papa94a}
C.~H. Papadimitriou.
\newblock \emph{Computational Complexity}.
\newblock Addison-Wesley, 1994.

\bibitem[Peters(2017)]{Pete17b}
D.~Peters.
\newblock Precise complexity of the core in dichotomous and additive hedonic
  games.
\newblock In \emph{Proceedings of the 5th International Conference on
  Algorithmic Decision Theory (ADT)}, pages 214--227, 2017.

\bibitem[Ray(2007)]{Ray07a}
D.~Ray.
\newblock \emph{A Game-Theoretic Perspective on Coalition Formation}.
\newblock Oxford University Press, 2007.

\bibitem[Stockmeyer(1977)]{St77}
L.~J. Stockmeyer.
\newblock The polynomial-time hierarchy.
\newblock \emph{Theoretical Computer Science}, 3\penalty0 (1):\penalty0 1--22,
  1977.

\bibitem[Sung and Dimitrov(2010)]{SuDi10a}
S.~C. Sung and D.~Dimitrov.
\newblock Computational complexity in additive hedonic games.
\newblock \emph{European Journal of Operational Research}, 203\penalty0
  (3):\penalty0 635--639, 2010.

\bibitem[Woeginger(2013)]{Woeg13a}
G.~J. Woeginger.
\newblock A hardness result for core stability in additive hedonic games.
\newblock \emph{Mathematical Social Sciences}, 65\penalty0 (2):\penalty0
  101--104, 2013.

\end{thebibliography}
\end{document}